\documentclass[11pt]{article}
\usepackage{amssymb}
\usepackage{amsfonts}
\usepackage{amsmath}
\usepackage[a4paper, top=1.05in, bottom = 1.05in, left = 1in, right = 1in]{geometry}
\usepackage[onehalfspacing]{setspace}
\usepackage[round]{natbib}
\usepackage[pdftex]{graphicx}
\usepackage{epstopdf}
\usepackage{rotating,booktabs}
\usepackage{subfigure}
\usepackage{color}
\usepackage{algorithm}
\usepackage{algorithmic}
\usepackage{multirow}
\newtheorem{theorem}{Theorem}
\newtheorem{assumption}{Assumption}

\newtheorem{definition}{Definition}

\newtheorem{proposition}{Proposition}

\newenvironment{proof}[1][Proof]{\textbf{#1.} }{\ \rule{0.5em}{0.5em}}

\begin{document}
\title{\textbf{Doubly Robust Estimation of Direct and Indirect Quantile Treatment Effects with Machine Learning}}
\author{Yu-Chin Hsu\thanks{Institute of Economics, Academia Sinica, 128, Section 2, Academia Road, Nankang, Taipei 115, Taiwan. E-mail: \texttt{ychsu@econ.sinica.edu.tw}.}\\Academia Sinica \and Martin Huber\thanks{University of Fribourg, Department of Economics, Bd. de P\'{e}rolles 90, 1700 Fribourg, Switzerland. E-mail: \texttt{martin.huber@unifr.ch}.}\\University of Fribourg \and Yu-Min Yen\thanks{Department of International Business, National Chengchi University, 64, Section 2, Zhi-nan Road, Wenshan, Taipei 116, Taiwan. E-mail: \texttt{yyu\_min@nccu.edu.tw}.} \\ National\\Chengchi University}
\date{\today}

\maketitle
\begin{abstract} We suggest double/debiased machine learning estimators of direct and indirect quantile treatment effects under a selection-on-observables assumption. This permits disentangling the causal effect of a binary treatment at a specific outcome rank into an indirect component that operates through an intermediate variable called mediator and an (unmediated) direct
impact. The proposed method is based on the efficient score functions of the cumulative distribution functions of potential outcomes, which are robust to certain misspecifications of the nuisance parameters, i.e., the outcome, treatment, and mediator models. We estimate these nuisance parameters by machine learning and use cross-fitting to reduce overfitting bias in the estimation of direct and indirect quantile treatment effects. We establish uniform consistency and asymptotic normality of our effect estimators. We also propose a multiplier bootstrap for statistical inference and show the validity of the multiplier bootstrap. Finally, we investigate the finite sample performance of our method in a simulation study and apply it to empirical data from the National Job Corp Study to assess the direct and indirect earnings effects of training.
\end{abstract}

\clearpage
\section{Introduction}
	Causal mediation analysis aims at understanding the mechanisms through which a treatment affects an outcome of interest. It disentangles the treatment effect into an indirect effect, which operates through a mediator, and a direct effect, which captures any causal effect not operating through the mediator. Such a decomposition of the total treatment effect permits learning the drivers of the effect, which may be helpful for improving the design of a policy or intervention. Causal mediation analysis typically focuses on the estimation of average indirect and direct effects, which may mask interesting effect heterogeneity across individuals. For this reason, several contributions focusing on total (rather than direct and indirect) effects consider quantile treatment effects (QTE) instead of average treatment effects (ATE). The QTE corresponds to the difference between the potential outcomes with and without treatment at a specific rank of the potential outcome distributions, but has so far received little attention in the causal mediation literature. 
	
	The main contribution of this paper is to propose  doubly robust/debiased machine learning (DML) estimators of the direct and indirect QTE under a selection-on-observables (or sequential ignorability) assumption, implying that the treatment and the mediator are as good as random when controlling for observed covariates. The method computes the quantile of a potential outcome by inverting an DML estimate of its  cumulative distributional function (c.d.f.). This approach makes use of the efficient score function of the c.d.f., into which models for the outcome, treatment, and mediator enter as plug-in or nuisance parameters. Relying on the efficient score function makes treatment effect estimation robust, i.e., first-order insensitive to (local) misspecifications of the nuisance parameters, a property known as \citet{Neyman1959}-orthogonality. This permits estimating the nuisance parameters by machine learning (which generally introduces regularization bias) and still obtains root-n-consistent treatment effect estimators, given that certain regularity conditions hold.  In addition, cross-fitting is applied to mitigate overfitting bias. Cross-fitting consists of estimating the nuisance parameter models and treatment effects in different subsets of the data and swapping the roles of the data to exploit the entire sample for treatment effect estimation, see \citet{CCDDHNR_2018}. We then establish uniform consistency and asymptotic normality of the effect estimators.
 
	For conducting statistical inference, we propose a multiplier bootstrap procedure and show the validity of the multiplier bootstrap. We also provide a simulation study to investigate the finite sample performance of our method. Finally, we apply our method to empirical data from the Job Corps Study to analyse the direct and indirect QTE of participation in a training program on earnings when considering general health as a mediator. The results point to positive direct effects of training across a large range of the earnings quantiles, while the indirect effects are generally close to zero and mostly statistically insignificant. 

 \begin{figure}[h]
	\centering
	\includegraphics[height=8cm, width = 11cm]{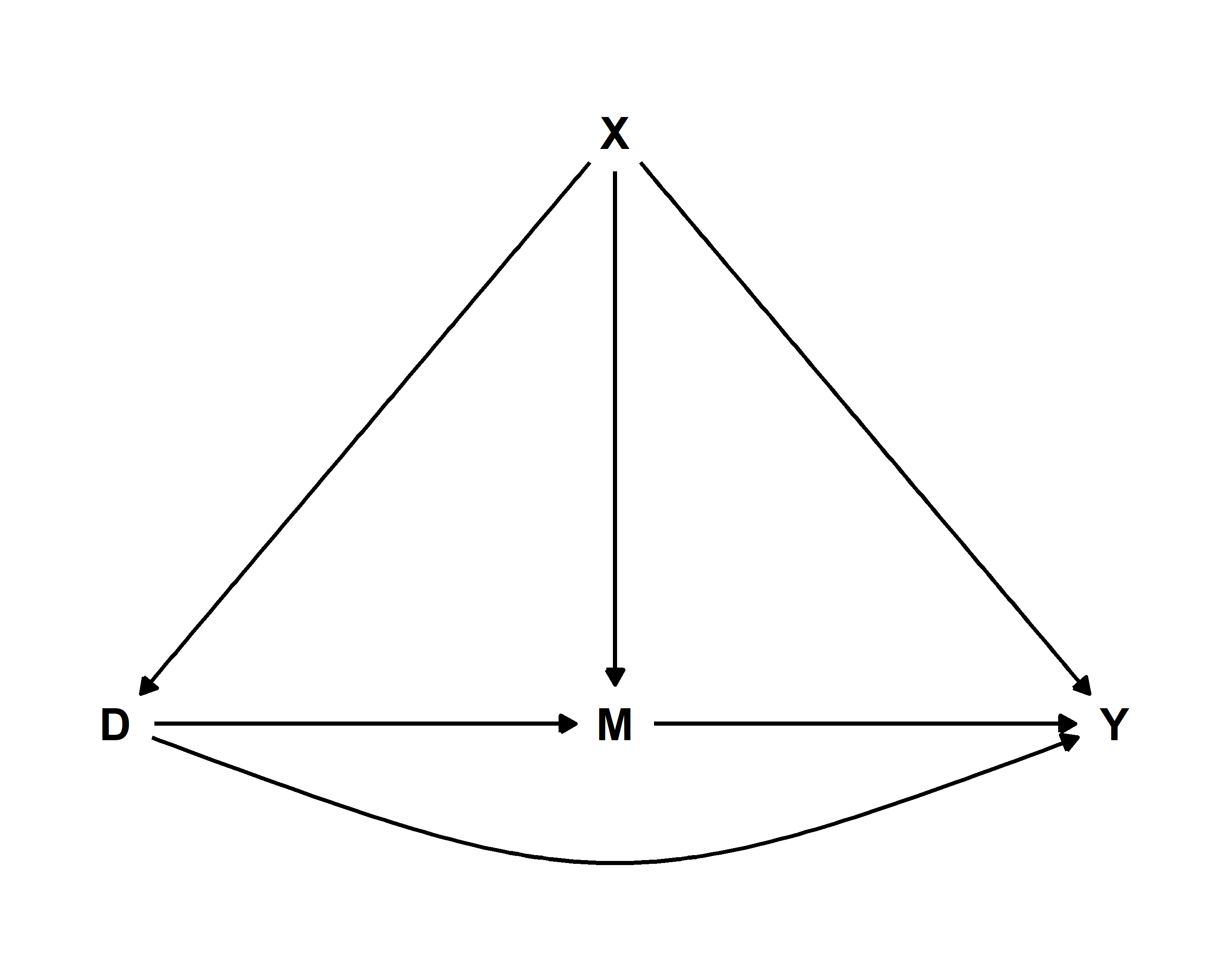}
	\caption{DAG illustrating causal links between outcome $Y$, treatment $D$, mediator $M$ 
		and covariates $X$.}
	\label{figure1}
\end{figure}

	To more formally discuss the direct and indirect effects of interest, let $Y$ denote the outcome of interest, $D$ the binary treatment, $M$ the mediator, and $X$ a vector of pre-treatment covariates. Following \citet{Pearl00}, we may represent causal relationships between $\left(Y,D,M,X\right)$
	by means of a directed acyclic graph (DAG), as provided in Figure \ref{figure1}. The causal arrows in the DAG imply that  $\left(D,M,X\right)$ may affect $Y$, $\left(D,X\right)$ may affect  $M$,  and $X$ may affect $D$. We can therefore define the  outcome as a function of the treatment and the mediator, $Y=Y\left(D, M\right)$,
	and the mediator as a function of the treatment, $M=M\left(D\right)$, while being agnostic about $X$. Furthermore, we make use of the potential outcome notation advocated by \citet{Neyman23} and \citet{Rubin74} to denote by $Y\left(d,m\right)$ the potential outcome if $D$ were
	set to a specific value $d\in\left\{ 0,1\right\}$ and $M$ were set to some value $m$ in the support of the mediator, while $M\left(d\right)$ denotes the potential mediator for $D=d$. 
 Accordingly, $Y\left(d, M\left(d\right)\right)$ is the potential outcome if $D$ were set to $d$, implying that the mediator is not forced to take a specific value $m$, but corresponds to its potential value under $D=d$. Depending on the actual treatment and mediator values of an observation, $Y\left(d, M\left(d\right)\right)$, $Y\left(d,m\right)$, $M\left(d\right)$ is either observed or counterfactual. Furthermore, the potential outcome $Y\left(d, M\left(1-d\right)\right)$  is inherently counterfactual, as no observation can be observed in the opposite treatment states $d$ and $1-d$ at the same time. 
 
Armed with this notation, we define the causal parameters of interest. The natural direct effect (NDE), which is for instance considered in \citet{RoGr92}, \citet{Pearl01} and \citet{TS_2012}, is based on a comparison of the potential outcomes when varying the treatment, but keeping the potential mediator fixed at treatment value $D=d$: $Y\left(1,M\left(d\right)\right)-Y\left(0,M\left(d\right)\right)$. The natural indirect effect (NIE) is based on a comparison of the potential outcomes when fixing the treatment to $D=d$, but varying the potential mediator according to the values it takes under treatment and non-treatment: $Y\left(d,M\left(1\right)\right)-Y\left(d,M\left(0\right)\right)$. It is worth noting that the NDE and the NIE defined upon opposite treatment states $d$ and $1-d$ sum up to the total effect (TE): $Y\left(1,M\left(1\right)\right)-Y\left(0,M\left(0\right)\right)$. 
	
	Previous methodological research on causal mediation predominantly focused on the  estimation of averages of the aforementioned NDE, NIE and TE or of averages of related path-wise causal effects \citep{IKY_2010,TS_2012,HHL_2019,FHLLS_2022,Zhou_2022}. We complement this literature by suggesting a method for estimating natural direct and indirect QTEs, which permits assessing the effects across the entire distribution of potential outcomes. The estimation of the total (rather than the direct or indirect) QTE has already been studied in multiple contributions \citep{AAI_2002,CH_2005,Firpo_2007,DH_2014,BCFH_2017,ALZ_2022,HLL_2022}. Among the few studies considering QTEs in causal mediation is \citet{BVSC_2017}, suggesting a two-stage quantile regression estimation to estimate the controlled direct QTE, i.e., $Y\left(1,m\right)-Y\left(0,m\right)$ at a specific rank, as well as a particular indirect QTE. The latter is based on first estimating the mediator at a specific rank and then including it in a quantile regression of the outcome, which generally differs from the natural indirect QTE considered in this paper.  Furthermore, our approach is nonparametric and relies on results on semiparametric efficiency, very much in contrast to the parametric approach of \citet{BVSC_2017}. \citet{HSS_2022} adapted the Changes-in-Changes (CiC) approach of \citet{AtheyImbens06} to estimate direct and indirect QTEs in subgroups defined in terms of how the mediator reacts to (or complies with) the treatment. The NDE and NIE investigated here differ from such subgroup-specific causal parameters and furthermore, our identification strategy relies on a selection-on-observables (or sequential ignorability) assumption rather than CiC.  

 The remainder of this study is organized as follows. Section \ref{Methodology} introduces the natural direct and indirect QTE, the identifying assumptions, the effect estimators based on double/debiased machine learning, and the multiplier bootstrap procedure for inference. Section \ref{theores} gives the theoretical results on the asymptotic behavior of our methods. Section \ref{sim} presents a simulation study that investigates the finite sample properties of our method. Section \ref{appl} provides an empirical application to data from the National Job Corps Study to assess the direct earnings effects of training across the earnings distribution, as well as the indirect effects operating via general health.  Section \ref{conclusion} concludes. 

\section{Methodology}\label{Methodology}

\subsection{Causal effects and Identifying Assumptions}

	To define the direct and indirect QTEs of interest, let $Q_{Z}\left(\tau\right):=\inf\left\{ q\in\mathbb{R}:P\left(Z\leq q\right)\geq\tau\right\} $
	denote the $\tau$th-quantile of a random variable $Z$, where $\tau\in(0,1)$. Furthermore, let $Q_{Z|V}\left(\tau\right):=\inf\left\{ q\in\mathbb{R}:P\left(Z\leq q|V\right)\geq\tau\right\}$
	denote the $\tau$th-quantile of $Z$ conditional on another random variable (or a random vector) $V$, where $\tau\in(0,1)$. Let $F_{Z}\left(z\right)$ and $f_{Z}\left(z\right)$ denote
	cumulative distribution function (c.d.f.) and probability density
	or probability mass function (p.d.f. or p.m.f.) of $Z$ at $z$, and
	$F_{Z|V}\left(z|v\right)$ and $f_{Z|V}\left(z|v\right)$ denote the
	c.d.f. and p.d.f. (or p.m.f.) of $Z$ at $z$ conditional on $V=v$. We define the natural direct quantile treatment effect (NDQTE) at the $\tau$th-quantile as:
	\begin{equation}
		\text{NDQTE}\left(\tau\right):=Q_{Y\left(1,M\left(0\right)\right)}\left(\tau\right)-Q_{Y\left(0,M\left(0\right)\right)}\left(\tau\right),\label{NDQTE}
	\end{equation}
	and the natural indirect quantile treatment effect (NIQTE) at the $\tau$th-quantile as: 
	\begin{equation}
		\text{NIQTE}\left(\tau\right):=Q_{Y\left(1,M\left(1\right)\right)}\left(\tau\right)-Q_{Y\left(1,M\left(0\right)\right)}\left(\tau\right).\label{NIQTE}
	\end{equation}
	The NDQTE in equation (\ref{NDQTE}) corresponds to the direct effect of the treatment when fixing the mediator at its potential value under non-treatment, $M(0)$. Alternatively, we may consider the NDQTE when conditioning on the potential mediator under treatment, $M(1)$: 
	\begin{equation}
		\text{NDQTE}^{\prime}\left(\tau\right):=Q_{Y\left(1,M\left(1\right)\right)}\left(\tau\right)-Q_{Y\left(0,M\left(1\right)\right)}\left(\tau\right).\label{NDQTE1}
	\end{equation} Likewise, the NIQTE in equation (\ref{NDQTE1}) is the indirect effect when varying the potential mediators but keeping the treatment fixed ad $D=1$, but we may also consider the indirect effect conditional on $D=0$:
	\begin{equation}\text{NIQTE}^{\prime}\left(\tau\right):=Q_{Y\left(0,M\left(1\right)\right)}\left(\tau\right)-Q_{Y\left(0,M\left(0\right)\right)}\left(\tau\right).\label{NIQTE1}
	\end{equation} If the effects in expressions (\ref{NDQTE}) and (\ref{NDQTE1}) (or (\ref{NIQTE}) and (\ref{NIQTE1})) are different, then this implies effect heterogeneity due to interaction effects between the treatment and the mediator. The sum of NDQTE (NDQTE') and NIQTE (NIQTE') yields the total quantile treatment effect
	(TQTE) at the $\tau$th-quantile, which includes all causal mechanisms through which the treatment affects the outcome: 
	\begin{align}
		\text{TQTE}\left(\tau\right)&=\text{NDQTE}\left(\tau\right)+\text{NIQTE}\left(\tau\right)=\text{NDQTE}^{\prime}\left(\tau\right)+\text{NIQTE}^{\prime}\left(\tau\right)\nonumber\\
		&=Q_{Y\left(1,M\left(1\right)\right)}\left(\tau\right)-Q_{Y\left(0,M\left(0\right)\right)}\left(\tau\right).\label{TQTE}
	\end{align}We aim at estimating the quantile treatment effects (\ref{NDQTE}) to (\ref{TQTE}).\footnote{We do not consider the controlled direct quantile treatment effect (CDQTE) at the $\tau$-quantile, $\text{CDQTE}\left(\tau\right):=Q_{Y\left(1,m\right)}\left(\tau\right)-Q_{Y\left(0,m\right)}\left(\tau\right)$, which can be identified under less stringent assumptions than required for the identification of natural effects. 
		} To this end, we first need to estimate the $\tau$th-quantile of the relevant potential outcomes, by inverting estimates of the corresponding c.d.f.'s at the $\tau$th-quantile. Let $F_{Y\left(d,M\left(d^{\prime}\right)\right)}\left(a\right)$ denote the c.d.f. of the potential outcome $Y\left(d,M\left(d^{\prime}\right)\right)$
	at value $a$. To identify $F_{Y\left(d,M\left(d^{\prime}\right)\right)}\left(a\right)$ in the data, we impose the following assumptions.
	\begin{assumption}
		\item[1.] 
		For any observation and $d\in\left\{ 0,1\right\} $ as well as $m$ in the support of $M$, $M=M\left(d\right)$ if $D=d$, and $Y=Y\left(d,m\right)$ if $D=d$ and $M=m$.
		\item[2.] 
		$\left(Y\left(d,m\right),M\left(d^{\prime}\right)\right)\perp D|X=x$
		for $\left(d,d^{\prime}\right)\in\left\{ 1,0\right\}^{2} $ and $m,x$ in the support
		of $(M,X)$. 
		\item[3.] 
		$Y\left(d,m\right)\perp M\left(d^{\prime}\right)|D=d^{\prime},X=x$
		for $\left(d,d^{\prime}\right)\in\left\{ 1,0\right\}^{2}$ and $m,x$ in the support of $\left(M,X\right)$. 
		\item[4.] 
		$f_{D|M,X}\left(d|m,x\right)>0$ for any $d\in\left\{ 1,0\right\} $ and $m,x$ in the support of $(M,X)$.
	\end{assumption}
	Assumption 1.1 implies the stable unit treatment value assumption (SUTVA), see \citet{Cox58} and \citet{Rubin80}, stating the potential mediators and potential outcomes are only a function of an individual's own treatment and mediator states, respectively, which are well defined (ruling out multiple treatment or mediator versions). Assumptions 1.2 and 1.3 are \textit{sequential ignorability or selection-on-observables conditions} \citep{IKY_2010} for causal mediation analysis. Assumption 1.2 states that conditional on $X$, the treatment variable $D$ is independent of the potential outcome $Y(d,m)$ and the potential mediator $M(d^{\prime})$. This assumption also implies that $Y\left(d,m\right)\perp D|M\left(d^{\prime}\right)=m^{\prime},X=x$. Assumption 1.3 requires that $Y(d,m)$ and $M(d^{\prime})$ are independent, too, conditional on $X$ and $D$. Even if treatment  $D$ were random, this would not suffice to identify direct and indirect effects and for this reason, we need to impose an identifying assumption like Assumption 1.3 to tackle the endogeneity of the mediator. Assumption 1.4 is a common support condition, which says that the treatment is not deterministic in covariates $X$ and mediator $M$ such that for each covariate-mediator combination in the population, both treated and non-treated subjects exist. 
	
	Under Assumptions 1.1 to 1.4, we obtain the following identification result.
	\begin{proposition} Under Assumptions 1.1 to 1.4, 
		\begin{equation}
	F_{Y\left(d,M\left(d^{\prime}\right)\right)}\left(a\right)  =  \int g_{d,d^{\prime},a}\left(x\right)f_{X}\left(x\right)dx\label{key_idf_ym}
		\end{equation}
		where $\left(d,d^{\prime}\right)\in \{0,1\}^{2}$, $a\in \mathcal{A}$ where $\mathcal{A}$ is a countable subset of $\mathbb{R}$, and 
		\begin{align*}
			g_{d,d^{\prime},a}\left(x\right) & =  \int F_{Y|D,M,X}\left(a|d,m,x\right)f_{M|D,X}\left(m|d^{\prime},x\right)dm\\
			& = E\left[F_{Y|D,M,X}\left(a|d,M,X\right)|d^{\prime},X=x\right].
		\end{align*}
	\end{proposition}
	
	The proof of Proposition 1 is provided in the appendix. Under Proposition 1, we may estimate $F_{Y\left(d,M\left(d^{\prime}\right)\right)}\left(a\right)$ based on plug-in estimation of the nuisance parameters $F_{Y|D,M,X}\left(a|d,m,X\right)$ and $f_{M|D,X}\left(m|d^{\prime},X\right)$: 
	\begin{equation}
		\hat{\theta}_{d,d^{\prime},a}^{YM}=\frac{1}{n}\sum_{i=1}^{n}\hat{g}_{d,d^{\prime},a}\left(X_{i}\right),\label{estimator_ym}
	\end{equation}
	where 
	\begin{equation}
		\hat{g}_{d,d^{\prime},a}\left(X_{i}\right)=\int\hat{F}_{Y|D,M,X}\left(a|d,m,X_{i}\right)\hat{f}_{M_{i}|D_{i},X_{i}}\left(m|d^{\prime},X_{i}\right)dm, \label{g_hat}
	\end{equation}
	and $\hat{F}_{Y|D,M,X}\left(a|d,m,X_{i}\right)$ and $\hat{f}_{M|D,X}\left(m|d^{\prime},X_{i}\right)$
	are estimates of $F_{Y|D,M,X}\left(a|d,m,X\right)$ and $f_{M|D,X}\left(m|d^{\prime},X_{i}\right)$. 
	If $M$ is a continuous variable, we may avoid estimating the conditional density $f_{M|D,X}\left(m|d^{\prime},X\right)$, and use the following alternative estimator for estimating $F_{Y\left(d,M\left(d^{\prime}\right)\right)}\left(a\right)$:
	\begin{equation}
		\hat{\theta}_{d,d^{\prime},a}^{RI}=\frac{1}{n}\sum_{i=1}^{n}\hat{E}\left[F_{Y|D,M,X}\left(a|d,M_{i},X_{i}\right)|d^{\prime},X_{i}\right],\label{estimator_yme}
	\end{equation}
	where $\hat{E}\left[F_{Y|D,M,X}\left(a|d,M_{i},X_{i}\right)|d^{\prime},X_{i}\right]$
	is an estimate of $E\left[F_{Y|D,M,X}\left(a|d,M_{i},X_{i}\right)|d^{\prime},X_{i}\right]$. For example, it might be based on a ``regression-imputation'' \citep{Zhou_2022}, corresponding to the fitted value of a linear regression of $F_{Y|D,M,X}\left(a|d,X_{i},M_{i}\right)$
	on $D_{i}$ and $X_{i}$ at $\left(d^{\prime},X_{i}\right)$. 
	
	The quality of estimators (\ref{estimator_ym}) and (\ref{estimator_yme}) crucially depends on the accuracy of nuisance parameter estimation. If the number of pretreatment covariates $X$ is small (low dimensional $X$) and the functional forms of the nuisance parameters are known, parametric methods can provide high-quality estimations on the nuisance parameters. In contrast, if $X$ is high dimensional and/or the nuisance parameters have complex forms, machine learning may be the preferred choice of estimation. However, applying ML directly to estimate expressions (\ref{estimator_ym}) or (\ref{estimator_yme}) may result in non-negligible bias induced by regularization and/or overfitting \citep{CCDDHNR_2018}. Causal machine learning algorithms aim at avoiding such biases by applying ML estimation when making use of Neyman-orthogonal moment conditions, which imply that the estimation of causal parameters is first order insensitive to (regularization) bias in the nuisance parameters, and of cross-fitting, which avoids overfitting.  One of these causal algorithms is double/debiased machine learning (DML) \citep{CCDDHNR_2018}, which has been previously adapted to the estimation of average effects in causal mediation analysis \citep{FHLLS_2022}, while this study extends it to the estimation of direct and indirect quantile treatment effects. 
	
	Let $Y_{a}=1\{Y\leq a\}$ be an indicator function which is one if outcome $Y$ is smaller than or equal to $a$ (and zero otherwise) and $W_{a}=\left(Y_{a}, D, M, X\right)$ be a vector of the observed variables. An estimator of the c.d.f. of the potential outcome that satisfies Neyman-orthogonality can be derived from the efficient influence function (EIF) of $F_{Y\left(d,M\left(d^{\prime}\right)\right)}\left(a\right)$:
	\begin{align}
		\psi_{d,d^{\prime},a}^{\theta}\left(W_{a};v_{a}\right) & =\psi_{d,d^{\prime},a}\left(W_{a};v_{a}\right)-\theta,\label{efficient_score}
	\end{align}
	where 
	\begin{align}
		\psi_{d,d^{\prime},a}\left(W_{a};v_{a}\right) & = \frac{1\left\{ D=d\right\} }{f_{D|X}\left(d^{\prime}|X\right)}\frac{f_{D|M,X}\left(d^{\prime}|M,X\right)}{f_{D|M,X}\left(d|M,X\right)}\times\left[Y_{a} -F_{Y|D,M,X}\left(a|d,M,X\right)\right]\nonumber \\
		& +\frac{1\left\{ D=d^{\prime}\right\} }{f_{D|X}\left(d^{\prime}|X\right)}\times\left[F_{Y|D,M,X}\left(a|d,M,X\right)-g_{d,d^{\prime},a}\left(X\right)\right]+g_{d,d^{\prime},a}\left(X\right)\label{psi_d_dprime_a}
	\end{align} for $\left(d,d^{\prime}\right)\in \{0,1\}^{2}$, $a\in \mathcal{A}$ and $v_{a}$ denoting the vector of nuisance parameters. Let $\theta_{d,d^{\prime},a}$ denote the value of $\theta$ that satisfies $E[\psi_{d,d^{\prime},a}^{\theta}\left(W_{a};v_{a}\right)]=0$:
	\begin{equation}
		\theta_{d,d^{\prime},a} = E\left[ \psi_{d,d^{\prime},a}\left(W_{a};v_{a}\right)\right].
		\label{triply_robust}
	\end{equation}
	We can show that $\theta_{d,d^{\prime},a}=F_{Y(d,M(d^{\prime})}\left(a\right)$ if Assumptions 1.1 to 1.4 hold, see the appendix for a derivation of these results. Therefore, we may use the sample analogue of equation (\ref{triply_robust}) to estimate $F_{Y\left(d,M\left(d^{\prime}\right)\right)}\left(a\right)$. A similar strategy was previously used to derive the triply robust approach for estimating $E\left[Y\left(d,M\left(d^{\prime}\right)\right)\right]$ in \citet{TS_2012} and \cite{FHLLS_2022}. If $d=d^{\prime}$, then the estimator of equation (\ref{triply_robust}) reduces to
	\begin{equation}
		\theta_{d,d,a} =  E\left[ \psi_{d,d,a}\left(W_{a};v_{a}\right)\right] \label{doubly_robust},
	\end{equation}
	where
	\begin{align*}
		\psi_{d,d,a}\left(W_{a};v_{a}\right)&=\frac{1\left\{ D=d\right\} }{f_{D|X}\left(d|X\right)}\times\left[Y_{a} -g_{d,d,a}\left(X\right)\right]+g_{d,d,a}\left(X\right),\\
		g_{d,d,a}\left(X\right) & = \int F_{Y|D,M,X}\left(a|d,m,X\right)f_{M|D,X}\left(m|d,X\right)dm\\
		& = F_{Y|D,X}\left(a|d,X\right),
	\end{align*}for $d\in \{0,1\}$ and $a\in \mathcal{A}$. We may use the sample analogue of equation (\ref{doubly_robust}) to estimate $F_{Y\left(d,M\left(d\right)\right)}\left(a\right)$. This is in analogy to the doubly robust approach for estimating $E\left[Y\left(d,M\left(d\right)\right)\right]$ in \citet{RRZ_1994} and \citet{Hahn_1998}.
	
	We note that by using the Bayes rule, we can rewrite equation (\ref{psi_d_dprime_a}) alternatively as:
	\begin{align}
		\psi_{d,d^{\prime},a}^{\prime}\left(W_{a};v_{a}^{\prime}\right) & = \frac{1\left\{ D=d\right\} }{f_{D|X}\left(d|X\right)}\frac{f_{M|D,X}\left(M|d^{\prime},X\right)}{f_{M|D,X}\left(M|d,X\right)}\times\left[Y_{a} -F_{Y|D,M,X}\left(a|d,M,X\right)\right]\nonumber \\
		& +\frac{1\left\{ D=d^{\prime}\right\} }{f_{D|X}\left(d^{\prime}|X\right)}\times\left[F_{Y|D,M,X}\left(a|d,M,X\right)-g_{d,d^{\prime},a}\left(X\right)\right]+g_{d,d^{\prime},a}\left(X\right). \label{psi_d_dprime_a1}
	\end{align} Therefore, 
	\begin{equation}
		\theta_{d,d^{\prime},a}^{\prime}=E\left[\psi_{d,d^{\prime},a}^{\prime}\left(W_{a};v_{a}^{\prime}\right)\right] \label{triply_robust1}
	\end{equation}
	can also be used to construct an estimator of $F_{Y(d,M(d^{\prime})}\left(a\right)$. There are several differences between the estimators based on equations (\ref{triply_robust}) and (\ref{triply_robust1}). Making use of equation (\ref{triply_robust}) requires estimating four nuisance parameters: $f_{D|X}\left(d|x\right)$, $f_{D|M,X}\left(d|m,x\right)$, $F_{Y|D,M,X}\left(a|d,m,x\right)$ and $g_{d,d^{\prime},a}\left(x\right)$. Since $D$ is binary, the first two nuisance parameters may for instance be estimated by a logit or probit model. The conditional c.d.f. $F_{Y|D,M,X}\left(a|d,m,x\right)$ can be estimated by distributional regression (DR) \citep{CFM_2013}. $g_{d,d^{\prime},a}\left(x\right)$ might be estimated by regression imputation as outlined in equation (\ref{estimator_yme}). However, the estimator based on equation (\ref{triply_robust1}) requires only three nuisance parameter estimates of $f_{D|X}\left(d|x\right)$, $f_{M|D,X}\left(m|d,x\right)$ and $F_{Y|D,M,X}\left(a|d,m,x\right)$. We may estimate $g_{d,d^{\prime},a}\left(x\right)$ based on equation (\ref{g_hat}) after having estimated $f_{M|D,X}\left(m|d,x\right)$ and $F_{Y|D,M,X}\left(a|d,m,x\right)$. The estimator based on equation (\ref{triply_robust1}) appears particularly attractive if the mediator $M$ is discrete and takes a finite (and relatively small) number of values. However, if $M$ is continuous, estimation based on equation (\ref{triply_robust}) may appear more attractive, because it avoids estimating the conditional density $f_{M|D,X}\left(m|d,x\right)$ and the integral in equation (\ref{g_hat}) to obtain an estimate of $g_{d,d^{\prime},a}\left(x\right)$.
	
	The estimators based on equations (\ref{triply_robust}) and (\ref{triply_robust1}) also differ in terms of their robustness to misspecification of the nuisance parameters.
	Let $\hat{\theta}_{d,d^{\prime},a}$ and $\hat{\theta}_{d,d^{\prime},a}^{\prime}$ denote estimators based on equations (\ref{triply_robust}) and (\ref{triply_robust1}) and the respective estimators of the nuisance parameters. Applying the theorem of semiparametric efficiency in \cite{TS_2012} and \cite{Zhou_2022}, we can show that under certain regularity conditions, the following results hold for estimating $F_{Y\left(d,M\left(d^{\prime}\right)\right)}\left(a\right)$ at outcome value $a$: 
	\begin{itemize}
		\item If $F_{Y|D,M,X}\left(a|d,m,x\right)$ and $f_{M|D,X}\left(m|d,x\right)$
		are correctly specified, $\hat{\theta}_{d,d^{\prime},a}^{\prime}\overset{p.}{\longrightarrow}\theta_{d,d^{\prime},a}^{\prime}$.
		\item If $F_{Y|D,M,X}\left(a|d,x,m\right)$ and $f_{D|X}\left(d^{\prime}|x\right)$
		are correctly specified, $\hat{\theta}_{d,d^{\prime},a}^{\prime}\overset{p.}{\longrightarrow}\theta_{d,d^{\prime},a}^{\prime}$.
		\item  If $f_{D|X}\left(d^{\prime}|x\right)$ and $f_{M|D,X}\left(m|d,x\right)$
		are correctly specified, $\hat{\theta}_{d,d^{\prime},a}^{\prime}\overset{p.}{\longrightarrow}\theta_{d,d^{\prime},a}^{\prime}$.
	\end{itemize}
	
	This implies that if two of the nuisance parameters entering equation (\ref{psi_d_dprime_a1}) are correctly specified, while also certain regularity conditions and Assumptions 1.1 to 1.4 hold, then $\hat{\theta}_{d,d^{\prime},a}^{\prime}$ is a consistent estimator of $F_{Y\left(d,M\left(d^{\prime}\right)\right)}\left(a\right)$ at outcome value $a$. In contrast, $\hat{\theta}_{d,d^{\prime},a}\overset{p.}{\longrightarrow}\theta_{d,d^{\prime},a}$ only holds if $f_{D|X}(d|x)$ is consistently estimated. If the latter holds and only one of the other three nuisance parameters in equation (\ref{psi_d_dprime_a}) is misspecified, while certain regularity conditions and Assumptions 1.1 to 1.4 are satisfied, then $\hat{\theta}_{d,d^{\prime},a}$ remains a consistent estimator of $F_{Y\left(d,M\left(d^{\prime}\right)\right)}\left(a\right)$ at outcome value $a$. Finally, if all nuisance parameters are correctly specified and consistently estimated, while certain regularity conditions and Assumptions 1 to 4 also hold, then both $\hat{\theta}_{d,d^{\prime},a}$ and $\hat{\theta}_{d,d^{\prime},a}^{\prime}$ are semiparametrically efficient. 
	
	\subsection{Improving Finite Sample Behavior}
	The estimate of the c.d.f. of $Y\left(d,M\left(d^{\prime}\right)\right)$ can be inverted at a specific rank $\tau$ to obtain an estimate of the $\tau$th quantile, which we denote by $Q_{Y\left(d,M\left(d^{\prime}\right)\right)}\left(\tau\right)$. Suppose $Y\left(d,M\left(d^{\prime}\right)\right)$ is continuous and let the grid of points used for the estimation be a non-decreasing sequence $\{a_{l}\}_{l=1}^{L}$, where $0<\underline{a}<a_{1}<a_{2}<\ldots<a_{L}<\bar{a}<\infty$. Let  $\hat{p}_{l}$ 
	denote an estimate of $F_{Y\left(d,M\left(d^{\prime}\right)\right)}\left(a_{l}\right)$ (e.g., the K-fold cross-fitting estimate, see Section 2.2). Note that $\hat{p}_{l}$ is not necessarily bounded away from 0 and 1 nor monotonically increasing in $a_{l}$, as required for a valid c.d.f. For this reason, we apply two additional constraints on the estimates $\hat{p}_{l}$, $l=1,\ldots,L$. The first one restricts their values to be within
	the range $\left[0,1\right]$. That is, we replace $\hat{p}_{l}$
	with $\tilde{p}_{l}=\max\left\{ \min\left\{ \hat{p}_{l},1\right\} ,0\right\} $.
	Then we follow \citet{CFG_2010} and use the rearrangement operator to sort $\tilde{p}_{l}$ in non-decreasing order. Let $\left(\tilde{p}_{\left(1\right)},\tilde{p}_{\left(2\right)},\ldots,\tilde{p}_{\left(L\right)}\right)$
	be the sorted sequence of $\tilde{p}_{l}$, $l=1,2,\ldots,L$.
	The sequence $\left(\tilde{p}_{\left(1\right)},\tilde{p}_{\left(2\right)},\ldots,\tilde{p}_{\left(L\right)}\right)$
	is our final estimate of the c.d.f. of $Y\left(d,M\left(d^{\prime}\right)\right)$
	at $\left(a_{1},a_{2},\ldots,a_{L}\right)$. 
	We then fit a function for points $\left(\tilde{p}_{\left(l\right)},a_{l}\right)$, $l=1,2,\ldots,L$ with linear interpolation, and use the fitted function to calculate the value of $a$ at rank $\tau$ to estimate  $Q_{Y\left(d,M\left(d^{\prime}\right)\right)}\left(\tau\right)$,\footnote{The monotonicity property is preserved under the linear interpolation.}, which permits estimating the quantile treatment effects (\ref{NDQTE}) to (\ref{TQTE}). When $Y\left(d,M\left(d^{\prime}\right)\right)$ is discrete, we need not fit the function for points $\left(\tilde{p}_{\left(l\right)},a_{l}\right)$; we may obtain $Q_{Y\left(d,M\left(d^{\prime}\right)\right)}\left(\tau\right)$ by directly using the definition of the $\tau$th-quantile.
	
	\subsection{K-Fold Cross-Fitting}
	Neyman-orthogonality may mitigate regularization bias coming from machine learning-based estimation of the nuisance parameters in
	equations (\ref{triply_robust}) or (\ref{triply_robust1}). To also safeguard against overfitting bias, we follow \citet{CCDDHNR_2018}
	and \cite{FHLLS_2022} and apply K-fold cross-fitting to estimate the nuisance parameters and the potential outcome distributions, $F_{Y\left(d,M\left(d^{\prime}\right)\right)}\left(a\right)$, in different parts of the data. To describe the approach,  let  $Y_{a,i}=1\{Y_{i}\leq a\}$ and $W_{a,i}=\left(Y_{a,i},D_{i},M_{i},X_{i}\right)$ denote the $i$th
	observation, $i=1,2,\ldots,n$. In the following, we use the estimator
	based on equation (\ref{triply_robust}) to illustrate K-fold cross-fitting. 
	\begin{enumerate}
		\item Randomly split the $n$ samples into $K$ (mutually exclusive) subsamples
		of equal sample size $n_{k}=n/K$, $k=1,2,\ldots,K$.
		Let $I_{k}$, $k=1,2,\ldots,K$ denote the set of indices for the
		$K$ different subsamples. Let $I_{k}^{c}$, $k=1,2,\ldots,K$ denote
		the complement set of $I_{k}$: $I_{k}^{c}=\left\{ 1,2,\ldots,n\right\} \setminus I_{k}$.
		\item For each $k$, estimate the model parameters of the nuisance parameters $F_{Y|D,M,X}\left(a\right)$,
		$f_{D|X}\left(d^{\prime}|X\right)$, $f_{D|M,X}\left(d|M,X\right)$
		and $g_{d,d^{\prime},a}\left(X\right)$ based on observations $W_{a,i}$,
		$i\in I_{k}^{c}$. For observations $W_{a,i}$, $i\in I_{k}$, predict the nuisance parameters: $\hat{F}_{Y|D,M,X}^{(k)}\left(a|D_{i},M_{i},X_{i}\right)$,
		$\hat{f}_{D|X}^{(k)}\left(d^{\prime}|X_{i}\right)$, $\hat{f}_{D|M,X}^{(k)}\left(d|M_{i},X_{i}\right)$,
		$\hat{f}_{D|M,X}^{(k)}\left(d^{\prime}|M_{i},X_{i}\right)$ and $\hat{g}_{d,d^{\prime},a}\left(X\right)$,
		$i\in I_{k}$.
		\item For each $k$, compute the estimate of $F_{Y\left(d,M\left(d^{\prime}\right)\right)}\left(a\right)$
		using the predicted nuisance parameters of step 2 as 
		\begin{align}
			\hat{\theta}_{d,d^{\prime},a}^{(k)} & =\frac{1}{n_{k}}\sum_{i\in I_{k}}\left\{ \frac{1\left\{ D_{i}=d\right\} \hat{f}_{D|M,X}^{(k)}\left(d^{\prime}|M_{i},X_{i}\right)}{\hat{f}_{D|X}^{(k)}\left(d^{\prime}|X_{i}\right)\hat{f}_{D|M,X}^{(k)}\left(d|M_{i},X_{i}\right)}\right.\nonumber \\
			& \times\left[1\left\{ Y_{i}\leq a\right\} -\hat{F}_{Y|D,M,X}^{(k)}\left(a|d,M_{i},X_{i}\right)\right]\label{K_fold_theta}\\
			& +\frac{1\left\{ D_{i}=d^{\prime}\right\} }{\hat{f}_{D|X}^{(k)}\left(d^{\prime}|X_{i}\right)}\left[\hat{F}_{Y|D,M,X}^{(k)}\left(a|d,M_{i},X_{i}\right)-\hat{g}_{d,d^{\prime},a}^{(k)}\left(X_{i}\right)\right]\nonumber \\
			& \left.+\hat{g}_{d,d^{\prime},a}^{(k)}\left(X_{i}\right)\right\} .\nonumber 
		\end{align}
		\item Average $\hat{\theta}_{d,d^{\prime},a}^{(k)}$ over $k=1,2,\ldots,K$
		to obtain the final estimate of $F_{Y\left(d,M\left(d^{\prime}\right)\right)}\left(a\right)$:
		\begin{equation}
			\hat{\theta}_{d,d^{\prime},a}=\frac{1}{K}\sum_{k=1}^{K}\hat{\theta}_{d,d^{\prime},a}^{(k)}.\label{theta_hat}
		\end{equation}
	\end{enumerate}
	We repeat steps 1 to 4 for a grid of points $a$,
	a non-decreasing sequence $\{a_{l}\}_{l=1}^{L}$, where $0<\underline{a}<a_{1}<a_{2}<\ldots<a_{L}<\bar{a}<\infty$,
	to construct the estimate of the c.d.f.\ profile of $Y\left(d,M\left(d^{\prime}\right)\right)$. In Section 3, we will establish the asymptotic properties uniformly over $a\in \cal{A}$ for the K-fold cross-fitting estimator of equation (\ref{theta_hat}). 
	
	Alternatively, we can also construct the K-fold cross-fitting estimator
	based on equation (\ref{triply_robust1}):
	\begin{equation}
		\hat{\theta}_{d,d^{\prime},a}^{\prime}=\frac{1}{K}\sum_{k=1}^{K}\hat{\theta}_{d,d^{\prime},a}^{\prime(k)},\label{theta_hat1}
	\end{equation}
	where 
	\begin{align}
		\hat{\theta}_{d,d^{\prime},a}^{\prime(k)}, & =\frac{1}{n_{k}}\sum_{i\in I_{k}}\left\{ \frac{1\left\{ D_{i}=d\right\} \hat{f}_{M|D,X}^{(k)}\left(M_{i}|d^{\prime},X_{i}\right)}{\hat{f}_{D|X}^{(k)}\left(d|X_{i}\right)\hat{f}_{M|D,X}^{(k)}\left(M_{i}|d,X_{i}\right)}\right.\nonumber \\
		& \times\left[1\left\{ Y_{i}\leq a\right\} -\hat{F}_{Y|D,M,X}^{(k)}\left(a|d,M_{i},X_{i}\right)\right]\label{K_fold_theta1}\\
		& +\frac{1\left\{ D_{i}=d^{\prime}\right\} }{\hat{f}_{D|X}^{(k)}\left(d^{\prime}|X_{i}\right)}\left[\hat{F}_{Y|D,M,X}^{(k)}\left(a|d,M_{i},X_{i}\right)-\hat{g}_{d,d^{\prime},a}^{(k)}\left(X_{i}\right)\right]\nonumber \\
		& \left.+\hat{g}_{d,d^{\prime},a}^{(k)}\left(X_{i}\right)\right\} .\nonumber 
	\end{align}

	\subsection{Nuisance Parameter Estimation}
	In the case when $D$ and $M$ are binary variables, $f_{D|X}\left(d|x\right)$ and $f_{M|D,X}\left(m|d,x\right)$
	can be estimated straightforwardly, for instance by logit or probit models. If $M$ is continuous, we might prefer $\hat{\theta}_{d,d^{\prime},a}$
	to avoid estimating $f_{M|D,X}\left(m|d,x\right)$, while the required 
	nuisance parameter $f_{D|M,X}\left(d|m,x\right)$ can be straightforwardly 
	estimated if $D$ is binary. The estimation of any nuisance parameter may be based on machine learning, for example, the lasso or neural networks. For estimating $F_{Y|D,M,X}\left(a|d,m,x\right)$, we use distributional regression (DR). Conditional on $\left(D=d,M=m,X=x\right)$,
	the conditional c.d.f. of $Y$ may be written as
	\begin{equation}
		F_{Y|D,M,X}\left(a|d,m,x\right)=E\left[Y_a |d,m,x\right],\label{eq1}
	\end{equation}
	where $a$ is a constant and $a\in\mathcal{A}\subset\mathbb{R}$,
	with $\mathcal{A}$ being a countable subset of $\mathbb{R}$, and
	$Y_{a}=1\left\{ Y\leq a\right\} $ being an indicator function for the event
	$\left\{ Y\leq a\right\} $. Equation (\ref{eq1}) is the building
	block for constructing the DR estimator \citep{FP_1995, CFM_2013}, which is based on using the binary dependent variable $Y_{a}$ to estimate $F_{Y|D,M,X}\left(a|d,m,x\right)$ by a regression approach that estimates $E\left[Y_{a}|d,m,x\right]$. For example, one may
	assume that the c.d.f. is linear in variables $\left(D,M,X\right)$,
	their higher-order terms and interaction terms, and estimate a linear probability model (LPM) by OLS. However, the LPM does not guarantee
	that the estimated $F_{Y|D,M,X}\left(a|d,m,x\right)$ will lie within
	the interval $\left[0,1\right]$. To overcome this difficulty, we
	may assume that
	\[
	F_{Y|D,M,X}\left(a|d,m,x\right)=G_{a}\left(v_{a}\left(d,m,x\right)\right),
	\]
	where $v_{a}:\left(D,M,X\right)\mapsto\mathbb{R}$ and $G_{a}\left(.\right)$
	is a link function which is non-decreasing and satisfies $G_{a}:\mathbb{R}\mapsto\left[0,1\right]$,
	$G_{a}\left(y\right)\rightarrow0$ if $y\rightarrow-\infty$ and $G_{r}\left(y\right)\rightarrow1$
	if $y\rightarrow\infty$. The choices of $v_{a}\left(D,M,X\right)$
	and the link function $G_{a}\left(.\right)$ are flexible. For example,
	$v_{a}\left(D,M,X\right)$ might be a neural network (see Section 3.1) or a transformation of $\left(D,M,X\right)$, which can vary with $a$. 
	Depending on whether $Y$ is continuous or discrete, the link function $G_{a}\left(.\right)$ may be the logit, probit, linear, log-log, Gosset, the Cox proportional
	hazard function or the incomplete Gamma function. 
	As pointed out by \citet{CFM_2013}, for any given link function $G_{a}\left(.\right)$, we can approximate  $F_{Y|D,M,X}\left(a|d,m,x\right)$ arbitrarily well if $v_{a}\left(D,M,X\right)$
	is sufficiently flexible. 
	
	There are various ways to implement DR, and a popular choice is maximum likelihood estimation. Let $\hat{v}_{a}\left(D,M,X\right)$
	denote the maximum likelihood estimator of $v_{a}\left(D,M,X\right)$, which is obtained
	by 
	\begin{equation}
		\max_{v_{a}\in\mathcal{V}_{a}}\sum_{i=1}^{n}\left\{ Y_a \ln G_{a}\left(v_{a}\left(D_{i},M_{i},X_{i}\right)\right)+(1-Y_a) \ln\left[1-G_{a}\left(v_{a}\left(D_{i},M_{i},X_{i}\right)\right)\right]\right\} ,\label{MLE}
	\end{equation}
	where $\mathcal{V}_{a}$ is the parameter space of $\hat{v}_{a}\left(D,M,X\right)$. Then, $G\left(\hat{v}_{a}\left(D,M,X\right)\right)$ is an estimate of $F_{Y|D,M,X}\left(a|d,m,x\right)$. $v_{a}$ may be estimated by machine learning methods when the dimension of $X$ is large and/or the functional form of $v_{a}$ is complex. 
	
	\subsection{Summarizing the Estimation Approach}
	Our estimation approach can be summarized as follows:
	\begin{itemize}
		\item[Step 1] \textit{Modeling the nuisance parameters}\\ The nuisance parameters include $f_{D|X}(d|x)$, $f_{D|M,X}(d|m,x)$, $F_{Y|D,M,X}(a|d,m,x)$ and $g_{d,d^{\prime},a}(X)$. Depending on the properties of $(Y,D,M)$, choose appropriate functional forms for the nuisance parameters.
		\item[Step 2] \textit{Estimating the nuisance parameters}\\ Estimate the nuisance parameters by K-fold cross-fitting, as described in points 1 and 2 in Section 2.2.
		\item[Step 3] \textit{Computing the Neyman-orthogonal estimator}\\ With estimates of the nuisance parameters from K-fold cross-fitting, compute the Neyman-orthogonal estimator as described in points 3 and 4 in Section 2.2.
		\item[Step 4] \textit{Repeating steps 2 to 3 for a grid of points $\{a_{l}\}_{l=1}^{L}$, where $0<\underline{a}<a_{1}<a_{2}<\ldots<a_{L}<\bar{a}<\infty$}\\ Notice that only $F_{Y|D,M,X}(a|d,m,x)$ needs to be re-estimated, not the remaining nuisance parameters. 
		\item[Step 5] \textit{Adopting the two constraints described in Section 2.1 to obtain $\tilde{p}_{(l)}$, the final estimate of $F_{Y\left(d,M\left(d^{\prime}\right)\right)}\left(a_{l}\right)$}
		\item[Step 6] \textit{If $Y(d,M(d^{\prime}))$ is continuous, fitting a function for the points $\left(\tilde{p}_{\left(l\right)},a_{l}\right)$, $l=1,2,\ldots,L$, and using the fitted function to obtain $\hat{Q}_{Y\left(d,M\left(d^{\prime}\right)\right)}\left(\tau\right)$, the estimate of $Q_{Y\left(d,M\left(d^{\prime}\right)\right)}\left(\tau\right)$; if $Y(d,M(d^{\prime}))$ is discrete, using the definition of the $\tau$th quantile to obtain $\hat{Q}_{Y\left(d,M\left(d^{\prime}\right)\right)}\left(\tau\right)$}. 
		\item[Step 7] \textit{Estimating the quantile treatment effects of interest as defined in equations (\ref{NDQTE}) to (\ref{TQTE}) based on $\hat{Q}_{Y\left(d,M\left(d^{\prime}\right)\right)}\left(\tau\right)$}. 
	\end{itemize}
	
	\subsection{Multiplier Bootstrap}
	We propose a multiplier bootstrap for statistical inference. Let $\{\xi_i\}_{i=1}^n$ be a sequence of i.i.d.\ (pseudo) random variables, independent of the sample path $\{(Y_{a,i},M_i,D_i,X_i)\}_{i=1}^n$, with $E[\xi_i]=0$, $Var(\xi_i)=1$ and $E\left[\exp\left(\left|\xi_i\right|\right)\right]<\infty$. Let $\hat{v}_{k,a}$ denote a vector containing the
	K-fold cross-fitting estimates of the nuisance parameters, whose model parameters are estimated based on observations in the complement set, $W_{a,i}$ with ${i\in I^{c}_{k}}$. The proposed multiplier bootstrap estimator for $\hat{\theta}_{d,d^{\prime},a}$ in equation (\ref{theta_hat}) is given by:
	\begin{equation}
		\hat{\theta}_{d,d^{\prime},a}^{*}=\hat{\theta}_{d,d^{\prime},a}+\frac{1}{n}\sum_{i=1}^{n}\xi_{i}\left(\psi_{d,d^{\prime},a}\left(W_{a,i};\hat{v}_{k,a}\right)-\hat{\theta}_{d,d^{\prime},a}\right),\label{m_bootstrap}
	\end{equation}
	The multiplier bootstrap estimator does not require re-estimating the nuisance parameters and re-calculating the causal parameters of interest in each bootstrap sample. This is particularly useful in our context, since our estimation approach is to be repeatedly applied across grid points $a$. After obtaining $\hat{\theta}_{d,d^{\prime},a}^{*}$ for all grid points, we use the procedures introduced in Section 2.1 to construct the bootstrap estimate of the c.d.f.\ of $Y(d,M(d^{\prime}))$ at these grid points and the bootstrap estimate of the $\tau$th quantile $Q_{Y(d,M(d^{\prime}))}^{*}(\tau)$. Section 3 establishes the uniform validity of the proposed multiplier bootstrap procedure.

\section{Theoretical Results}\label{theores}
\subsection{Notation}
In this section, we show that the proposed K-fold cross-fitting estimator in Section 2.2 is uniformly valid under certain conditions. We focus on establishing the theoretical properties of the estimator in equation (\ref{theta_hat}) and note that the assumptions and procedures required for demonstrating uniform validity of estimation based on equation (\ref{theta_hat1}) are similar and omitted for this reason. To ease notation in our analysis, let $g_{1d}^{0}\left(X\right):=F_{D|X}^{0}\left(d|X\right)$, $g_{2d}^{0}\left(M,X\right):=F_{D|M,X}^{0}\left(d|M,X\right)$, $g_{3a}^{0}\left(D,M,X\right):=F_{Y|D,M,X}^{0}\left(a|D,M,X\right)$ and $g_{4ad}^{0}\left(D,X\right):=E\left[F_{Y|D,M,X}^{0}\left(a|d,M,X\right)|D,X\right]$ denote the nuisance parameters. Let $v_{a}^{0}$ denote the vector containing these true nuisance parameters and $\mathcal{G}_{a}$ be the set of all $v_{a}^{0}$. Let $F_{Y\left(d,M\left(d^{\prime}\right)\right)}^{0}\left(a\right)$ denote the true c.d.f.\ of the potential outcome $Y(d,M(d^{\prime}))$. The EIF of $F_{Y\left(d,M\left(d^{\prime}\right)\right)}^{0}\left(a\right)$ is $\psi_{d,d^{\prime},a}^{\theta}\left(W_{a};v_{a}^{0}\right)=\psi_{d,d^{\prime},a}\left(W_{a};v_{a}^{0}\right)-\theta$, where
\begin{align*}
	\psi_{d,d^{\prime},a}\left(W_{a};v_{a}^{0}\right) & =\frac{1\left\{ D=d\right\} \left[1-g_{2d}^{0}\left(M,X\right)\right]}{\left[1-g_{1d}^{0}\left(X\right)\right]g_{2d}^{0}\left(M,X\right)}\times\left[1\left\{ Y\leq a\right\} -g_{3a}^{0}\left(d,M,X\right)\right]\\
	& +\frac{1\left\{ D=d^{\prime}\right\} }{1-g_{1d}^{0}\left(X\right)}\left[g_{3a}^{0}\left(d,M,X\right)-g_{4ad}^{0}\left(d^{\prime}X\right)\right]+g_{4ad}^{0}\left(d^{\prime},X\right).\label{psi_0}
\end{align*}
Under Assumptions 1.1 to 1.4, it can be shown that 
\begin{equation}
	F_{Y\left(d,M\left(d^{\prime}\right)\right)}^{0}\left(a\right)=E\left[\psi_{d,d^{\prime},a}\left(W_{a},v_{a}^{0}\right)\right]\label{identification}
\end{equation}
for all $a\in \mathcal{A}$ and $\left(d,d^{\prime}\right)\in\left\{ 0,1\right\} ^{2}$ (see proof of Theorem 1).

In the subsequent theoretical analysis, the expectation $E\left[.\right]$ is
operated under the probability $P\in\mathcal{P}_{n}$. Let $N=n/K$
be the size in a fold or subsample, where $K$ is a fixed number. Let $E_{n}:=n^{-1}\sum_{i=1}^{n}\varsigma_{W_{i}}$
and $E_{N,k}:=N^{-1}\sum_{i\in I_{k}}\varsigma_{W_{i}}$ where $\varsigma_{w}$
is a probability distribution degenerating at $w$ and $I_{k}$ is
a set of indices of observations in the $k$th subsample. Let $Z\rightsquigarrow Z^{\prime}$
denote a random variable $Z$ that weakly converges to a random variable
$Z^{\prime}$. Let $\left\Vert x\right\Vert $ denote the
$l^{1}$ norm and $\left\Vert x\right\Vert _{q}$ denote
the $l^{q}$ norm, $q\geq2$ for a deterministic vector $x$.
Let $\left\Vert X\right\Vert _{P,q}$ denote $\left(E[\left\Vert X\right\Vert ^{q}]\right)^{1/q}$
for a random vector $X$. The function $\psi_{d,d^{\prime},a}$
for identifying the parameter of interest and constructing the estimator
is such that $\psi_{d,d^{\prime},a}\left(w,t\right):\mathcal{W}_{a}\times\mathcal{V}_{a}\longmapsto\mathbb{R}$,
where $\left(d,d^{\prime}\right)\in\left\{ 0,1\right\} ^{2}$, $a\in\mathcal{A}\subset\mathbb{R}$,
$\mathcal{W}_{a}\subset\mathbb{R}^{d_{w}}$ is a $d_{w}$ dimensional
Borel set and $\mathcal{V}_{a}$ is a $G$ dimensional set of Borel
measurable maps. Let $\boldsymbol{\psi}_{a}=\left(\psi_{1,1,a},\psi_{1,0,a}\psi_{0,1,a},\psi_{0,0,a}\right)$
and $\boldsymbol{\psi}_{a}:\mathcal{W}_{a}\times\mathcal{V}_{a}\longmapsto\mathbb{R}^{4}$.
Let $v_{a,g}^{0}:\mathcal{U}_{a}\longmapsto\mathbb{R}$ denote the
$g$th true nuisance parameter, where $\mathcal{U}_{a}\subseteq\mathcal{W}_{a}$
is a Borel set, and $v_{a}^{0}:=\left(v_{a,1}^{0},\ldots,v_{a,G}^{0}\right)\in\mathcal{V}_{a}$
denote the vector of these true nuisance parameters. Let $\hat{v}_{k,a,g}:\mathcal{U}_{a}\longmapsto\mathbb{R}$
denote an estimate of $v_{a,g}^{0}$, which is obtained by using the
K-fold cross-fitting, such that the model parameters of the nuisance parameters are estimated based on observations in the complement set, $W_{a,j}$ with $j\in I_{k}^{c}$.
Let $\hat{v}_{k,a}:=\left(\hat{v}_{k,a,1},\ldots,\hat{v}_{k,a,G}\right)$
denote the vector of these estimates. $v_{a}^{0}$ and $\hat{v}_{k,a}$
are both functions of $U_{a}\in\mathcal{U}_{a}$, a subvector of $W_{a}\in\mathcal{W}_{a}$.
But to ease notation, we will write $v_{a}^{0}$ and $\hat{v}_{k,a}$
instead $v_{a}^{0}\left(U_{a}\right)$ and $\hat{v}_{k,a}\left(U_{a}\right)$.

$\boldsymbol{\psi}_{a}\left(W_{a},v\right)$ denotes $\boldsymbol{\psi}_{a}$
with elements $\psi_{d,d^{\prime},a}\left(W_{a};v\right)$,
$\left(d,d^{\prime}\right)\in\left\{ 0,1\right\} ^{2}$. The
parameter of interest is $F_{Y\left(d,M\left(d^{\prime}\right)\right)}^{0}\left(a\right)$,
which can be identified by equation (\ref{identification}), the expectation of
$\psi_{d,d^{\prime},a}\left(W_{a};v\right)$ evaluated at the true nuisance parameters $v_{a}^{0}$. Let $\theta_{d,d^{\prime},a}^{0}:=E\left[\psi_{d,d^{\prime},a}\left(W_{a},v_{a}^{0}\right)\right]$. The proposed estimator of $\theta_{d,d^{\prime},a}^{0}$ is the K-fold cross-fitting estimator $\hat{\theta}_{d,d^{\prime},a}$ of (\ref{theta_hat}), in which $\hat{\theta}_{d,d^{\prime},a}^{\left(k\right)}=N^{-1}\sum_{i\in I_{k}}\psi_{d,d^{\prime},a}\left(W_{a,i};\hat{v}_{k,a}\right).$
In our case, the vector $\hat{v}_{k,a}$ contains estimates of the
four true nuisance parameters $g_{1d}^{0}\left(X\right),g_{2d}^{0}\left(M,X\right),g_{3a}^{0}\left(D,M,X\right)$
and $g_{4ad}^{0}\left(D,X\right)$ when estimating the model parameters based on observations in the complement set, $W_{a,i}$ with $i\in I_{k}^{c}$. Let
$\boldsymbol{\theta}_{a}^{0}$, $\hat{\boldsymbol{\theta}}_{a}$, $\hat{\boldsymbol{\theta}}_{a}^{\left(k\right)}$ and $\mathbf{F}^{0}(a)$ denote vectors containing $\theta_{d,d^{\prime},a}^{0}$, $\hat{\theta}_{d,d^{\prime},a}$, $\hat{\theta}_{d,d^{\prime},a}^{(k)}$ and $F_{Y\left(d,M\left(d^{\prime}\right)\right)}^{0}\left(a\right)$ for different $\left(d,d^{\prime}\right)\in\{0,1\}^{2}$. We note that $\hat{\boldsymbol{\theta}}_{a}=K^{-1}\sum_{k=1}^{K}\hat{\boldsymbol{\theta}}_{a}^{\left(k\right)}$ and $\boldsymbol{\theta}_{a}^{0}=E\left[\boldsymbol{\psi}_{a}\left(W_{a};v_{a}^{0}\right)\right]$, and if equation (\ref{identification}) holds for all $a\in\mathcal{A}$ and $\left(d,d^{\prime}\right)\in \{0,1\}^{2}$, then $\mathbf{F}^{0}(a)=\boldsymbol{\theta}_{a}^{0}$.
\subsection{Main Results}
To establish the uniform validity of $\hat{\boldsymbol{\theta}}_{a}$
when estimating $\mathbf{F}^{0}(a)$, we impose the following
conditions. 
\begin{assumption} 
	\item[1.] For $\mathcal{P}:=\bigcup_{n=n_{0}}^{\infty}\mathcal{P}_{n}$,
	$Y_{a}:=1\left\{ Y\leq a\right\} $ satisfies 
	\begin{align*}
		\lim_{\delta\searrow0}\sup_{P\in\mathcal{P}}\sup_{d_{\mathcal{A}}\left(a,\bar{a}\right)\leq\delta}\left\Vert Y_{a}-Y_{\bar{a}}\right\Vert _{P,2} & =0,\\
		\sup_{P\in\mathcal{P}}E\sup_{a\in\mathcal{A}}\left|Y_{a}\right|^{2+c} & <\infty,
	\end{align*}
	where $\left(a,\bar{a}\right)\in\mathcal{A}$ and $\mathcal{A}$ are a
	totally bounded metric space equipped with a semimetric $d_{\mathcal{A}}$.
	The uniform covering number of the set $\mathcal{G}_{5}:=\left\{ Y_{a}:a\in\mathcal{A}\right\} $
	satisfies 
	\[
	\sup_{Q}\log N\left(\epsilon\left\Vert \mathcal{G}_{5}\right\Vert _{Q,2},\mathcal{G}_{5},\left\Vert \right\Vert _{Q,2}\right)\leq C\log\frac{\text{e}}{\epsilon},
	\]
	for all $P\in\mathcal{P}$, where $B_{5}\left(W\right)=\sup_{a\in\mathcal{A}}\left|Y_{a}\right|$
	is an envelope function with the supremum taken over all finitely
	discrete probability measures $Q$ on $\left(\mathcal{W},\mathcal{X}_{\mathcal{W}}\right)$. 
	
	\item[2.] For $d\in\left\{ 0,1\right\} $, $P\left(\varepsilon_{1}<g_{1d}^{0}\left(X\right)<1-\varepsilon_{1}\right)=1$
	and $P\left(\varepsilon_{2}<g_{2d}^{0}\left(M,X\right)<1-\varepsilon_{2}\right)=1$,
	where $\varepsilon_{1},\varepsilon_{2}\in\left(0,1/2\right)$.
	
	\item[3.]The models for estimating the nuisance parameters $v_{a}^{0}$
	have functional forms  
	\[
	\left(h_{1}\left(f\left(x\right)^{\top}\boldsymbol{\beta}_{1}\right),h_{2}\left(f\left(m,x\right)^{\top}\boldsymbol{\beta}_{2}\right),h_{3}\left(f\left(d,m,x\right)^{\top}\boldsymbol{\beta}_{3}\right),h_{4}\left(f\left(d^{\prime},x\right)^{\top}\boldsymbol{\beta}_{4}\right)\right),
	\]
	respectively, and satisfy the following conditions. 
	\begin{itemize}
		\item[3a.] \textbf{Functional forms of $h_{i}(.)$} The functions $h_{i}$,
		$i=1,2,3$ take the forms of commonly used link functions 
		\[
		\mathcal{L}=\left\{ \mathbf{I}\text{d},\varLambda,1-\varLambda,\Phi,1-\Phi\right\} ,
		\]
		where $\mathbf{I}\text{d}$ is the identity function, $\varLambda$
		is the logistic link, and $\Phi$ is the probit link.Function
		$h_{4}$ has the form of the identity function $\mathbf{I}\text{d}$. 
		\item[3b.] \textbf{Dictionary controls $f(.)$} The dimension of exogenous variables
		is $\dim\left(X\right)=p$ and $\log p=o\left(n^{-1/3}\right)$. The functions $h_{i}$ contain a linear combination of dictionary controls
		$f\left(.\right)$, where $\text{dim}\left(f\left(x\right)\right)=p\times1$
		(dimension of $X$), $\text{dim}\left(f\left(m,x\right)\right)=(p+1)\times1$
		($X$ plus one mediator), $\text{dim}\left(f\left(d,m,x\right)\right)=(p+2)\times1$
		($X$ plus one mediator and one treatment variable), and
		$\text{dim}\left(f\left(d^{\prime},x\right)\right)=(p+1)\times1$
		($X$ plus one treatment variable). 
		\item[3c.] \textbf{Approximately sparsity} The vectors of coefficients $\boldsymbol{\beta}_{i},$
		$i=1,\ldots,4$ satisfy $\left\Vert \boldsymbol{\beta}_{i}\right\Vert _{0}\leq s_{i}$,
		where $\left\Vert \right\Vert _{0}$ denotes the $l^{0}$ norm and
		$s_{i}$ denotes the sparsity index. Furthermore, $\sum_{i=1}^{4}s_{i}\leq s\ll n$
		and 
		\[
		s^{2}\log^{2}\left(p\vee n\right)\log^{2}n\leq\delta_{n}n.
		\]
		Let $\bar{\boldsymbol{\beta}}_{i}$, denote estimators of $\boldsymbol{\beta}_{i}$.
		These estimators are sparse such that $\sum_{i=1}^{4}\left\Vert \bar{\boldsymbol{\beta}}_{i}\right\Vert _{0}\leq C^{\prime}s,$
		where $C^{\prime}<1$ is some constant. 
		\item[3d.] \textbf{Gram matrix} The empirical and population norms induced by the Gram matrix formed by the dictionary $f$ are equivalent on sparse subsets: 
		\[
		\sup_{\left\Vert \delta\right\Vert _{0}\leq s\log n}\left|\left\Vert f^{\top}\delta\right\Vert _{\mathbb{P}_{n},2}/\left\Vert f^{\top}\delta\right\Vert _{P,2}-1\right|\rightarrow0
		\]
		as $n\rightarrow\infty$, and also $\left\Vert \left\Vert f\right\Vert _{\infty}\right\Vert _{P,\infty}\leq L_{n}$.
		\item[3e.] $\left\Vert X\right\Vert _{P,q}<\infty$.
	\end{itemize}
	\item[4.] Given a random subset $I$ of $\left[n\right]=\left\{ 1,\ldots,n\right\} $
	of size $n/K$, let $\hat{\boldsymbol{\beta}_{i}}$ denote an estimate
	of coefficient vector $\boldsymbol{\beta}_{i}$ defined in Assumption
	3. These estimated nuisance parameters 
	\begin{align*}
		\hat{v}_{a} & :=\left(h_{1}\left(f\left(X\right)^{\top}\hat{\boldsymbol{\beta}}_{1}\right),h_{2}\left(f\left(M,X\right)^{\top}\hat{\boldsymbol{\beta}}_{2}\right),h_{3}\left(f\left(D,M,X\right)^{\top}\hat{\boldsymbol{\beta}}_{3}\right),h_{4}\left(f\left(D,X\right)^{\top}\hat{\boldsymbol{\beta}}_{4}\right)\right)\\
		& =\left(\hat{g}_{1d}\left(X\right),\hat{g}_{2d}\left(M,X\right),\hat{g}_{3a}\left(D,M,X\right),\hat{g}_{4ad}\left(D,X\right)\right)
	\end{align*}
	satisfy the following conditions concerning their estimation
	quality. For $d\in\left\{ 0,1\right\} $ 
	\begin{align*}
		P\left(\varepsilon_{1}<\hat{g}_{1d}\left(X\right)<1-\varepsilon_{1}\right) & =1,\\
		P\left(\varepsilon_{2}<\hat{g}_{2d}\left(M,X\right)<1-\varepsilon_{2}\right) & =1,
	\end{align*}
	where $\varepsilon_{1},\varepsilon_{2}>0$. Let $\delta_{n}$ be a sequence converging to zero from above at
	a speed at most polynomial in $n$, e.g., $\delta_{n}\geq n^{-c}$ for some $c>0$. With probability $P$ 
	at least $1-\Delta_{n}$, for $d\in\left\{ 0,1\right\} $, all $a\in\mathcal{A}$
	and $q\geq4$, 
	\begin{align*}
		\left\Vert \hat{v}_{a}-v_{a}^{0}\right\Vert _{P,q} & \leq C,\\
		\left\Vert \hat{v}_{a}-v_{a}^{0}\right\Vert _{P,2} & \leq\delta_{n}n^{-\frac{1}{4}},\\
		\left\Vert \hat{g}_{1d}\left(X\right)-0.5\right\Vert _{P,\infty} & \leq0.5-\epsilon,\\
		\left\Vert \hat{g}_{2d}\left(M,X\right)-0.5\right\Vert _{P,\infty} & \leq0.5-\epsilon,\\
		\left\Vert \hat{g}_{1d}\left(X\right)-g_{1d}^{0}\left(X\right)\right\Vert _{P,2}\left\Vert \hat{g}_{2d}\left(M,X\right)-g_{2d}^{0}\left(M,X\right)\right\Vert _{P,2} & \leq\delta_{n}n^{-\frac{1}{2}},\\
		\left\Vert \hat{g}_{1d}\left(X\right)-g_{1d}^{0}\left(X\right)\right\Vert _{P,2}\left\Vert \hat{g}_{3a}\left(D,M,X\right)-g_{3a}^{0}\left(D,M,X\right)\right\Vert _{P,2} & \leq\delta_{n}n^{-\frac{1}{2}},\\
		\left\Vert \hat{g}_{1d}\left(X\right)-g_{1d}^{0}\left(X\right)\right\Vert _{P,2}\left\Vert \hat{g}_{4ad}\left(D,X\right)-g_{4ad}^{0}\left(D,X\right)\right\Vert _{P,2} & \leq\delta_{n}n^{-\frac{1}{2}},\\
		\left\Vert \hat{g}_{2d}\left(M,X\right)-g_{2d}^{0}\left(M,X\right)\right\Vert _{P,2}\left\Vert \hat{g}_{3a}\left(D,M,X\right)-g_{3a}^{0}\left(D,M,X\right)\right\Vert _{P,2} & \leq\delta_{n}n^{-\frac{1}{2}}.
	\end{align*}
\end{assumption}

Note that in Assumption 2.4, $\hat{v}_{a}$ is by definition constructed based on observations in the complement set $\left(W_{a,i}\right)_{i\in I^{c}}$: $\hat{v}_{a}=\hat{v}_{a}\left(\left(W_{a,i}\right)_{i\in I^{c}}\right)$. When the random subset $I=I_{k}$, then $\hat{v}_{a}=\hat{v}_{k,a}$. Let $\mathbb{G}_{n}$ denote an empirical process $\mathbb{G}_{n}f\left(W\right)=\sqrt{n}\left(E_{n}f\left(W\right)-E\left[f\left(W\right)\right]\right)$,
where $f$ is any $P\in\mathcal{P}_{n}$ integrable function on the
set $\mathcal{W}$. Let $\mathbb{G}_{P}f\left(W\right)$ denote the
limiting process of $\mathbb{G}_{n}f\left(W\right)$, which is a Gaussian
process with zero mean and a finite covariance matrix $E\left[\left(f\left(W\right)-E\left[f\left(W\right)\right]\right)\left(f\left(W\right)-E\left[f\left(W\right)\right]\right)^{\top}\right]$ under probability $P$ (the $P$-Brownian bridge). Based on our notation and assumptions, we obtain the following result concerning the asymptotic behaviour of our estimator. \begin{theorem} If Assumptions 1 and 2 hold, the K-fold cross-fitting estimator $\hat{\boldsymbol{\theta}}_{a}=K^{-1}\sum_{k=1}^{K}\hat{\boldsymbol{\theta}}_{a}^{\left(k\right)}$
	for estimating $\mathbf{F}^{0}(a)$ satisfies 
	\[
	\sqrt{n}\left(\hat{\boldsymbol{\theta}}_{a}-\mathbf{F}^{0}(a)\right)_{a\in\mathcal{A}}=Z_{n,P}+o_{P}\left(1\right),
	\]
	in $l^{\infty}\left(\mathcal{A}\right)^{4}$, uniformly in $P\in\mathcal{P}_{n}$,
	where $Z_{n,P}:=\left( \mathbb{G}_{n}\left(\boldsymbol{\psi}_{a}\left(W_{a};v_{a}^{0}\right)-\boldsymbol{\theta}_{a}^{0}\right)\right) _{a\in\mathcal{A}}$.
	Furthermore, 
	\[
	Z_{n,P}\rightsquigarrow Z_{P}
	\]
	in $l^{\infty}\left(\mathcal{A}\right)^{4}$, uniformly in $P\in\mathcal{P}_{n}$,
	where $Z_{P}:=\left( \mathbb{G}_{P}\left(\boldsymbol{\psi}_{a}\left(W_{a};v_{a}^{0}\right)-\boldsymbol{\theta}_{a}^{0}\right)\right) _{a\in\mathcal{A}}$ and paths of $\mathbb{G}_{P}\left(\boldsymbol{\psi}_{a}\left(W_{a};v_{a}^{0}\right)-\boldsymbol{\theta}_{a}^{0}\right)$ have the properties that uniformly in $P\in\mathcal{P}_{n}$,
	\begin{align*}
		\sup_{P\in\mathcal{P}_{n}}E\left[\sup_{a\in\mathcal{A}}\left\Vert \mathbb{G}_{P}\left(\boldsymbol{\psi}_{a}\left(W_{a};v_{a}^{0}\right)-\boldsymbol{\theta}_{a}^{0}\right)\right\Vert \right] & <\infty,\\
		\lim_{\delta\rightarrow0}\sup_{P\in\mathcal{P}_{n}}E\left[\sup_{d_{\mathcal{A}}\left(a,\bar{a}\right)}\left\Vert \mathbb{G}_{P}\left(\boldsymbol{\psi}_{a}\left(W_{a};v_{a}^{0}\right)-\boldsymbol{\theta}_{a}^{0}\right)-\mathbb{G}_{P}\left(\boldsymbol{\psi}_{\bar{a}}\left(W_{\bar{a}};v_{\bar{a}}^{0}\right)-\boldsymbol{\theta}_{\bar{a}}^{0}\right)\right\Vert \right] & =0.
	\end{align*}
\end{theorem}

Next, we establish the uniform validity of the multiplier bootstrap under Assumptions 1 and 2. As in Assumption 2.4, let $\hat{v}_{a}$ denote the
cross-fitting estimator of the nuisance parameters, whose model parameters are estimated based on observations in the complement set, $W_{a,i}$ with ${i\in I^{c}_{k}}$. Recall the multiplier bootstrap estimator in equation (\ref{m_bootstrap}) and the definition of the random variable $\xi$. 
By independence of $\xi$ and $W_{a}$, we have that
\[
E\left[\xi\left(\psi_{d,d^{\prime},a}\left(W_{a};\hat{v}_{a}\right)-\hat{\theta}_{d,d^{\prime},a}\right)\right]=E\left[\xi\right]E\left[\psi_{d,d^{\prime},a}\left(W_{a};\hat{v}_{a}\right)-\hat{\theta}_{d,d^{\prime},a}\right]=0,
\]
and therefore,  
\begin{align*}
	\sqrt{n}\left(\hat{\theta}_{d,d^{\prime},a}^{*}-\hat{\theta}_{d,d^{\prime},a}\right)
	& =\mathbb{G}_{n}\xi\left(\psi_{d,d^{\prime},a}\left(W_{a};\hat{v}_{a}\right)-\hat{\theta}_{d,d^{\prime},a}\right).
\end{align*}
Let $\hat{\boldsymbol{\theta}}_{a}^{*}$ denote a vector containing the multiplier bootstrap estimators $\hat{\theta}_{d,d^{\prime},a}^{*}$ for different $\left(d,d^{\prime}\right)\in \{0,1\}^{2}$. We write the
previous result in vector form as
\[
\sqrt{n}\left(\hat{\boldsymbol{\theta}}_{a}^{*}-\hat{\boldsymbol{\theta}}_{a}\right)=\mathbb{G}_{n}\xi\left(\boldsymbol{\psi}_{a}\left(W_{a};\hat{v}_{k,a}\right)-\hat{\boldsymbol{\theta}}_{a}\right),
\]
and let $Z_{n,P}^{*}:=\left(\mathbb{G}_{n}\xi\left(\boldsymbol{\psi}_{a}\left(W_{a};\hat{v}_{k,a}\right)-\hat{\boldsymbol{\theta}}_{a}\right)\right)_{a\in\mathcal{A}}$. Then we obtain the following result on the asymptotic behavior of the multiplier bootstrap.
\begin{theorem} If Assumptions 1 and 2 hold, the large sample law
	$Z_{P}$ of $Z_{n,P}$, can be consistently approximated by the bootstrap
	law $Z_{n,P}^{*}$: 
	\[
	Z_{n,P}^{*}\rightsquigarrow_{B}Z_{P}
	\]
	uniformly over $P\in\mathcal{P}_{n}$ in $l^{\infty}\left(\mathcal{A}\right)^{4}$.
\end{theorem} 

Let $\phi_{\tau}\left(F_{X}\right):=\inf\left\{ a\in\mathbb{R}:F_{X}\left(a\right)\geq\tau\right\} $
be the $\tau$th quantile function of a random variable $X$ whose c.d.f. is $F_{X}$. The von Mises expansion of $\phi_{\tau}\left(F_{X}\right)$
(p.292 in \citet{Vaart_1998}) is given by:
\[
\phi_{\tau}\left(E_{n}\right)-\phi_{\tau}\left(E\right)=\frac{1}{\sqrt{n}}\phi_{\tau,E}^{\prime}\left(\mathbb{G}_{n}\right)+\ldots+\frac{1}{m!}\frac{1}{n^{m/2}}\phi_{\tau,E}^{\left(k\right)}\left(\mathbb{G}_{n}\right)+\ldots,
\]
where $\phi_{\tau,E}^{\prime}\left(.\right)$ is a linear derivative
map and $\mathbb{G}_{n}$ denotes an empirical process:  $\mathbb{G}_{n}f\left(W\right)=\sqrt{n}\left(E_{n}f\left(W\right)-E\left[f\left(W\right)\right]\right)$. 

Let $\phi_{\boldsymbol{\theta}}^{\prime}:=\left(\phi_{\tau,\boldsymbol{\theta}}^{\prime}\right)_{\tau\in\mathcal{T}}$,
where $\boldsymbol{\theta}=\left(\boldsymbol{\theta}_{a}\right)_{a\in\mathcal{A}}$. Let $Q_{Y\left(d,M\left(d^{\prime}\right)\right)}^{0}\left(\tau\right):=\inf\left\{ a\in\mathbb{R}:F_{Y(d,M(d^{\prime}))}^{0}(a)\geq\tau\right\} $,
$\hat{Q}_{Y\left(d,M\left(d^{\prime}\right)\right)}\left(\tau\right):=\inf\left\{ a\in\mathbb{R}:\hat{\theta}_{d,d^{\prime},a}\geq\tau\right\} $
and   $\hat{Q}_{Y\left(d,M\left(d^{\prime}\right)\right)}^{*}\left(\tau\right):=\inf\left\{a\in\mathbb{R}:\hat{\theta}_{d,d^{\prime},a}^{*}\geq\tau\right\}$. Let $\mathbf{Q}^{0}_{\tau}$, $\hat{\mathbf{Q}}_{\tau}$ and $\mathbf{\hat{Q}}^{*}_{\tau}$ denote the corresponding vectors containing  $Q_{Y\left(d,M\left(d^{\prime}\right)\right)}^{0}\left(\tau\right)$, $\hat{Q}_{Y\left(d,M\left(d^{\prime}\right)\right)}\left(\tau\right)$ and $\hat{Q}_{Y\left(d,M\left(d^{\prime}\right)\right)}^{*}\left(\tau\right)$ over different $(d,d^{\prime})\in \{0,1\}^{2}$, respectively. 
We then obtain the following result of uniform validity for the estimation of quantiles, which can be proven by invoking the functional delta theorems (Theorems B.3 and B.4) of \citet{BCFH_2017}.
\begin{theorem} If Assumptions 1 and 2 hold, 
	\begin{align*}
		\sqrt{n}\left(\hat{\mathbf{Q}}_{\tau}-\mathbf{Q}_{\tau}^{0}\right)_{\tau\in\mathcal{T}} & \rightsquigarrow T_{P}:=\phi_{\boldsymbol{\theta}}^{\prime}\left(Z_{P}\right),\\
		\sqrt{n}\left(\mathbf{\hat{Q}}_{\tau}^{*}-\hat{\mathbf{Q}}_{\tau}\right)_{\tau\in\mathcal{T}} & \rightsquigarrow_{B}T_{P}:=\phi_{\boldsymbol{\theta}}^{\prime}\left(Z_{P}\right).
	\end{align*}
	uniformly over $P\in\mathcal{P}_{n}$ in $l^{\infty}\left(\mathcal{T}\right)^{4}$,
	where $\mathcal{T}\subset(0,1)$, $T_{P}$ is a zero mean tight Gaussian process for each $P\in\mathcal{P}_{n}$
	and $Z_{P}:=\left(\mathbb{G}_{P}\boldsymbol{\psi}_{a}\left(W_{a};v_{a}^{0}\right)\right)_{a\in\mathcal{A}}$.
\end{theorem} 

\section{Simulation}\label{sim}

\subsection{Simulation Design}
This section presents a simulation study to examine the finite sample performance of
the proposed DML estimators in equations (\ref{theta_hat}) and (\ref{theta_hat1}).
We consider the following data-generating process for the observed covariates $X=\left(X_{1},X_{2},X_{3}\right)$, where 
\begin{eqnarray*}
	X_{1} & = & 0.75V_{1}+0.1V_{2}+0.15V_{3},\\
	X_{2} & = & 0.15V_{1}+0.7V_{2}+0.15V_{3},\\
	X_{3} & = & 0.14V_{1}+0.08V_{2}+0.78V_{3},
\end{eqnarray*}
with $V_{1}$, $V_{2}$ and $V_{3}$ being i.i.d.\ random variables following a chi-squared distribution 
with 1 degree of freedom. The binary treatment variable
$D$ is generated based on the following model:
\begin{eqnarray*}
	D & = & 1\left\{ 0.371+X^{\top}\left(0.198,0.125,-0.323\right)+\varepsilon_{D}>0\right\} .
\end{eqnarray*}
For the binary mediator $M$, the data-generating process is 
\[
M=1\left\{ -0.070+0.710D+X^{\top}\left(-0.054,-0.482,0.299\right)+\varepsilon_{M}>0\right\} .
\]
The model for outcome $Y$ is defined as follows: 
\begin{eqnarray*}
	h\ensuremath{\left(D,M,X\right)} & = & 0.766+0.458D+0.836DM+0.383M+X^{\top}\left(0.640,0.260,0.474\right),\\
	Y & = & h\left(D,M,X\right)^{-1}\varepsilon_{Y}.
\end{eqnarray*}
The error terms $\left(\varepsilon_{Y},\varepsilon_{D},\varepsilon_{M}\right)$
are mutually independent standard normal random variables and also independent
of $X$. Analytically computing the unconditional c.d.f. and quantiles of the
potential outcome $Y\left(d,M\left(d^{\prime}\right)\right)$ is difficult for the data generating process considered. For this reason, we use a Monte Carlo simulation to approximate the true values of the c.d.f. and the quantiles. We draw 40 million observations of $\left(\varepsilon_{Y},\varepsilon_{M},V_{1},V_{2},V_{3}\right)$
from their respective true distributions, and for each observation, we calculate the corresponding potential outcomes $Y\left(d,M\left(d^{\prime}\right)\right)$ for $(d,d^{\prime})\in \{0,1\}^{2}$. 
In the next step, we evaluate profiles of the empirical c.d.f.'s and quantiles of the 40 million sampled potential outcomes and use the evaluated profiles as approximations to their true profiles. In Figure \ref{figure3} in the appendix,
the upper panel shows the approximate true profiles of the c.d.f.'s $F_{Y\left(d,M\left(d^{\prime}\right)\right)}\left(a\right)$, while the lower panel provides the approximate true profiles
of quantiles $Q_{Y\left(d,M\left(d^{\prime}\right)\right)}(\tau)$ for $(d,d^{\prime})\in \{0,1\}^{2}$. Figure \ref{figure4} in the appendix depicts the approximate true profiles of NDQTE (NDQTE'), NIQTE (NIQTE') and TQTE across quantiles. 

In our simulation design, the nuisance parameters have the following functional forms:
\begin{eqnarray*}
	F_{Y|D,M,X}\left(a|D,M,X\right) & = & \Phi\left(\beta_{0,a}+\alpha_{0.a}D+\alpha_{1,a}DM+\beta_{1,a}M+X^{\top}\boldsymbol{\beta}_{2,a}\right)\\
	f_{D|X}\left(D=1|X\right) & = & \Phi\left(\lambda_{0}+X^{\top}\boldsymbol{\lambda}_{1}\right)\\
	f_{M|D,X}\left(M=1|D,X\right) & = & \Phi\left(b_{0}+a_{0}D+X^{\top}\boldsymbol{b}_{1}\right),
\end{eqnarray*}
where $\Phi\left(.\right)$ is the c.d.f.\ of a standard normal random
variable. 
The vector of parameters satisfies $\left(\beta_{0,a},\alpha_{0,a},\alpha_{1,a},\beta_{1,a},\boldsymbol{\beta}_{2,a}^{\top}\right)=a\times\left(\beta_{0},\alpha_{0},\alpha_{1},\beta_{1},\boldsymbol{\beta}_{2}^{\top}\right)$.
When running the simulations, we also include a set of auxiliary (exogenous) variables:
$X^{aug}:=\left(X_{1}^{aug},X_{2}^{aug},\ldots,X_{J}^{aug}\right)$,
$X_{j}^{aug}=U_{1j}\left(V_{1}+V_{2}+V_{3}\right)+U_{2j}\left(Z_{j}\right)^{2}$,
$j=1,\ldots,J$, where $U_{1j}\sim i.i.d.U\left(0,0.2\right)$, $U_{2j}\sim i.i.d.U\left(0.8,1\right)$.
$\mathbf{Z}:=\left(Z_{1},Z_{2},\ldots,Z_{J}\right)$ follow a multivariate
normal distribution with a mean vector $\mathbf{0}$ and a covariance
matrix with elements $0.5^{\left|j-l\right|},$ $j,l=1,\ldots,J$.
$V_{1},V_{2},V_{3},U_{1j},U_{2j}$, $\mathbf{Z}$ and the error terms
$\left(\varepsilon_{Y},\varepsilon_{D},\varepsilon_{M}\right)$ are
mutually independent. Depending on the realized values of $U_{1j}$
and $U_{2j}$, the correlations $cor\left(X_{1},X_{j}^{a}\right)$, $cor\left(X_{2},X_{j}^{a}\right)$,
$cor\left(X_{3},X_{j}^{a}\right)$ vary, and
on average they amount to 0.139, 0.147 and 0.135, respectively. 

\subsection{The Post Lasso Estimator}
Let $W_{a,i}^{aug}=\left(Y_{i},D_{i},M_{i},X_{i},X_{i}^{aug}\right)$
denote the $i$th observation of the simulated data. When applying K-fold cross-fitting to $W_{a,i}^{aug}$, we estimate the models of the nuisance parameters based on post-lasso regression: we first estimate the models by lasso regression and then re-estimate the models without (lasso) penalization when including only those regressors with non-zero coefficients in the respective previous lasso steps. We denote 
the lasso estimator of the coefficients of $F_{Y|D,M,X}\left(a|D,M,X\right)$
in K-fold cross-fitting by 
\begin{equation}
	\hat{\boldsymbol{\gamma}}_{Y,a}\in\arg\min_{\boldsymbol{\gamma}_{Y,a}\in\mathbb{R}^{p}}\frac{1}{\left|I_{k}^{c}\right|}\sum_{i\in I_{k}^{c}}L\left(W_{a,i}^{aug};\boldsymbol{\gamma}_{Y,a}\right)+\frac{\gamma}{\left|I_{k}^{c}\right|}\left\Vert \hat{\varPsi}\boldsymbol{\gamma}_{Y,a}\right\Vert _{1},\label{lasso}
\end{equation}
where $p$ is the number of covariates, $\left|I_{k}^{c}\right|$
is the number of observations in the complement set $I_{k}^{c}$, $\gamma$ is the penalty
parameter, $\left\Vert .\right\Vert _{1}$ denotes the $l_{1}$ norm
and $\hat{\varPsi}$ is a diagonal matrix of penalty loadings. Here,
the loss function $L\left(.\right)$ corresponds to that in equation (\ref{MLE})
and the link function $G_{a}\left(.\right)$ is $\Phi\left(.\right)$.
When solving the lasso estimation problem of expression (\ref{lasso}), we only
impose a penalty on $\left(X,X^{aug}\right)$ and
therefore, the first four diagonal elements (for the intercept term,
$D$, $M$ and $DM$) of $\hat{\varPsi}$ are ones, while the remaining diagonal
elements are zeros. The value of the penalty parameter is determined based on the procedure outlined in \citet{BCFH_2017}.
The other nuisance parameters $f_{D|X}\left(d|X\right)$
and $f_{M|D,X}\left(m|D,X\right)$ are estimated in an analogous way and we denote by $\hat{\boldsymbol{\gamma}}_{D}$ and $\hat{\boldsymbol{\gamma}}_{M}$
their corresponding lasso estimators. Let $\tilde{\Xi}$ denote
the union of variables in $\left(X,X^{aug}\right)$ with non-zero lasso coefficient estimates in one or several lasso regressions of the three nuisance parameters, with $\tilde{\Xi}\subseteq\text{supp}\left(\hat{\boldsymbol{\gamma}}_{Y,a}\right)\cup\text{supp}\left(\hat{\boldsymbol{\gamma}}_{D}\right)\cup\text{supp}\left(\hat{\boldsymbol{\gamma}}_{M}\right)$.
The post-lasso estimator of $F_{Y|D,M,X}\left(a|D,M,X\right)$
is defined as 
\begin{equation}
	\tilde{\boldsymbol{\gamma}}_{Y,a}\in\arg\min_{\boldsymbol{\gamma}_{Y,a}\in\mathbb{R}^{p}}\frac{1}{\left|I_{k}^{c}\right|}\sum_{i\in I_{k}^{c}}L\left(W_{a,i}^{aug};\boldsymbol{\gamma}_{Y,a}\right):\text{supp}\left(\boldsymbol{\gamma}_{Y,a}\right)\subseteq\text{supp}\left(\hat{\boldsymbol{\gamma}}_{Y,a}\right)\cup\tilde{\Xi}.\label{post_lasso_Y_a}
\end{equation}
The post-lasso estimators of $f_{D|X}\left(d|X\right)$
and $f_{M|D,X}\left(m|D,X\right)$ are obtained analogously. Based on the post-lasso estimates of the coefficients, we estimate the nuisance parameters among observations $W_{a,i}^{aug}$, $i\in I_{k}$, and use them to compute estimators (\ref{K_fold_theta}) or (\ref{K_fold_theta1}). 

When using the estimator based on equation (\ref{theta_hat}), we approximate the nuisance
parameters $f_{D|M,X}\left(d|M,X\right)$ by a probit model $\Phi\left(\lambda_{2}+\lambda_{3}M+X^{\top}\boldsymbol{\lambda}_{4}\right)$. The post-lasso approach for estimating $f_{D|M,X}\left(d|M,X\right)$ is the same as before. To estimate $E\left[F_{Y|D,M,X}\left(a|D,M,X\right)|d^{\prime},X\right]$,
we approximate $E\left[F_{Y|D,M,X}\left(a|D,M,X\right)|D,X\right]$
by a linear model $\beta_{3}+\beta_{4}D+X^{\top}\boldsymbol{\beta}_{5}$.
We calculate post-lasso estimates of $F_{Y|D,M,X}\left(a|D,M,X\right)$
among observations in the complement set, $W_{a,i}^{aug}$ with $i\in I_{k}^{c}$, and estimate $\left(\beta_{3},\beta_{4},\boldsymbol{\beta}_{5}\right)$
by linearly regressing these estimates on $D$ and those covariates previously selected for computing the post-lasso estimate of $F_{Y|D,M,X}\left(a|D,M,X\right)$. We then use the linear regression coefficients coming from the complement set to make cross-fitted predictions among observations with indices $i\in I_{k}$ and $D_{i}=d^{\prime}$, which serve as estimates of $E\left[F_{Y|D,M,X}\left(a|D,M,X\right)|d^{\prime},X\right]$. 

\subsection{Simulation Results}

To evaluate the performance of the proposed DML estimators of the c.d.f.'s of the potential outcomes across grid values $a$, 
 we calculate integrated mean squared error (IMSE)
and integrated Anderson--Darling weighted MSE (IWMSE) for each simulation: 
\begin{eqnarray}
	\text{IMSE} & = & \int_{a\in\mathcal{A}}\left[\hat{F}_{Y\left(d,M\left(d^{\prime}\right)\right)}\left(a\right)-F_{Y\left(d,M\left(d^{\prime}\right)\right)}\left(a\right)\right]^{2}dF_{Y\left(d,M\left(d^{\prime}\right)\right)}\left(a\right),\label{ISE}\\
	\text{IWMSE} & = & \int_{a\in\mathcal{A}}\frac{\left[\hat{F}_{Y\left(d,M\left(d^{\prime}\right)\right)}\left(a\right)-F_{Y\left(d,M\left(d^{\prime}\right)\right)}\left(a\right)\right]^{2}}{F_{Y\left(d,M\left(d^{\prime}\right)\right)}\left(a\right)\left(1-F_{Y\left(d,M\left(d^{\prime}\right)\right)}\left(a\right)\right)}dF_{Y\left(d,M\left(d^{\prime}\right)\right)}\left(a\right).\label{IWISE}
\end{eqnarray}
To assess the performance of the estimators of the quantiles of the potential outcomes, $Q_{Y\left(d,M\left(d^{\prime}\right)\right)}$, we compute the integrated absolute error
(IAE) across ranks $\tau$: 
\begin{equation}
	\text{IAE}  =  \frac{1}{\left|\mathcal{T}\right|}\sum_{\tau\in\mathcal{T}}\left|\hat{Q}_{Y\left(d,M\left(d^{\prime}\right)\right)}\left(\tau\right)-Q_{Y\left(d,M\left(d^{\prime}\right)\right)}\left(\tau\right)\right|,\label{IAE}
\end{equation}
where $\mathcal{T}$ is the grid of ranks, which we set to $(0.05,0.06,\ldots,0.95)$. Furthermore, we calculate the IAE for the estimators of the quantile treatment effects defined in equations (\ref{NDQTE}) to (\ref{TQTE}). In the simulations, we set $K=3$ for 3-fold cross-fitting. The number of auxiliary variables is $J=250$ and we consider sample sizes 2,500, 5,000 and 10,000 observations in the simulations. The reported performance
measures are averages for $\left(d,d^{\prime}\right)\in\{0,1\}^{2}$ over 1,000 simulations.

\begin{table}
	\centering \caption{Integrated mean squared error (IMSE) and integrated Anderson Darling
		weighted MSE (IWMSE) for estimated potential outcome distributions 
  }
	\begin{tabular}{lccccccc}
		\hline 
		& \multicolumn{7}{c}{$\hat{\theta}_{d,d^{\prime},a}$}\tabularnewline
		\cline{2-8} \cline{3-8} \cline{4-8} \cline{5-8} \cline{6-8} \cline{7-8} \cline{8-8} 
		&  & IMSE  &  &  &  & IWMSE  & \tabularnewline
		\cline{2-4} \cline{3-4} \cline{4-4} \cline{6-8} \cline{7-8} \cline{8-8} 
		& 2,500  & 5,000  & 10,000  &  & 2,500  & 5,000  & 10,000\tabularnewline
		\hline 
		$F_{Y\left(1,M\left(1\right)\right)}$  & 0.114  & 0.056  & 0.027  &  & 0.640  & 0.315  & 0.149\tabularnewline
		$F_{Y\left(1,M\left(0\right)\right)}$  & 0.159  & 0.079  & 0.038  &  & 0.891  & 0.443  & 0.213\tabularnewline
		$F_{Y\left(0,M\left(1\right)\right)}$  & 0.314  & 0.130  & 0.062  &  & 1.555  & 0.661  & 0.315\tabularnewline
		$F_{Y\left(0,M\left(0\right)\right)}$  & 0.211  & 0.100  & 0.047  &  & 1.060  & 0.506  & 0.242\tabularnewline
		\hline 
		& \multicolumn{7}{c}{$\hat{\theta}_{d,d^{\prime},a}^{\prime}$}\tabularnewline
		\cline{2-8} \cline{3-8} \cline{4-8} \cline{5-8} \cline{6-8} \cline{7-8} \cline{8-8} 
		&  & IMSE  &  &  &  & IWMSE  & \tabularnewline
		\cline{2-4} \cline{3-4} \cline{4-4} \cline{6-8} \cline{7-8} \cline{8-8} 
		& 2,500  & 5,000  & 10,000  &  & 2,500  & 5,000  & 10,000\tabularnewline
		\hline 
		$F_{Y\left(1,M\left(1\right)\right)}$  & 0.127  & 0.057  & 0.028  &  & 0.690  & 0.316  & 0.157\tabularnewline
		$F_{Y\left(1,M\left(0\right)\right)}$  & 0.169  & 0.078  & 0.037  &  & 0.927  & 0.437  & 0.209\tabularnewline
		$F_{Y\left(0,M\left(1\right)\right)}$  & 0.267  & 0.130  & 0.066  &  & 1.360  & 0.661  & 0.336\tabularnewline
		$F_{Y\left(0,M\left(0\right)\right)}$  & 0.214  & 0.107  & 0.051  &  & 1.080  & 0.539  & 0.260\tabularnewline
		\hline
	\end{tabular}\label{table1} 
\end{table}

Table \ref{table1} reports the results for the IMSE and IWMSE (scaled by 1,000), Table \ref{table2} those for the IAE. All the performance measures behave rather favorably. As the sample size increases, the performance measures (and thus, estimation errors) decline sharply. However, for different combinations of $(d,d^{\prime})$, the levels of the performance measures are different, especially when the sample size is small. Estimation errors are significantly larger if $(d,d^{\prime}) = (0,1)$ and $(0,0)$, rather than $(d,d^{\prime})=(1,1)$ and $(1,0)$. This is also reflected by the performance measures of NDQTE (NDQTE') and TQTE, which point to higher errors than those of NIQTE (NIQTE'). When comparing the performance measures of the two estimators based on equations (\ref{theta_hat}) and (\ref{theta_hat1}), we find some differences in their levels when the sample size is small. However, the differences vanish as the sample size increases, which suggests that the two estimators perform equally well asymptotically in the simulation design considered. 

\begin{table}
	\centering \caption{Integrated absolute error (IAE) for estimated quantiles of potential outcomes and quantile treatment effects 
  }
	\begin{tabular}{lccccccc}
		\hline 
		&  & $\hat{\theta}_{d,d^{\prime},a}$ &  &  &  & $\hat{\theta}_{d,d^{\prime},a}^{\prime}$ & \tabularnewline
		\cline{2-4} \cline{3-4} \cline{4-4} \cline{6-8} \cline{7-8} \cline{8-8} 
		& 2,500  & 5,000  & 10,000  &  & 2,500  & 5,000  & 10,000 \tabularnewline
		\hline 
		$Q_{Y\left(1,M\left(1\right)\right)}$  & 0.010  & 0.007  & 0.005  &  & 0.010  & 0.007  & 0.005 \tabularnewline
		$Q_{Y\left(1,M\left(0\right)\right)}$  & 0.014  & 0.010  & 0.007  &  & 0.014  & 0.010  & 0.007 \tabularnewline
		$Q_{Y\left(0,M\left(1\right)\right)}$  & 0.028  & 0.022  & 0.017  &  & 0.028  & 0.022  & 0.018 \tabularnewline
		$Q_{Y\left(0,M\left(0\right)\right)}$  & 0.031  & 0.026  & 0.022  &  & 0.031  & 0.026  & 0.022 \tabularnewline
		&  &  &  &  &  &  & \tabularnewline
		NDQTE  & 0.034 & 0.027 & 0.023 &  & 0.034 & 0.028 & 0.023\tabularnewline
		NDQTE' & 0.030 & 0.023 & 0.018 &  & 0.029 & 0.023 & 0.018\tabularnewline
		NIQTE  & 0.008 & 0.006 & 0.004 &  & 0.008 & 0.006 & 0.004\tabularnewline
		NIQTE' & 0.016 & 0.012 & 0.010 &  & 0.016 & 0.012 & 0.010\tabularnewline
		TQTE  & 0.033 & 0.026 & 0.022 &  & 0.033 & 0.027 & 0.023\tabularnewline
		\hline 
	\end{tabular}\label{table2} 
\end{table}

\section{Empirical Application}\label{appl}

\subsection{The Job Corps Data}

We apply the proposed estimators of natural direct and indirect quantile treatment effects to data from the National Job Corps Study, in order to evaluate the impact of the Job Corps (JC) training program on earnings of young individuals with disadvantaged backgrounds. JC is the largest and most comprehensive job training program for disadvantaged youth in the US. It provides participants with vocational training and/or classroom education, housing, and board over an average duration of 8 months. Participants also receive health education as well as health and dental care. \citet{SchBuGl01} and \citet{SchBuMc08} assess the average effects of random assignment to JC on several labor market outcomes and find it to increase education, employment, and earnings in the longer run. Other contributions evaluate more specific aspects or components of JC, like the average effect of the time spent in training or of particular training sequences on employment and earnings, see e.g.\ \citet{FlGoNe12} and \citet{BoHuLa2020}. 

Furthermore, several studies conduct mediation analyses to assess the average direct and indirect effects of program participation. \citet{FlFl09} and \citet{Hu14} consider work experience or employment as mediators, respectively, and find positive direct effects of JC on earnings and general health, respectively, when invoking a selection-on-observables assumption. \citet{FlFl10} avoid the latter assumption based on a partial identification approach based on which they compute upper and lower bounds for the causal mechanisms of JC when considering the achievement of a GED, high school degree, or vocational degree as mediators. Under their strongest set of bounding assumptions, they find a positive direct effect on labor market outcomes, net of the indirect mechanism via obtaining a degree. \citet{FrHu17} base their mediation analysis on separate instrumental variables for the treatment and the mediator and find a positive indirect effect of training on earnings through an increase in the number of hours worked.  We contribute to the causal mediation literature on the effectiveness of the JC program by considering quantile treatment effects across different ranks of the potential outcome distributions, which provides more insights on effect heterogeneity than the evaluation of average effects.

For our empirical analysis, we consider the JC data provided in the \texttt{causalweight} package by \citet{BH_2022} for the statistical software \texttt{R}, which is a processed data set with 9,240 observations that contains a subset of the variables available in the original National Job Corps Study. Our outcome of interest is weekly earnings in the third year after the assignment (the variable \texttt{earny3} in the \texttt{JC} data frame), while the treatment is a binary indicator for participation in any (classroom-based or vocational) training in the first year after program assignment (\texttt{trainy1}). We aim at assessing whether training directly affects the earnings outcome, and whether it also has an indirect effect by affecting health. For this reason, we consider general health one year after program assignment (\texttt{health12}) as mediator, a categorical variable ranging from 1 (excellent health) to 4 (poor health). The motivation is that participation in training aimed at increasing human capital and labor market perspectives may have an impact on mental health, which in turn may affect labor market success. Furthermore, JC might also affect physical health through health education and health/dental care, which can influence labor market success, too. For this reason, we aim at disentangling the direct earnings effect of training and its indirect effect operating via health. 

\begin{table}
	\centering \caption{Estimates of Average Effects}
	\begin{tabular}{ccccccc}
		\hline 
		 & TE & NDE & NIE & NDE' & NIE'\tabularnewline
		\hline 
		Effect   & 16.591 & 16.995 & -0.403 & 16.586 & 0.005\tabularnewline
		Std.err  & 3.740 & 3.747 & 0.190 & 3.770 & 0.553\tabularnewline
		p-value  & 0.000 & 0.000 & 0.034 & 0.000 & 0.992\tabularnewline
		\hline 
	\end{tabular}\label{table4} 
\end{table}

The data set also contains 28 pre-treatment covariates, which include socio-economic information such as a study participant's gender, age, ethnicity, (own) education and parents' education, mother tongue, marital status, household size, previous employment, earnings and welfare receipt, health status, smoking behavior, alcohol consumption, and whether a study participant has at least one child. We assume that sequential ignorability of the treatment and the mediator holds conditional on these observed characteristics, implying that the permit controlling for any factors jointly affecting training participation and the earnings outcome, training participation and health 12 months after assignment, or health and earnings. To make lasso-based estimation of the nuisance parameters in our DML approach more flexible, we create interaction terms between all of the 28 covariates and squared terms for any non-binary covariates. This entails a total of 412 control variables that include both the original covariates and the higher order/interaction terms which we include in our DML approach. Table \ref{table3} provides summary statistics for the outcome, the treatment, the mediator and the covariates. 

\subsection{Effect Estimates}

Before considering quantile treatment effects, we first estimate the average direct and indirect effects by a K-fold cross-fitting estimator based on Theorem 2 in \cite{FHLLS_2022}, as implemented in the \texttt{causalweight} package for \texttt{R}.
Table \ref{table4} reports the estimated average total effect (TE) of training, the average natural direct effects (NDE and NDE') and the average natural indirect effects (NIE and NIE') operating via general health.  
The TE estimate (Effect) suggests that participation in JC increases average weekly earnings in the third year by roughly 16 to 17 USD. As the estimated mean potential outcome under non-treatment amounts to approximately 161 USD, the program increases weekly earnings by roughly 10\% according to our estimate. The TE is highly statistically significant as the standard error (Sdt.err) of 3.740 is rather low relative to the effect estimate, such that p-value that is close to zero. 

The total effect seems to be predominantly driven by the direct impact of training on earnings, as both NDE and NDE' are of similar magnitude as TE and highly statistically significant. In contrast, the indirect effect under non-treatment ($d=0$), NIE', is close to zero and insignificant, while that under treatment ($d=1$), NIE, amounts to -0.403 USD and is statistically significant at the 5\% level. Bearing in mind that the health mediator is inversely coded (a smaller value implies better health), this negative estimate suggests a positive average indirect effect of training participation on earnings under treatment, which is, however, rather modest. Furthermore, the  effect heterogeneity across NIE and NIE' points to moderate interaction effects of the treatment and the mediator: the impact of health on earnings appears to be somewhat more important under training than without training.     

\begin{figure}[H]
	\centering
	\mbox{
		\subfigure{\includegraphics[height=5cm, width = 8cm]{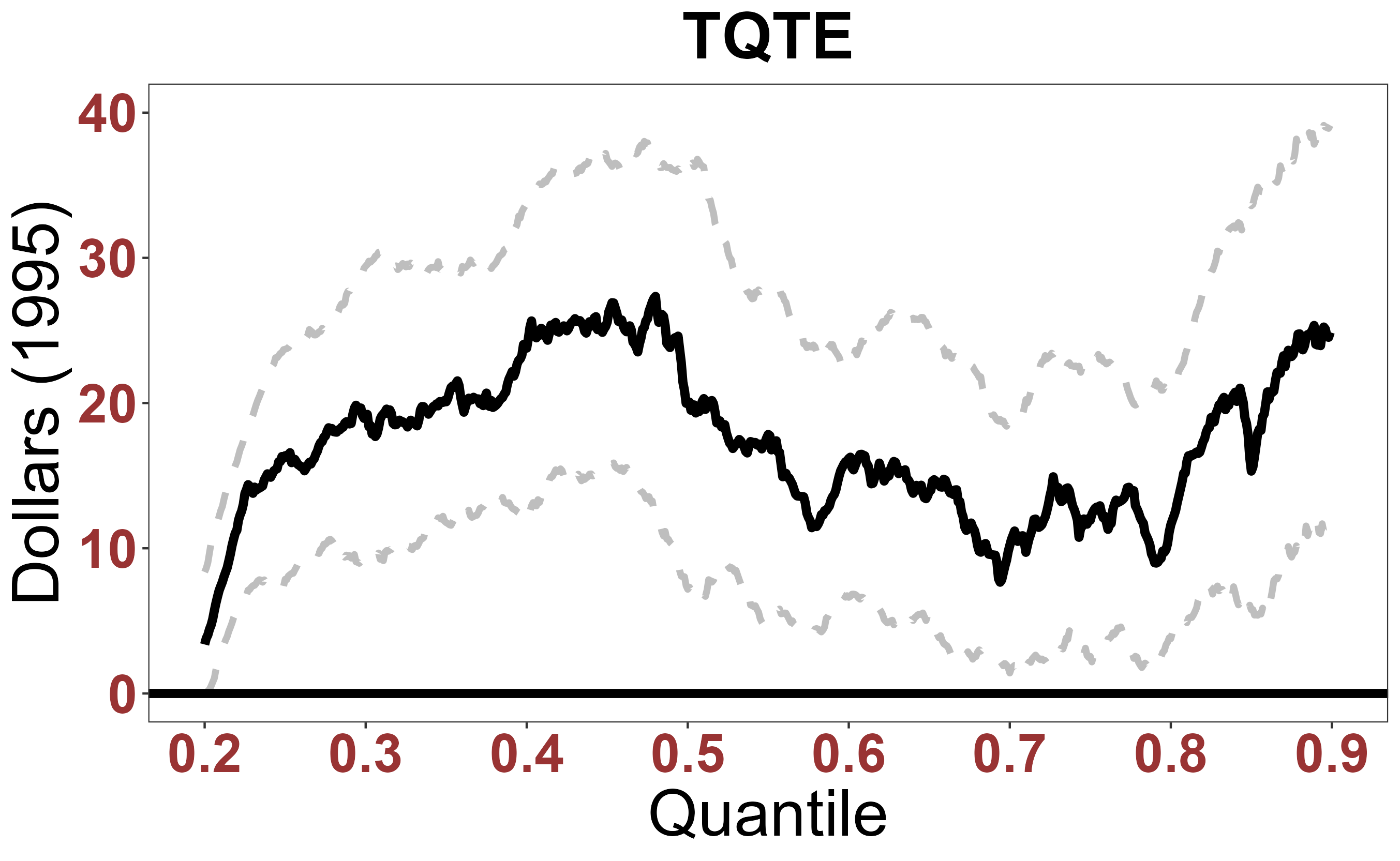}}
	}
	\caption{Estimates of the TQTE across ranks 0.2 to 0.9 (solid lines), based on inverting $\hat{\theta}_{d,d^{\prime},a}$. 95\% confidence intervals (dashed lines) are based on the multiplier bootstrap.}
	\label{figure5}
\end{figure}

The average effects might mask interesting effect heterogeneity across ranks of the earnings distribution. For this reason, 
we estimate the total quantile treatment effect (TQTE), natural direct quantile treatment effects (NDQTE and NDQTE') and natural indirect quantile treatment effects (NIQTE and NIQTE') across ranks ($\tau$) 0.2 to 0.9. To this end, we invert our K-fold cross-fitting estimator $\hat{\theta}_{d,d^{\prime},a}$ of equation (\ref{theta_hat}) and estimate the nuisance parameters by post-lasso regression as outlined in Section 4.2. Figures \ref{figure5} and \ref{figure6} depict the estimates of the causal effects (on the y-axis) across $\tau$ (on the x-axis), which correspond to the solid lines in the respective graphs. The dashed lines provide the 95\% confidence intervals based on the multiplier bootstrap introduced in Section 2.5. 

The quantile treatment effects are by and large in  line with the average treatment effects. TQTE, NDQTE and NDQTE' are statistically significantly positive at the 5\% across almost all ranks $\tau$ considered and generally quite similar to each other. In contrast, all of the NIQTE estimates (the indirect effects under $d=1$) are relatively close to zero and statistically insignificant. The majority of the NIQTE' estimates (the indirect effects under $d=0$) are not statistically significantly different from zero either. However, several of the negative effects measured at lower ranks (roughly between the 0.2th and 0.4th quantiles) are marginally statistically significant and point to an earnings-increasing indirect effect under non-treatment (due to inverse coding of the health mediator). This potentially interesting pattern is averaged out when considering NIE' (the average indirect effect under $d=0$), which we found to be virtually zero and insignificant, see Table \ref{table4}. Finally, the non-monotonic shape of the point estimates of TQTE, NDQTE and NDQTE' across ranks $\tau$ suggests heterogeneous effects at different quantiles of the potential earnings distributions. At the same time, the width of the confidence intervals suggests that the null hypothesis of homogeneous effects cannot be rejected for most of the quantiles considered. 

\begin{figure}[H]
	\centering
	\mbox{
		\subfigure{\includegraphics[height=5.5cm, width = 8cm]{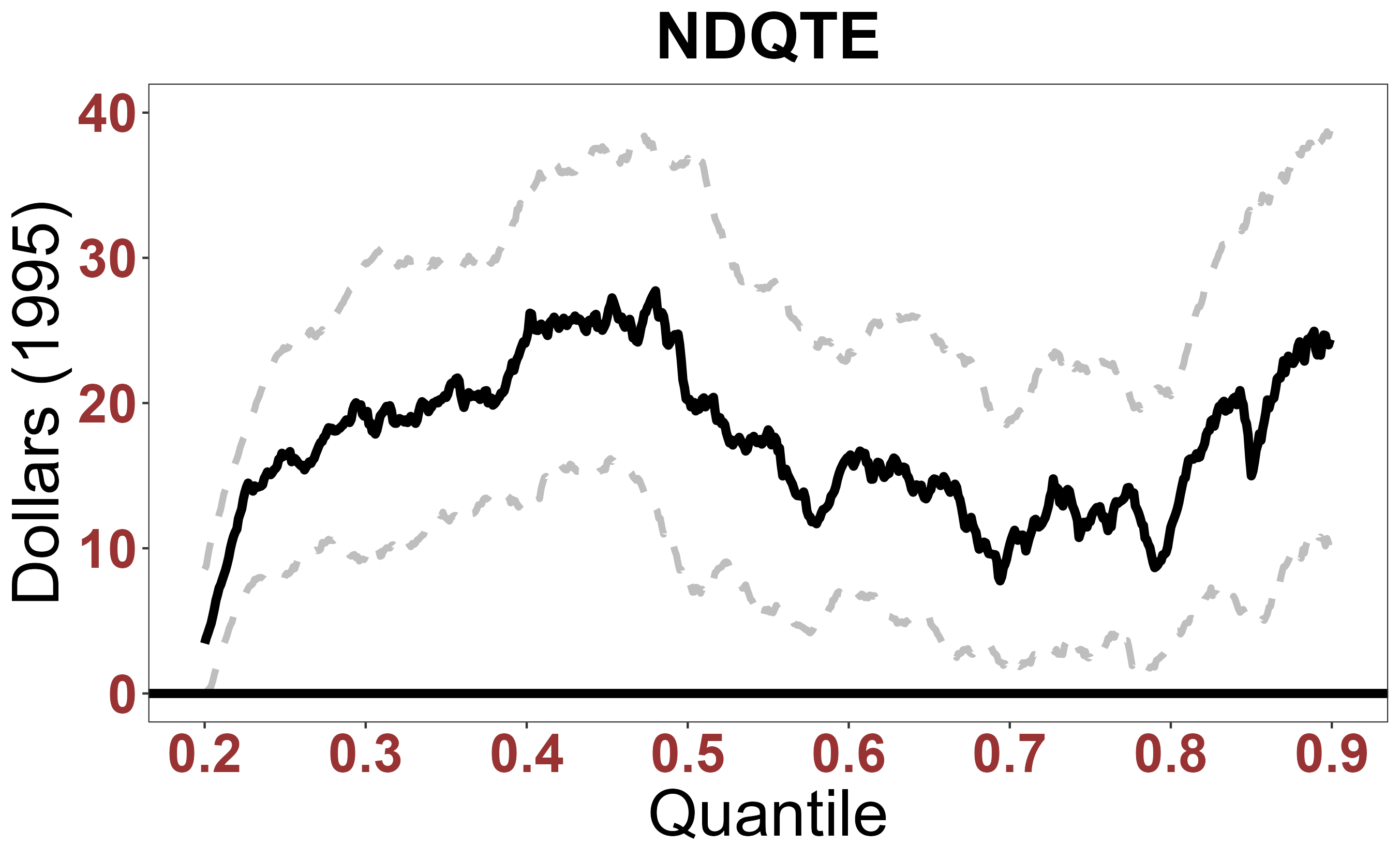}}
            \subfigure{\includegraphics[height=5.5cm, width = 8cm]{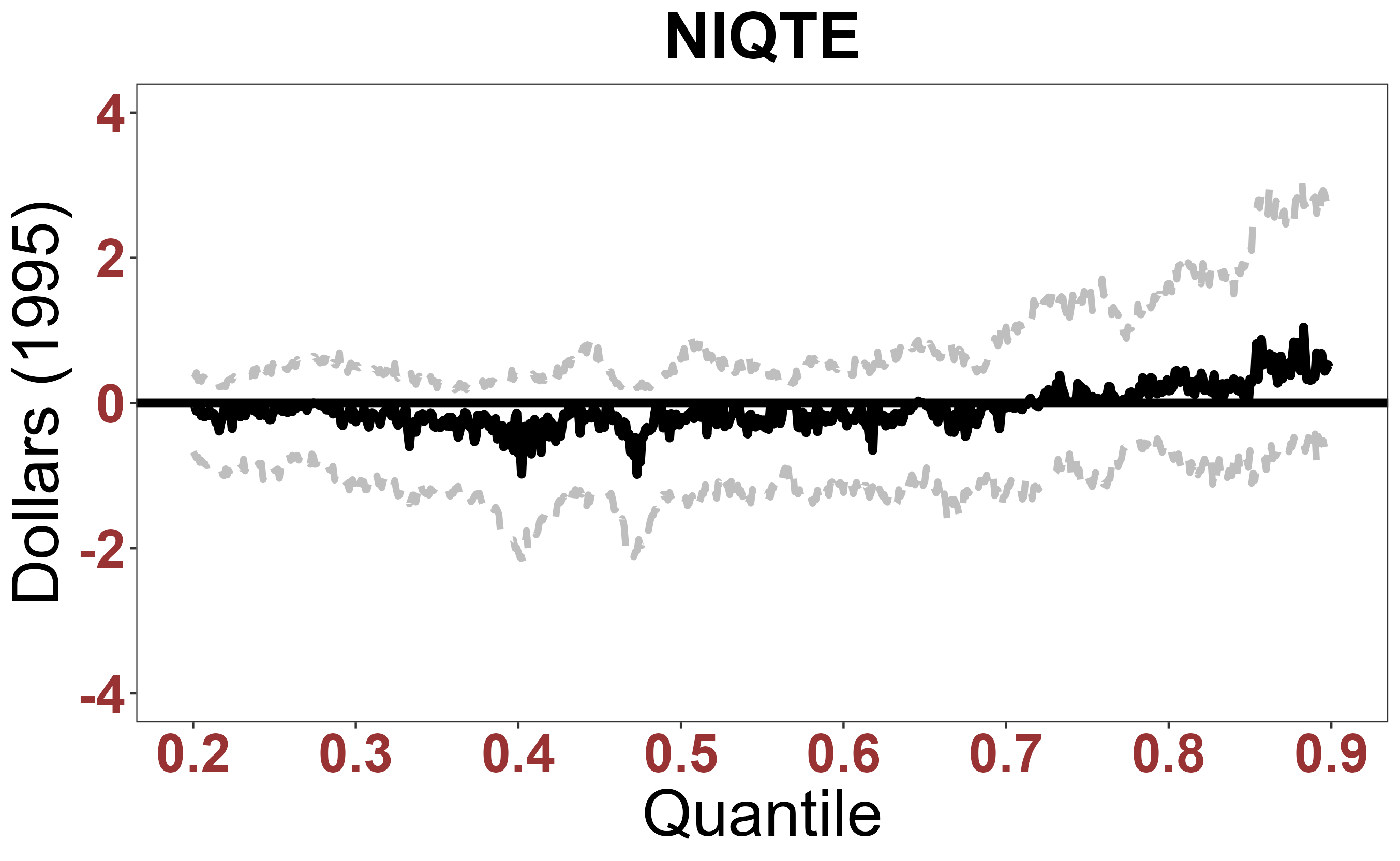}}
	}	
	\mbox{
			\subfigure{\includegraphics[height=5.5cm, width = 8cm]{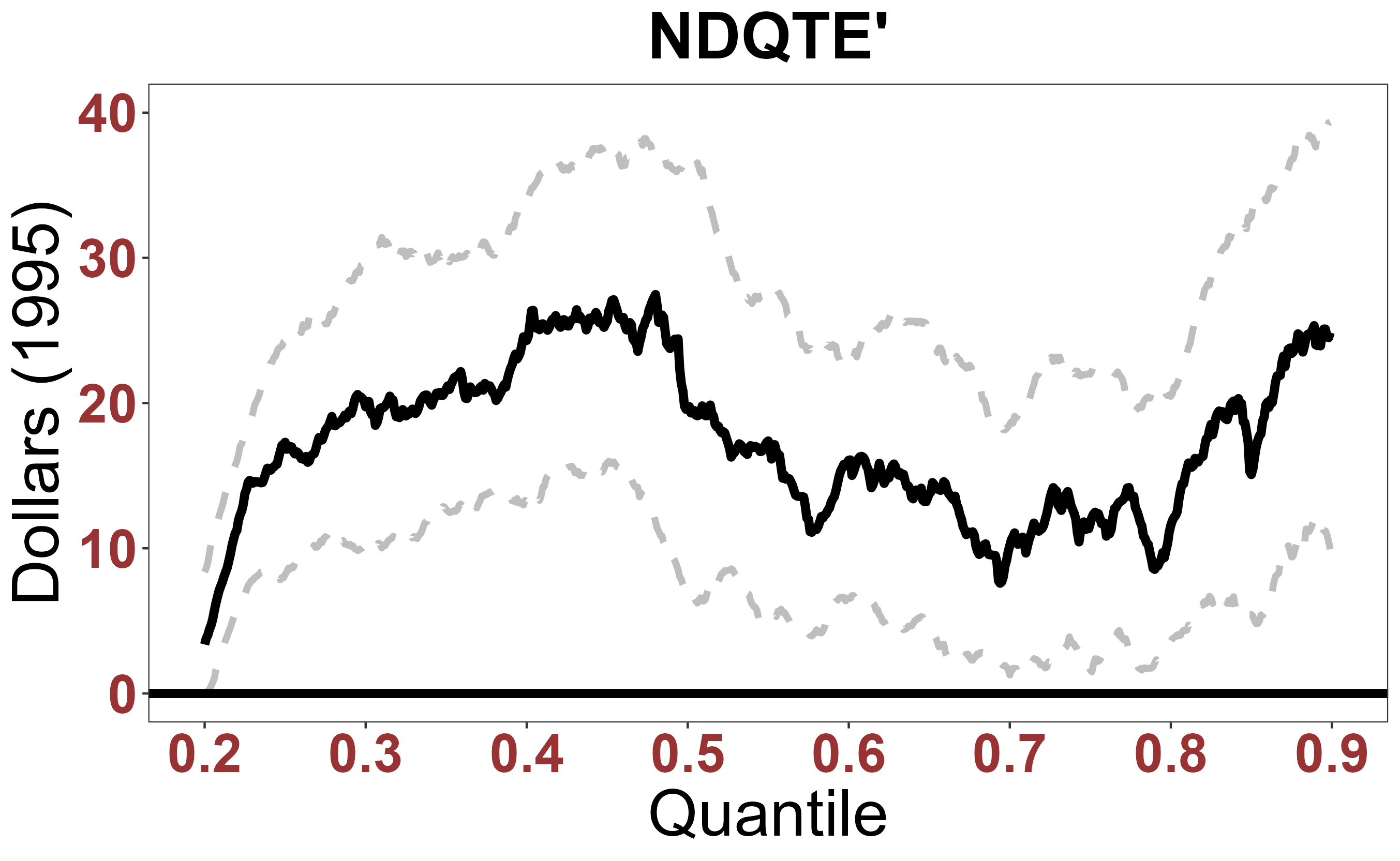}}
		\subfigure{\includegraphics[height=5.5cm, width = 8cm]{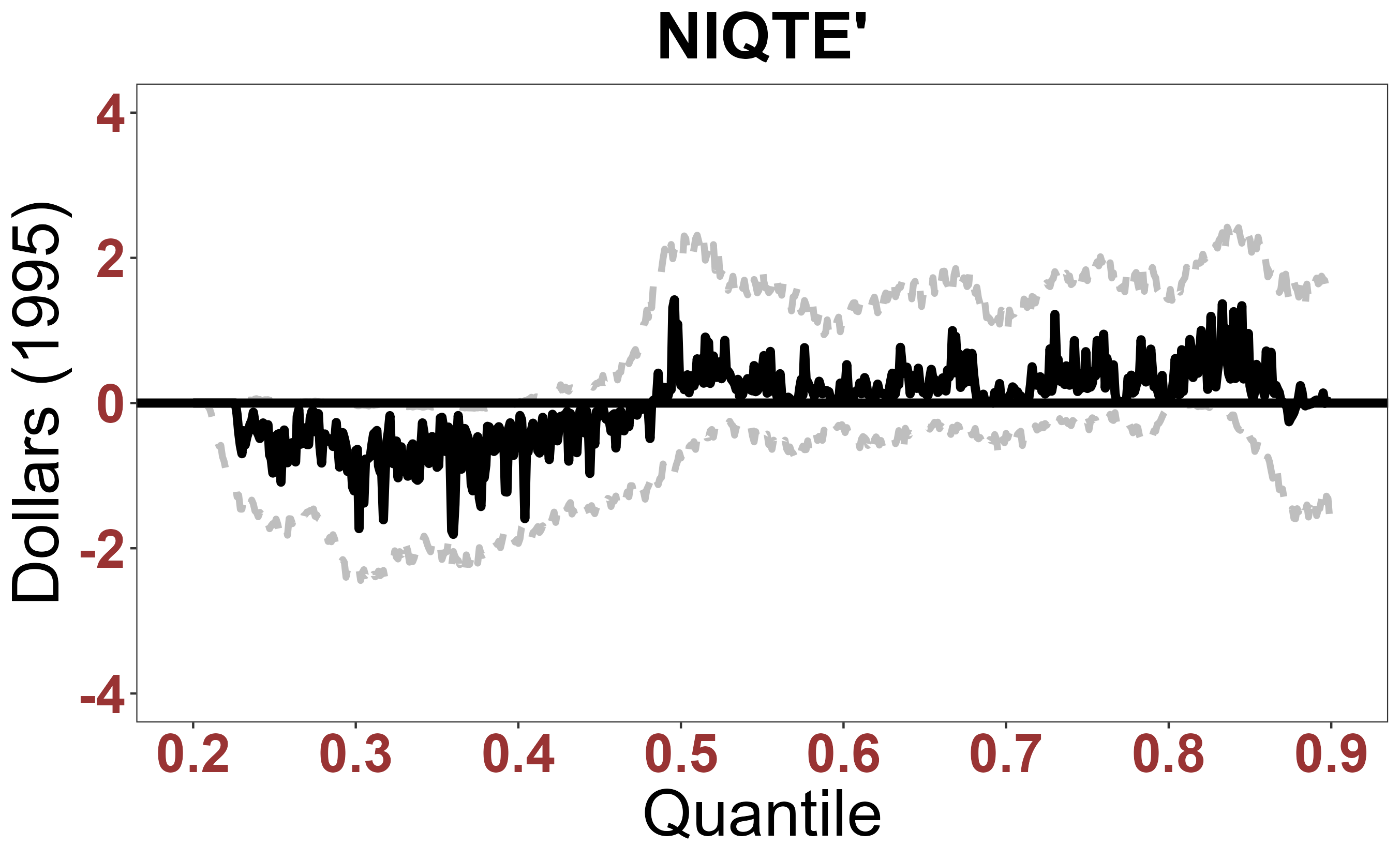}}
	}		
	\caption{Estimates of the NDQTE, NIQTE, NDQTE' and NIQTE' across ranks 0.2 to 0.9 (solid lines), based on inverting $\hat{\theta}_{d,d^{\prime},a}$. 95\% confidence intervals (dashed lines) are based on the multiplier bootstrap.}
	\label{figure6}
\end{figure}

\section{Conclusion}\label{conclusion}

We proposed a DML approach for estimating natural direct and indirect quantile treatment effects under a sequential ignorability assumption. The method relies on the efficient score functions of the potential outcomes' cumulative distributional functions, which are inverted to compute the quantiles as well as the treatment effects (as the differences in potential outcomes at those quantiles). The robustness property of the efficient score functions permits estimating the nuisance parameters (outcome, treatment, and mediator models) by machine learning and cross-fitting avoids overfitting bias. We demonstrated that our quantile treatment effect estimators are root-n-consistent and asymptotically normal. Furthermore, we suggested a multiplier bootstrap and demonstrated its consistency for uniform statistical inference. We also investigated the finite sample performance of our estimators by means of a simulation study. Finally, we applied our method to data from the National Job Corp Study to evaluate the direct earnings effects of training across the earnings distribution, as well as the indirect effects operating via general health. We found positive and statistically significant direct effects across a large range of the earnings quantiles, while the indirect effects were generally close to zero and mostly statistically insignificant.


{\large \renewcommand{\theequation}{A-\arabic{equation}} %
\setcounter{equation}{0} \appendix
}
\appendix \numberwithin{equation}{section}

\clearpage

\section{Appendix}
\small
\subsection{Proof of Proposition 1}
\begin{proof} The proof relies on using Assumptions 1.1 through 1.4. Under these assumptions, it can be shown that: 
	\begin{eqnarray*}
		F_{Y\left(d,M\left(d^{\prime}\right)\right)}\left(a\right) & = & \int P\left(Y\left(d,M\left(d^{\prime}\right)\right)\leq a|X=x\right)f_{X}\left(x\right)dx\text{ (by iterated expectation) }\\
		& = & \int\int P\left(Y\left(d,m\right)\leq a|M\left(d^{\prime}\right)=m,X=x\right)dP\left(M\left(d^{\prime}\right)=m|X=x\right)f_{X}\left(x\right)dx\\
		&  & \text{ (by iterated expectation) }\\
		& = & \int\int P\left(Y\left(d,m\right)\leq a|D=d^{\prime},M\left(d^{\prime}\right)=m,X=x\right)\\
		&  & \times dP\left(M\left(d^{\prime}\right)=m|D=d^{\prime},X=x\right)f_{X}\left(x\right)dx\text{ (by Assumption 2) }\\
		& = & \int\int P\left(Y\left(d,m\right)\leq a|D=d^{\prime},M=m,X=x\right)dP\left(M=m|D=d^{\prime},X=x\right)f_{X}\left(x\right)dx\\
		&  & \text{ (by Assumption 1) }\\
		& = & \int\int P\left(Y\left(d,m\right)\leq a|D=d^{\prime},X=x\right)dP\left(M=m|D=d^{\prime},X=x\right)f_{X}\left(x\right)dx\\
		&  & \text{ (by Assumption 3) }\\
		& = & \int\int P\left(Y\left(d,m\right)\leq a|D=d,X=x\right)dP\left(M=m|D=d^{\prime},X=x\right)f_{X}\left(x\right)dx\\
		&  & \text{ (by Assumption 2) }\\
		& = & \int\int P\left(Y\left(d,m\right)\leq a|D=d,M=m,X=x\right)dP\left(M=m|D=d^{\prime},X=x\right)f_{X}\left(x\right)dx\\
		&  & \text{ (by Assumption 3) }\\
		& = & \int P\left(Y\leq a|D=d,M=m,X=x\right)dP\left(M=m|D=d^{\prime},X=x\right)f_{X}\left(x\right)dx\\
		&  & \text{ (by Assumption 1) }\\
		& = & \int\int F_{Y|D,M,X}\left(a|d,m,x\right)f_{M|D,X}\left(m|d^{\prime},x\right)f_{X}\left(x\right)dmdx.
	\end{eqnarray*}
\end{proof}

\subsection{Derivations of the EIF}
The derivation of the efficient influence function (EIF) of an estimand is based on calculating Gateaux derivatives for the estimand. 
Let $P$ denote the true data generating distribution and $\Psi\left(P\right)$ the estimand of interest, which is a statistical functional of $P$.
The Gateaux derivative of $\Psi\left(.\right)$ measures how the estimand
$\Psi\left(.\right)$ changes as $P$ shifts in the direction of another
distribution, say $\tilde{P}$. Let $P_{t}=t\tilde{P}+\left(1-t\right)P$,
where $t\in\left[0,1\right]$. Formally, the Gateaux derivative of 
estimand $\Psi\left(.\right)$ when changing $P$ in the direction of
$\tilde{P}$ is defined as 
\begin{equation}
	\lim_{t\downarrow0}\left(\frac{\Psi\left(P_{t}\right)-\Psi\left(P\right)}{t}\right)=\left.\frac{d}{dt}\Psi\left(P_{t}\right)\right|_{t=0},\label{Gateaux_derivative}
\end{equation}
if the limit on the right-hand side exists. It can be shown that under
certain regularity conditions, the EIF of $\Psi\left(P\right)$ under
the distribution $\tilde{P}$ is equal to Gateaux derivative (\ref{Gateaux_derivative})
\citep{HDDV_2022}. This fact provides a convenient way of deriving
the EIF. Following \citet{HDDV_2022}, we use the strategy of ``point
mass contamination'' to derive the EIF of $F_{Y\left(d,M\left(d^{\prime}\right)\right)}\left(a\right)$.
Specifically, we consider $\tilde{P}$ to be a point mass of a single
observation, say $\tilde{o}$, and then the EIF of $\Psi\left(P\right)$
evaluated at $\tilde{o}$ is equal to the Gateaux derivative (\ref{Gateaux_derivative}). This derivation strategy appears attractive when the treatment variable $D$ is discrete.\footnote{Notice that if $D$ is not discrete, this strategy can not be used,
	and the derivation needs to rely on using other methods instead, see
	\citet{FK_2019,Levy_2019,IN_2022}.} Let 
\[
1_{\tilde{o}}\left(o\right)=\begin{cases}
	1 & \text{if }o=\tilde{o}\\
	0 & \text{otherwise}
\end{cases}
\]
denote the Dirac delta function with respect to $\tilde{o}$. If the
density function for $P$ is $f_{O}\left(o\right)$, the density function
for $P_{t}$ is $f_{O}^{t}\left(o\right)=t1_{\tilde{o}}\left(o\right)+\left(1-t\right)f_{O}\left(o\right)$
and 
\[
\left.\frac{d}{dt}f_{O}^{t}\left(o\right)\right|_{t=0}=1_{\tilde{o}}\left(o\right)-f_{O}\left(o\right),
\]
and $f_{O}^{t}\left(o\right)=f_{O}\left(o\right)$ when $t=0$. Under Assumptions 1.1 to 1.4, it follows from
Proposition 1 that
\begin{eqnarray*}
	F_{Y\left(d,M\left(d^{\prime}\right)\right)}\left(a\right) & = & \int\int F_{Y|D,M,X}\left(a|d,m,x\right)f_{M|D,X}\left(m|d^{\prime},x\right)f_{X}\left(x\right)dmdx\\
	& = & \int\int\int1\left\{ y\leq a\right\} \frac{f_{Y,D,M,X}\left(y,d,m.x\right)}{f_{D,M,X}\left(d,m,x\right)}\frac{f_{M,D,X}\left(m,d^{\prime},x\right)}{f_{D,X}\left(d^{\prime},x\right)}f_{X}\left(x\right)dydmdx.
\end{eqnarray*}
Let 
\[
\Psi\left(P_{t}\right):=\int\int\int1\left\{ y\leq a\right\} \frac{f_{Y,D,M,X}^{t}\left(y,d,m.x\right)}{f_{D,M,X}^{t}\left(d,m,x\right)}\frac{f_{M,D,X}^{t}\left(m,d^{\prime},x\right)}{f_{D,X}^{t}\left(d^{\prime},x\right)}f_{X}^{t}\left(x\right)dydmdx.
\]
We would like to calculate the Gateau derivative: 
\[
\left.\frac{d}{dt}\Psi\left(P_{t}\right)\right|_{t=0}=\left.\frac{d}{dt}\left(\int\int\int1\left\{ y\leq a\right\} \frac{f_{Y,D,M,X}^{t}\left(y,d,m.x\right)}{f_{D,M,X}^{t}\left(d,m,x\right)}\frac{f_{M,D,X}^{t}\left(m,d^{\prime},x\right)}{f_{D,X}^{t}\left(d^{\prime},x\right)}f_{X}^{t}\left(x\right)dydmdx\right)\right|_{t=0}.
\]
It can be shown that 
\begin{eqnarray}
	\left.\frac{d}{dt}\Psi\left(P_{t}\right)\right|_{t=0} & = & \int\int\int1\left\{ y\leq a\right\} \left.\left(\frac{f_{Y,D,M,X}^{t}\left(y,d,m.x\right)}{f_{D,M,X}^{t}\left(d,m,x\right)}\frac{f_{M,D,X}^{t}\left(m,d^{\prime},x\right)}{f_{D,X}^{t}\left(d^{\prime},x\right)}\frac{d}{dt}f_{X}^{t}\left(x\right)\right)\right|_{t=0}dydmdx\notag\\\label{Gateaux}\\
	&  & +\int\int\int1\left\{ y\leq a\right\} \left.\left(f_{X}^{t}\left(x\right)\frac{d}{dt}\frac{f_{Y,D,M,X}^{t}\left(y,d,m.x\right)}{f_{D,M,X}^{t}\left(d,m,x\right)}\frac{f_{M,D,X}^{t}\left(m,d^{\prime},x\right)}{f_{D,X}^{t}\left(d^{\prime},x\right)}\right)\right|_{t=0}dydmdx.\notag\\\label{Gateaux1}
\end{eqnarray}
Considering the expression within the integral of (\ref{Gateaux}), 
\begin{eqnarray*}
	\left.\left(\frac{f_{Y,D,M,X}^{t}\left(y,d,m.x\right)}{f_{D,M,X}^{t}\left(d,m,x\right)}\frac{f_{M,D,X}^{t}\left(m,d^{\prime},x\right)}{f_{D,X}^{t}\left(d^{\prime},x\right)}\frac{d}{dt}f_{X}^{t}\left(x\right)\right)\right|_{t=0} & = & \frac{f_{Y,D,M,X}^{t}\left(y,d,m.x\right)}{f_{D,M,X}^{t}\left(d,m,x\right)}\frac{f_{M,D,X}^{t}\left(m,d^{\prime},x\right)}{f_{D,X}^{t}\left(d^{\prime},x\right)}\times\left[1_{\tilde{x}}\left(x\right)-f_{X}\left(x\right)\right]\\
	& = & f_{Y|D,M,X}\left(y|d,m.x\right)f_{M|D,X}\left(m|d^{\prime},x\right)1_{\tilde{x}}\left(x\right)\\
	&  & -f_{Y|D,M,X}\left(y|d,m.x\right)f_{M|D,X}\left(m|d^{\prime},x\right)f_{X}\left(x\right).
\end{eqnarray*}
Therefore, (\ref{Gateaux}) can be further expressed as 
\begin{equation}
	E\left[F_{Y|D,M,X}\left(a|d,M,\tilde{x}\right)|d^{\prime},\tilde{x}\right]-E\left[g_{d,d^{\prime},a}\left(X\right)\right].\label{EIF}
\end{equation}
Considering the xpression within the integral of (\ref{Gateaux1}), 
\begin{eqnarray}
	\left.\left(f_{X}^{t}\left(x\right)\frac{d}{dt}\frac{f_{Y,D,M,X}^{t}\left(y,d,m.x\right)}{f_{D,M,X}^{t}\left(d,m,x\right)}\frac{f_{M,D,X}^{t}\left(m,d^{\prime},x\right)}{f_{D,X}^{t}\left(d^{\prime},x\right)}\right)\right|_{t=0} & = & \left.\left(f_{X}^{t}\left(x\right)\frac{f_{Y,D,M,X}^{t}\left(y,d,m.x\right)}{f_{D,M,X}^{t}\left(d,m,x\right)}\frac{d}{dt}\frac{f_{M,D,X}^{t}\left(m,d^{\prime},x\right)}{f_{D,X}^{t}\left(d^{\prime},x\right)}\right)\right|_{t=0}\notag\\\label{Gateaux2}\\
	&  & +\left.\left(f_{X}^{t}\left(x\right)\frac{f_{M,D,X}^{t}\left(m,d^{\prime},x\right)}{f_{D,X}^{t}\left(d^{\prime},x\right)}\frac{d}{dt}\frac{f_{Y,D,M,X}^{t}\left(y,d,m.x\right)}{f_{D,M,X}^{t}\left(d,m,x\right)}\right)\right|_{t=0}.\notag\\\label{Gateaux3}
\end{eqnarray}
Considering (\ref{Gateaux2}), 
\begin{eqnarray*}
	\left.\frac{d}{dt}\frac{f_{M,D,X}^{t}\left(m,d^{\prime},x\right)}{f_{D,X}^{t}\left(d^{\prime},x\right)}\right|_{t=0} & = & \left.\left(\frac{1}{f_{D,X}^{t}\left(d^{\prime},x\right)}\frac{d}{dt}f_{M,D,X}^{t}\left(m,d^{\prime},x\right)\right)\right|_{t=0}\\
	&  & -\left.\left(\frac{f_{M,D,X}^{t}\left(m,d^{\prime},x\right)}{\left(f_{D,X}^{t}\left(d^{\prime},x\right)\right)^{2}}\frac{d}{dt}f_{D,X}^{t}\left(d^{\prime},x\right)\right)\right|_{t=0}\\
	& = & \frac{1\left\{ D=d^{\prime}\right\} }{f_{D,X}\left(d^{\prime},x\right)}\left[1_{\left(\tilde{m},\tilde{x}\right)}\left(m,x\right)-f_{M|D,X}\left(m|d^{\prime},x\right)1_{\tilde{x}}\left(x\right)\right].
\end{eqnarray*}
Using some algebra, the part of (\ref{Gateaux1}) appearing in (\ref{Gateaux2})
can be expressed as 
\begin{equation}
	\frac{1\left\{ D=d^{\prime}\right\} }{f_{D|X}\left(d^{\prime}|\tilde{x}\right)}\left(F_{Y|D,M,X}\left(a|d,\tilde{m},\tilde{x}\right)-E\left[F_{Y|D,M,X}\left(a|d,M,\tilde{x}\right)|d^{\prime},\tilde{x}\right]\right).\label{EIF1}
\end{equation}
Concerning (\ref{Gateaux3}), 
\begin{eqnarray*}
	\left.\frac{d}{dt}\frac{f_{Y,D,M,X}^{t}\left(y,d,m.x\right)}{f_{D,M,X}^{t}\left(d,m,x\right)}\right|_{t=0} & = & \left.\left(\frac{1}{f_{D,M,X}^{t}\left(d,m,x\right)}\frac{d}{dt}f_{Y,D,M,X}^{t}\left(y,d,m.x\right)\right)\right|_{t=0}\\
	&  & -\left.\left(\frac{f_{Y,D,M,X}^{t}\left(y,d,m,x\right)}{\left(f_{D,M,X}^{t}\left(d,m,x\right)\right)^{2}}\frac{d}{dt}f_{D,M,X}^{t}\left(d,m,x\right)\right)\right|_{t=0}\\
	& = & \frac{1}{f_{D,M,X}\left(d,m,x\right)}\left[1_{\left(\tilde{y},\tilde{m},\tilde{x}\right)}\left(y,m,x\right)1\left\{ D=d\right\} -f_{Y,D,M,X}\left(y,d,m,x\right)\right]\\
	&  & -\frac{f_{Y|D,M,X}\left(y|d,m,x\right)}{f_{D,M,X}\left(d,m,x\right)}\left[1_{\left(\tilde{m},\tilde{x}\right)}\left(m,x\right)1\left\{ D=d\right\} -f_{D,M,X}\left(d,m,x\right)\right]\\
	& = & \frac{1\left\{ D=d\right\} }{f_{D,M,X}\left(d,m,x\right)}\left[1_{\left(\tilde{y},\tilde{m},\tilde{x}\right)}\left(y,m,x\right)-f_{Y|D,M,X}\left(y|d,m,x\right)1_{\tilde{m},\tilde{x}}\left(m,x\right)\right].
\end{eqnarray*}
Furthermore, the part of (\ref{Gateaux1}) appearing in (\ref{Gateaux3})
can be expressed as 
\begin{equation}
	\frac{1\left\{ D=d\right\} }{f_{D|X}\left(d|\tilde{x}\right)}\frac{f_{M|D,X}\left(\tilde{m}|d^{\prime},\tilde{x}\right)}{f_{M|D,X}\left(\tilde{m}|d,\tilde{x}\right)}\left(1\left\{ \tilde{y}\leq a\right\} -F_{Y|D,M,X}\left(a|d,\tilde{m},\tilde{x}\right)\right).\label{EIF2}
\end{equation}
Combing (\ref{EIF}), (\ref{EIF1}) and (\ref{EIF2}), we obtain
\begin{eqnarray*}
	\left.\frac{d}{dt}\Psi\left(P_{t}\right)\right|_{t=0} & = & E\left[F_{Y|D,M,X}\left(a|d,M,\tilde{x}\right)|d^{\prime},\tilde{x}\right]-E\left[g_{d,d^{\prime},a}\left(X\right)\right]\\
	&  & +\frac{1\left\{ D=d^{\prime}\right\} }{f_{D|X}\left(d^{\prime}|\tilde{x}\right)}\left(F_{Y|D,M,X}\left(a|d,\tilde{m},\tilde{x}\right)-E\left[F_{Y|D,M,X}\left(a|d,M,\tilde{x}\right)|d^{\prime},\tilde{x}\right]\right)\\
	&  & +\frac{1\left\{ D=d\right\} }{f_{D|X}\left(d|\tilde{x}\right)}\frac{f_{M|D,X}\left(\tilde{m}|d^{\prime},\tilde{x}\right)}{f_{M|D,X}\left(\tilde{m}|d,\tilde{x}\right)}\left(1\left\{ \tilde{y}\leq a\right\} -F_{Y|D,M,X}\left(a|d,\tilde{m},\tilde{x}\right)\right).
\end{eqnarray*}
If we replace the notation $\left(\tilde{y},\tilde{m},\tilde{x}\right)$
with $\left(Y,M,X\right)$ and notice that $g_{d,d^{\prime},a}\left(X\right):=E\left[F_{Y|D,M,X}\left(a|d,M,X\right)|d^{\prime},X\right]$
and $E\left[g_{d,d^{\prime},a}\left(X\right)\right]=F_{Y\left(d,M\left(d^{\prime}\right)\right)}\left(a\right)$,
then $E\left[\left.\frac{d}{dt}\Psi\left(P_{t}\right)\right|_{t=0}\right]=0$
implies that 
\[
F_{Y\left(d,M\left(d^{\prime}\right)\right)}\left(a\right)=E\left[\psi^{\prime}\left(W_{a},v_{a}^{\prime}\right)\right],
\]
where $\psi^{\prime}\left(W_{a},v_{a}^{\prime}\right)$ is defined
in equation (\ref{psi_d_dprime_a1}). We can apply the Bayes rule to rewrite
the term (\ref{EIF2}) as 
\[
\frac{1\left\{ D=d\right\} }{f_{D|X}\left(d^{\prime}|\tilde{x}\right)}\frac{f_{D|M,X}\left(d^{\prime}|\tilde{m},\tilde{x}\right)}{f_{D|M,X}\left(d|\tilde{m},\tilde{x}\right)}\left(1\left\{ \tilde{y}\leq a\right\} -F_{Y|D,M,X}\left(a|d,\tilde{m},\tilde{x}\right)\right),
\]
and $E\left[\left.\frac{d}{dt}\Psi\left(P_{t}\right)\right|_{t=0}\right]=0$
now implies that 
\[
F_{Y\left(d,M\left(d^{\prime}\right)\right)}\left(a\right)=E\left[\psi\left(W_{a},v_{a}\right)\right],
\]
where $\psi\left(W_{a},v_{a}\right)$ is defined in equation (\ref{psi_d_dprime_a}).

\subsection{Proofs of Theorems in Section 3}
\begin{proof}[Proof of Theorem 1]
The proof relies on Theorem A.1 in Appendix A.4, which states that K-fold cross-fitting is uniformly valid for estimating a parameter of interest under certain regularity conditions. We show that the conditions in Assumption 2 are sufficient for the proposed K-fold cross-fitting estimator to satisfy Assumptions A.1.1 to A.1.8, which are required for establishing Theorem A.1. Notice that the condition $\left\Vert \hat{v}_{a}-v_{a}^{0}\right\Vert _{P,2} \leq\delta_{n}n^{-\frac{1}{4}}$ in Assumption 2.4 already satisfies Assumption A.1.8. For this reason, we will only verify Assumption A.1.1 to A.1.7 in the subsequent discussion. We first derive several preliminary results which are useful for the proof to follow.

Let $\mathcal{G}_{an}$ be the set of \[v:=\left(g_{1d}\left(X\right),g_{2d}\left(M,X\right),g_{3a}\left(D,M,X\right),g_{4ad}\left(D,X\right)\right).\]
 $g_{1d}\left(X\right),g_{2d}\left(M,X\right),g_{3a}\left(D,M,X\right)$
and $g_{4ad}\left(D,X\right)$ are $P$-integrable functions such
that for $d\in\left\{ 0,1\right\}$, 
\begin{align}
	P\left(\varepsilon_{1}<g_{1d}\left(X\right)<1-\varepsilon_{1}\right) & =1,\label{bound_1}\\
	P\left(\varepsilon_{2}<g_{2d}\left(M,X\right)<1-\varepsilon_{2}\right) & =1,\label{bound_2}
\end{align}
where $\varepsilon_{1},\varepsilon_{2}\in\left(0,1/2\right)$, and
with probability $P$ at least $1-\Delta_{n}$, for $d\in\left\{ 0,1\right\} $,
all $a\in\mathcal{A}$ and $q\geq4$, 
\begin{align*}
	\left\Vert v-v_{a}^{0}\right\Vert _{P,q} & \leq C,\\
	\left\Vert v-v_{a}^{0}\right\Vert _{P,2} & \leq\delta_{n}n^{-\frac{1}{4}},\\
	\left\Vert g_{1d}\left(X\right)-0.5\right\Vert _{P,\infty} & \leq0.5-\epsilon,\\
	\left\Vert g_{2d}\left(M,X\right)-0.5\right\Vert _{P,\infty} & \leq0.5-\epsilon,\\
	\left\Vert g_{1d}\left(X\right)-g_{1d}^{0}\left(X\right)\right\Vert _{P,2}\left\Vert g_{2d}\left(M,X\right)-g_{2d}^{0}\left(M,X\right)\right\Vert _{P,2} & \leq\delta_{n}n^{-\frac{1}{2}},\\
	\left\Vert g_{1d}\left(X\right)-g_{1d}^{0}\left(X\right)\right\Vert _{P,2}\left\Vert g_{3a}\left(D,M,X\right)-g_{3a}^{0}\left(D,M,X\right)\right\Vert _{P,2} & \leq\delta_{n}n^{-\frac{1}{2}},\\
	\left\Vert g_{1d}\left(X\right)-g_{1d}^{0}\left(X\right)\right\Vert _{P,2}\left\Vert g_{4ad}\left(D,X\right)-g_{4ad}^{0}\left(D,X\right)\right\Vert _{P,2} & \leq\delta_{n}n^{-\frac{1}{2}},\\
	\left\Vert g_{2d}\left(M,X\right)-g_{2d}^{0}\left(M,X\right)\right\Vert _{P,2}\left\Vert g_{3a}\left(D,M,X\right)-g_{3a}^{0}\left(D,M,X\right)\right\Vert _{P,2} & \leq\delta_{n}n^{-\frac{1}{2}}.
\end{align*}
Notice that $\hat{v}_{a}\in\mathcal{G}_{an}$ by Assumption 2.4. Proving the result for all functions in
$\mathcal{G}_{an}$ implies that it also holds for $\hat{v}_{a}$. By
Assumption 2.2, it can be shown that for $\left(d,d^{\prime}\right)\in\left\{ 0,1\right\} ^{2}$,
the event 
\[
0<\frac{\varepsilon_{2}}{1-\varepsilon_{2}}<\frac{g_{2d}^{0}\left(M,X\right)}{g_{2d^{\prime}}^{0}\left(M,X\right)}<\frac{1-\varepsilon_{2}}{\varepsilon_{2}}<\infty
\]
holds with probability one. Furthermore,
\begin{small}
	\begin{align*}
		\left\Vert g_{3a}\left(D,M,X\right)-g_{3a}^{0}\left(D,M,X\right)\right\Vert _{P,q} & =\left(E\left[E\left[\left|g_{3a}\left(D,M,X\right)-g_{3a}^{0}\left(D,M,X\right)\right|^{q}|M,X\right]\right]\right)^{\frac{1}{q}}\\
		& =\left(E\left[\sum_{d\in\left\{ 0,1\right\} }\left|g_{3a}\left(d,M,X\right)-g_{3a}^{0}\left(d,M,X\right)\right|^{q}g_{2d}^{0}\left(M,X\right)\right]\right)^{\frac{1}{q}}\\
		& \geq\varepsilon_{2}^{\frac{1}{q}}\left(E\left[\sum_{d\in\left\{ 0,1\right\} }\left|g_{3a}\left(d,M,X\right)-g_{3a}^{0}\left(d,M,X\right)\right|^{q}\right]\right)^{\frac{1}{q}}\\
		& \geq\varepsilon_{2}^{\frac{1}{q}}\left(\max_{d\in\left\{ 0,1\right\} }E\left[\left|g_{3a}\left(d,M,X\right)-g_{3a}^{0}\left(d,M,X\right)\right|^{q}\right]\right)^{\frac{1}{q}}.
	\end{align*}
\end{small}
which implies that for $d\in\left\{ 0,1\right\} $ and all $a\in\mathcal{A}$,
with $P$ probability at least $1-\Delta_{n}$, 
\[
\left\Vert g_{3a}\left(d,M,X\right)-g_{3a}^{0}\left(d,M,X\right)\right\Vert _{P,q}\leq C\varepsilon_{2}^{-\frac{1}{q}},
\]
by the assumption $\left\Vert v-v_{a}^{0}\right\Vert _{P,q}\leq C$
and $\varepsilon_{2}>0$. A similar argument can be used to show that
for $d\in\left\{ 0,1\right\} $ and all $a\in\mathcal{A}$, with 
probability $P$ at least $1-\Delta_{n}$, 
\begin{align}
	\left\Vert g_{3a}\left(d,M,X\right)-g_{3a}^{0}\left(d,M,X\right)\right\Vert _{P,2} & \leq\delta_{n}n^{-\frac{1}{4}}\varepsilon_{2}^{-\frac{1}{q}}\lesssim\delta_{n}n^{-\frac{1}{4}},\label{bound_3}\\
	\left\Vert g_{4ad}\left(d^{\prime},X\right)-g_{4d}^{0}\left(d^{\prime},X\right)\right\Vert _{P,2} & \leq\delta_{n}n^{-\frac{1}{4}}\varepsilon_{1}^{-\frac{1}{q}}\lesssim\delta_{n}n^{-\frac{1}{4}},\label{bound_4}
\end{align}
by assumption $\left\Vert v-v_{a}^{0}\right\Vert _{P,2}\leq\delta_{n}n^{-\frac{1}{4}}$
and $\varepsilon_{2}>0$. Similarly, 
\begin{align}
	\left\Vert g_{1d}\left(X\right)-g_{1d}^{0}\left(X\right)\right\Vert _{P,2}\left\Vert g_{3a}\left(d,M,X\right)-g_{3a}^{0}\left(d,M,X\right)\right\Vert _{P,2} & \leq\varepsilon_{2}^{-\frac{1}{q}}\delta_{n}n^{-\frac{1}{2}}\lesssim\delta_{n}n^{-\frac{1}{2}},\label{bound_5}\\
	\left\Vert g_{1d}\left(X\right)-g_{1d}^{0}\left(X\right)\right\Vert _{P,2}\left\Vert g_{4ad}\left(d^{\prime},X\right)-g_{4ad}^{0}\left(d^{\prime},X\right)\right\Vert _{P,2} & \leq\varepsilon_{1}^{-\frac{1}{q}}\delta_{n}n^{-\frac{1}{2}}\lesssim\delta_{n}n^{-\frac{1}{2}},\label{bound_6}\\
	\left\Vert g_{2d}\left(M,X\right)-g_{2d}^{0}\left(M,X\right)\right\Vert _{P,2}\left\Vert g_{3a}\left(d,M,X\right)-g_{3a}^{0}\left(d,M,X\right)\right\Vert _{P,2} & \leq\varepsilon_{2}^{-\frac{1}{q}}\delta_{n}n^{-\frac{1}{2}}\lesssim\delta_{n}n^{-\frac{1}{2}}.\label{bound_7}
\end{align}

\subsubsection*{Verifying Assumption A.1.1: $F_{Y(d,M(d^{\prime}))}^{0}(a)=\theta_{d,d^{\prime},a}^{0}$}
\subsubsection*{The case when $d\protect\neq d^{\prime}$}
Recall that 
\begin{align}
	\psi_{d,d^{\prime},a}\left(W_{a};v_{a}^{0}\right) & =\frac{1\left\{ D=d\right\} \left[1-g_{2d}^{0}\left(M,X\right)\right]}{\left[1-g_{1d}^{0}\left(X\right)\right]g_{2d}^{0}\left(M,X\right)}\times\left[Y_{a} -g_{3a}^{0}\left(d,M,X\right)\right]\label{efficient_score1}\\
	& +\frac{1\left\{ D=d^{\prime}\right\} }{1-g_{1d}^{0}\left(X\right)}\left[g_{3a}^{0}\left(d,M,X\right)-g_{4ad}^{0}\left(d^{\prime},X\right)\right]+g_{4ad}^{0}\left(d^{\prime},X\right).\label{efficient_score2}
\end{align}
Conditional on $M$ and $X$, the expectation of the first term after the equals sign in expression (\ref{efficient_score1})
is 
\begin{align*}
	E\left[\frac{1\left\{ D=d\right\} \left[1-g_{2d}^{0}\left(M,X\right)\right]}{\left[1-g_{1d}^{0}\left(X\right)\right]g_{2d}^{0}\left(M,X\right)}Y_{a} |M,X\right] & =\frac{1-g_{2d}^{0}\left(M,X\right)}{\left[1-g_{1d}^{0}\left(X\right)\right]g_{2d}^{0}\left(M,X\right)}\times E\left[1\left\{ D=d\right\} Y_{a} |M,X\right]\\
	& =\frac{1-g_{2d}^{0}\left(M,X\right)}{\left[1-g_{1d}^{0}\left(X\right)\right]g_{2d}^{0}\left(M,X\right)}\times g_{2d}\left(M,X\right)g_{3a}\left(d,M,X\right)\\
	& =\frac{1-g_{2d}^{0}\left(M,X\right)}{1-g_{1d}^{0}\left(X\right)}g_{3a}\left(d,M,X\right),
\end{align*}
and the expectation of the last term right in expression  (\ref{efficient_score1}) is
\begin{align*}
	E\left[\frac{1\left\{ D=d\right\} \left[1-g_{2d}^{0}\left(M,X\right)\right]}{\left[1-g_{1d}^{0}\left(X\right)\right]g_{2d}^{0}\left(M,X\right)}g_{3a}^{0}\left(d,M,X\right)|M,X\right] & =\frac{1-g_{2d}^{0}\left(M,X\right)}{\left[1-g_{1d}^{0}\left(X\right)\right]g_{2d}^{0}\left(M,X\right)}\times g_{2d}\left(M,X\right)g_{3a}\left(d,M,X\right)\\
	& =\frac{1-g_{2d}^{0}\left(M,X\right)}{1-g_{1d}^{0}\left(X\right)}g_{3a}\left(d,M,X\right).
\end{align*}
Therefore, conditional on $M$ and $X$, the term in expression (\ref{efficient_score1})
is zero and its unconditional expectation is also zero. Concerning the terms
in expression (\ref{efficient_score2}), notice that the first term 
\begin{align*}
	E\left[\frac{1\left\{ D=d^{\prime}\right\} }{1-g_{1d}^{0}\left(X\right)}\left[g_{3a}^{0}\left(d,M,X\right)-g_{4ad}^{0}\left(d^{\prime},X\right)\right]|X\right] & =\frac{1}{1-g_{1d}^{0}\left(X\right)}E\left[g_{3a}^{0}\left(d,M,X\right)1\left\{ D=d^{\prime}\right\} |X\right]\\
	& -g_{4ad}^{0}\left(d^{\prime},X\right)\\
	& =0
\end{align*}
Therefore, the expectation of the first term in expression (\ref{efficient_score2})
is zero. The expectation of the second term of in expression (\ref{efficient_score2})
is $E\left[g_{4ad}^{0}\left(d^{\prime},X\right)\right]=F_{Y(d,M(d^{\prime}))}^{0}(a)$ by Proposition 1. Therefore, it follows that $F_{Y(d,M(d^{\prime}))}^{0}(a)=\theta_{d,d^{\prime},a}^{0}$.

\subsubsection*{The case when $d=d$}
Now, we have
\begin{align}
	\psi_{d,d,a}\left(W_{a};v_{a}^{0}\right) & =\frac{1\left\{ D=d\right\} }{g_{1d}^{0}\left(X\right)}\times\left[Y_{a} -g_{4ad}^{0}\left(d,X\right)\right]+g_{4ad}^{0}\left(d,X\right),\label{efficient_score3}
\end{align}
where $g_{4ad}^{0}\left(d,X\right)=F_{Y|D,X}^{0}\left(a|d,X\right)$. Concerning
the first term right of the equals sign in equation (\ref{efficient_score3}),
\begin{align*}
	E\left[\frac{1\left\{ D=d\right\} }{g_{1d}^{0}\left(X\right)}\times\left[Y_{a} -g_{4ad}^{0}\left(d,X\right)\right]|X\right] & =\frac{E\left[1\left\{ D=d\right\} Y_{a} |X\right]}{g_{1d}^{0}\left(X\right)}-g_{4ad}^{0}\left(d,X\right)\\
	& =0.
\end{align*}
Concerning the last term in equation (\ref{efficient_score3}), notice that \[F_{Y|D,X}^{0}\left(a|d,X\right)=F_{Y\left(d,M\left(d\right)\right)|D,X}^{0}\left(a|d,X\right)=F_{Y\left(d,M\left(d\right)\right)|X}^{0}\left(a|X\right)\]
by Assumptions 1.1 and 1.2, and therefore,
\begin{equation*}
	E\left[g_{4ad}^{0}\left(d,X\right)\right]  =E\left[F_{Y\left(d,M\left(d\right)\right)|X}^{0}\left(a|X\right)\right]=F_{Y(d,M(d))}^{0}(a).
\end{equation*}
Combining the previous results, it follows that $F_{Y(d,M(d^{\prime}))}^{0}(a)=\theta_{d,d^\prime,a}^{0}$
holds for all $a\in\mathcal{A}$ and $\left(d,d^{\prime}\right)\in\left\{ 0,1\right\} ^{2}$.


\subsubsection*{Verifying Assumption A.1.2}
If we treat $v$ as deterministic, the second order Gateau derivative of the map $v\longmapsto E\left[\psi_{d,d^{\prime},a}\left(W_{a};v\right)\right]$ exists and is continuous on $v\in\mathcal{G}_{an}$, and this property holds for each $\left(d,d^{\prime}\right)\in\left\{ 0,1\right\} ^{2}$ and $a\in\mathcal{A}$. Therefore $v\longmapsto E\left[\boldsymbol{\psi}_{a}\left(W_{a};v\right)\right]$ is twice continuously Gateau-differetiable for $a\in\mathcal{A}$.


\subsubsection*{Verifying Assumption A.1.3 (Neyman near orthogonality)}
\subsubsection*{The case when $d\protect\neq d^{\prime}$}
Recall that 
\begin{align*}
	\psi_{d,d^{\prime},a}\left(W_{a},v_{a}^{0}\right) & =\frac{1\left\{ D=d\right\} \left[1-g_{2d}^{0}\left(M,X\right)\right]}{\left[1-g_{1d}^{0}\left(X\right)\right]g_{2d}^{0}\left(M,X\right)}\left[1\left\{ Y\leq a\right\} -g_{3a}^{0}\left(d,M,X\right)\right]\\
	& +\frac{1\left\{ D=d^{\prime}\right\} }{1-g_{1d}^{0}\left(X\right)}\left[g_{3a}^{0}\left(d,M,X\right)-g_{4ad}^{0}\left(d^{\prime},X\right)\right]+g_{4ad}^{0}\left(d^{\prime},X\right).
\end{align*}
Let 
\[
\mu_{d,d^{\prime},a}\left(\mathbf{t}\right)=\frac{1\left\{ D=d\right\} \left(1-t_{2}\right)}{\left(1-t_{1}\right)t_{2}}\left(1\left\{ Y\leq a\right\} -t_{3}\right)+\frac{1\left\{ D=d^{\prime}\right\} }{1-t_{1}}\left(t_{3}-t_{4}\right)+t_{4},
\]
where $\mathbf{t}=\left(t_{1,},\ldots,t_{4}\right)$. If we set $\mathbf{t}=\left(g_{1d}\left(X\right),g_{2d}\left(M,X\right),g_{3a}\left(d,M,X\right),g_{4ad}\left(d^{\prime},X\right)\right)$,
then $E\left[\mu_{d,d^{\prime},a}\left(\mathbf{t}\right)\right]=E\left[\psi_{d,d^{\prime},a}\left(W_{a},v\right)\right]$.
If we set 
\[
\mathbf{t}=\mathbf{t}^{0}=\left(g_{1d}^{0}\left(X\right),g_{2d}^{0}\left(M,X\right),g_{3a}^{0}\left(d,M,X\right),g_{4ad}^{0}\left(d^{\prime},X\right)\right),
\]
then $E\left[\mu_{d,d^{\prime},a}\left(\mathbf{t}_{0}\right)\right]=E\left[\psi_{d,d^{\prime},a}\left(W_{a},v_{a}^{0}\right)\right]$.
Furthermore, $\left\Vert \partial_{v}E\left[\psi_{d,d^{\prime},a}\left(W_{a},v_{a}^{0}\right)\left[v-v_{a}^{0}\right]\right]\right\Vert =\left\Vert E\left[\partial_{\mathbf{t}}\mu_{d,d^{\prime},a}\left(\mathbf{t}_{0}\right)\left[\mathbf{t}-\mathbf{t}_{0}\right]\right]\right\Vert $.
The partial derivatives with respect to $\left(t_{1},\ldots,t_{4}\right)$
are given by 
\begin{align}
	\partial_{t_{1}}\mu_{d,d^{\prime},a}\left(\mathbf{t}\right) & =\frac{1\left\{ D=d\right\} \left(1-t_{2}\right)}{\left(1-t_{1}\right)^{2}t_{2}}\left(1\left\{ Y\leq a\right\} -t_{3}\right)+\frac{1\left\{ D=d^{\prime}\right\} }{(1-t_{1})^{2}}\left(t_{3}-t_{4}\right),\label{partial_mu_t1}\\
	\partial_{t_{2}}\mu_{d,d^{\prime},a}\left(\mathbf{t}\right) & =-\frac{1\left\{ D=d\right\} }{\left(1-t_{1}\right)t_{2}^{2}}\left(1\left\{ Y\leq a\right\} -t_{3}\right),\label{partial_mu_t2}\\
	\partial_{t_{3}}\mu_{d,d^{\prime},a}\left(\mathbf{t}\right) & =\frac{1\left\{ D=d^{\prime}\right\} }{1-t_{1}}-\frac{1\left\{ D=d\right\} \left(1-t_{2}\right)}{\left(1-t_{1}\right)t_{2}},\label{partial_mu_t3}\\
	\partial_{t_{4}}\mu_{d,d^{\prime},a}\left(\mathbf{t}\right) & =1-\frac{1\left\{ D=d^{\prime}\right\} }{1-t_{1}}.\label{partial_mu_t4}
\end{align}
Replacing $\mathbf{t}$ with $\mathbf{t}^{0}$ and taking expectations in equation (\ref{partial_mu_t1}), we have 
\begin{align}
	E\left[\partial_{t_{1}}\mu_{d,d^{\prime},a}\left(\mathbf{t}_{0}\right)\left[g_{1d}\left(X\right)-g_{1d}^{0}\left(X\right)\right]\right] & =E\left[\frac{1-g_{2d}^{0}\left(M,X\right)}{\left(1-g_{1d}^{0}\left(X\right)\right)^{2}g_{2d}^{0}\left(M,X\right)}\times\left[E\left[1\left\{ Y\leq a\right\} 1\left\{ D=d\right\} |M,X\right]\right.\right.\label{partial_t1}\\
	& \left.\left.-g_{2d}^{0}\left(M,X\right)g_{3a}^{0}\left(d,M,X\right)\right]\left[g_{1d}\left(X\right)-g_{1d}^{0}\left(X\right)\right]\right]\nonumber \\
	& +E\left[\frac{1}{\left(1-g_{1d}^{0}\left(X\right)\right)^{2}}\left[E\left[1\left\{ D=d^{\prime}\right\} g_{3a}^{0}\left(d,M,X\right)|X\right]\right.\right.\label{partial_t11}\\
	& \left.-\left(1-g_{1d}^{0}\left(X\right)\right)g_{4ad}^{0}\left(d^{\prime},X\right)\right]\left[g_{1d}\left(X\right)-g_{1d}^{0}\left(X\right)\right]\nonumber \\
	& =0,\nonumber 
\end{align}
because in expressions (\ref{partial_t1}) and (\ref{partial_t11}), we have that 
\begin{align}
	E\left[1\left\{ D=d\right\} 1\left\{ Y\leq a\right\} |M,X\right] & =P_{Y,D|M,X}^{0}\left(Y\leq a,D=d|M,X\right)\nonumber \\
	& =g_{2d}^{0}\left(M,X\right)g_{3a}^{0}\left(d,M,X\right),\label{partial_t13}\\
	E1\left[1\left\{ D=d^{\prime}\right\} g_{3a}^{0}\left(d,M,X\right)|X\right] & =\int g_{3a}^{0}\left(d,m,X\right)P_{D,M|X}^{0}\left(D=d^{\prime},M=m|X\right)dm\nonumber \\
	& =\left(1-g_{1d}^{0}\left(X\right)\right)g_{4ad}^{0}\left(d^{\prime},X\right).\nonumber 
\end{align}
When taking expectations in equation (\ref{partial_mu_t2}), we have 
\begin{align}
	E\left[\partial_{t_{2}}\mu_{d,d^{\prime},a}\left(\mathbf{t}_{0}\right)\left[g_{2d}\left(M,X\right)-g_{2d}^{0}\left(M,X\right)\right]\right] & =E\left[-\frac{1}{\left(1-g_{1d}^{0}\left(X\right)\right)\left(g_{2d}^{0}\left(M,X\right)\right)^{2}}\right.\\
	& \times\left[E\left[1\left\{ D=d^{\prime}\right\} 1\left\{ Y\leq a\right\} |M,X\right]-g_{2d}^{0}\left(M,X\right)g_{3a}^{0}\left(d,M,X\right)\right]\nonumber \\
	& \left.\times\left[g_{2d}\left(M,X\right)-g_{2d}^{0}\left(M,X\right)\right]\right]\nonumber \\
	& =0,\nonumber 
\end{align}
by making use of expression (\ref{partial_t13}). When taking expectations in equation (\ref{partial_mu_t3}), we have
\begin{align*}
	E\left[\partial_{t_{3}}\mu_{d,d^{\prime},a}\left(\mathbf{t}_{0}\right)\left[g_{3a}\left(d,M,X\right)-g_{3a}^{0}\left(d,M,X\right)\right]\right] & =E\left[E\left[\frac{1\left\{ D=d^{\prime}\right\} }{1-g_{1d}^{0}\left(X\right)}-\frac{1\left\{ D=d\right\} \left(1-g_{2d}^{0}\left(M,X\right)\right)}{\left(1-g_{1d}^{0}\left(X\right)\right)g_{2d}^{0}\left(M,X\right)}|M,X\right]\right.\\
	& \left.\times\left[g_{3a}\left(d,M,X\right)-g_{3a}^{0}\left(d,M,X\right)\right]\right]\\
	& =0,
\end{align*}
because 
\begin{align*}
	E\left[\frac{1\left\{ D=d^{\prime}\right\} }{1-g_{1d}^{0}\left(X\right)}-\frac{1\left\{ D=d\right\} \left(1-g_{2d}^{0}\left(M,X\right)\right)}{\left(1-g_{1d}^{0}\left(X\right)\right)g_{2d}^{0}\left(M,X\right)}|M,X\right] & =0.
\end{align*}
When taking expectations in equation (\ref{partial_mu_t4}), we have 
\begin{align*}
	E\left[\partial_{t_{4}}\mu_{d,d^{\prime},a}\left(\mathbf{t}_{0}\right)\left[g_{4ad}\left(d^{\prime},X\right)-g_{4ad}^{0}\left(d^{\prime},X\right)\right]\right] & =E\left[E\left[\left(1-\frac{1\left\{ D=d^{\prime}\right\} }{1-g_{1d}^{0}\left(X\right)}\right)|X\right]\left[g_{4ad}\left(d^{\prime},X\right)-g_{4ad}^{0}\left(d^{\prime},X\right)\right]\right]\\
	& =0,
\end{align*}
since 
\[
E\left[\left(1-\frac{1\left\{ D=d^{\prime}\right\} }{1-g_{1d}^{0}\left(X\right)}\right)|X\right]=1-\frac{1-g_{1d}^{0}\left(X\right)}{1-g_{1d}^{0}\left(X\right)}=0.
\]

\subsubsection*{The case when $d=d^{\prime}$}

Recall that 
\begin{align*}
	\psi_{d,d,a}\left(W_{a},v_{a}^{0}\right) & =\frac{1\left\{ D=d\right\} }{g_{1d}^{0}\left(X\right)}\left[1\left\{ Y\leq a\right\} -g_{4ad}^{0}\left(d,X\right)\right]+g_{4ad}^{0}\left(d,X\right).
\end{align*}
Now $\mu_{d,d^{\prime},a}\left(\mathbf{t}\right)$ is given by 
\[
\mu_{d,d^{\prime},a}\left(\mathbf{t}\right)=\frac{1\left\{ D=d\right\} }{t_{1}}\left(1\left\{ Y\leq a\right\} -t_{4}\right)+t_{4}.
\]
The partial derivatives with respect to $\left(t_{1},t_{4}\right)$
are given by 
\begin{align}
	\partial_{t_{1}}\mu_{d,d^{\prime},a}\left(\mathbf{t}\right) & =-\frac{1\left\{ D=d\right\} }{t_{1}^{2}}\left(1\left\{ Y\leq a\right\} -t_{4}\right)\label{partial_mu1_t1},\\
	\partial_{t_{4}}\mu_{d,d^{\prime},a}\left(\mathbf{t}\right) & =1-\frac{1\left\{ D=d\right\} }{t_{1}}.\label{partial_mu1_t4}
\end{align}
Replacing $\mathbf{t}$ with $\mathbf{t}^{0}$ and taking expectation in equation (\ref{partial_mu1_t1}), we have 
\begin{align*}
	E\left[\partial_{t_{1}}\mu_{d,d^{\prime},a}\left(\mathbf{t}_{0}\right)\left[g_{1d}\left(X\right)-g_{1d}^{0}\left(X\right)\right]\right] & =E\left[\frac{1}{\left(g_{1d}^{0}\left(X\right)\right)^{2}}\times\left[E\left[1\left\{ Y\leq a\right\} 1\left\{ D=d\right\} |X\right]\right.\right.\\
	& \left.\left.-g_{1d}^{0}\left(X\right)g_{4ad}^{0}\left(d,X\right)\right]\left[g_{1d}\left(X\right)-g_{1d}^{0}\left(X\right)\right]\right]\\
	& =0,
\end{align*}
since the term 
\begin{align*}
	g_{1d}^{0}\left(X\right)g_{4ad}^{0}\left(d,X\right) & =P_{Y,D,M,X}^{0}\left(Y\leq a,d|X\right)=E\left[1\left\{ D=d\right\} Y_{a}|X\right].
\end{align*}
Taking expectations in equation (\ref{partial_mu1_t4}), we have 
\begin{align*}
	E\left[\partial_{t_{4}}\mu_{d,d^{\prime},a}\left(\mathbf{t}_{0}\right)\left[g_{4ad}\left(d,X\right)-g_{4ad}^{0}\left(d,X\right)\right]\right] & =E\left[E\left[1-\frac{1\left\{ D=d\right\} }{g_{1d}^{0}\left(X\right)}|X\right]\left[g_{4ad}\left(d^{\prime},X\right)-g_{4ad}^{0}\left(d^{\prime},X\right)\right]\right]\\
	& =0,
\end{align*}
since 
\[
E\left[1-\frac{1\left\{ D=d\right\} }{g_{1d}^{0}\left(X\right)}|X\right]=1-\frac{g_{1d}^{0}\left(X\right)}{g_{1d}^{0}\left(X\right)}=0.
\]

Combining all previous results, it follows that 
\[
\left\Vert \partial_{v}E\left[\psi_{d,d^{\prime},a}\left(W_{a},v_{a}^{0}\right)\left[v-v_{a}^{0}\right]\right]\right\Vert =0
\]
holds for each $\left(d,d^{\prime}\right)\in\left\{ 0,1\right\} ^{2}$
and all $a\in\mathcal{A}$. Therefore, 
\[
\left\Vert \partial_{v}E\left[\boldsymbol{\psi}_{a}\left(W_{a},v_{a}^{0}\right)\left[v-v_{a}^{0}\right]\right]\right\Vert =0
\]
holds all $a\in\mathcal{A}$.


\subsubsection*{Verifying Assumption A.1.4a $\left(r_{n}\protect\leq\delta_{n}n^{-\frac{1}{4}}\right)$}

\subsubsection*{The case when $d\protect\neq d^{\prime}$}

We have 
\begin{align*}
	\psi_{d,d^{\prime},a}\left(W_{a},v\right)-\psi_{d,d^{\prime},a}\left(W_{a},v_{a}^{0}\right) & =\left\{ \frac{1\left\{ D=d\right\} \left[1-g_{2d}\left(M,X\right)\right]}{\left[1-g_{1d}\left(X\right)\right]g_{2d}\left(M,X\right)}-\frac{1\left\{ D=d\right\} \left[1-g_{2d}^{0}\left(M,X\right)\right]}{\left[1-g_{1d}^{0}\left(X\right)\right]g_{2d}^{0}\left(M,X\right)}\right\} 1\left\{ Y\leq a\right\} \\
	& +\left\{ \frac{1\left\{ D=d^{\prime}\right\} }{1-g_{1d}\left(X\right)}-\frac{1\left\{ D=d\right\} \left[1-g_{2d}\left(M,X\right)\right]}{\left[1-g_{1d}\left(X\right)\right]g_{2d}\left(M,X\right)}\right\} \times g_{3a}\left(d,M,X\right)\\
	& -\left\{ \frac{1\left\{ D=d^{\prime}\right\} }{1-g_{1d}^{0}\left(X\right)}-\frac{1\left\{ D=d\right\} \left[1-g_{2d}^{0}\left(M,X\right)\right]}{\left[1-g_{1d}^{0}\left(X\right)\right]g_{2d}^{0}\left(M,X\right)}\right\} \times g_{3a}^{0}\left(d,M,X\right)\\
	& +\left[1-\frac{1\left\{ D=d^{\prime}\right\} }{1-g_{1d}\left(X\right)}\right]g_{4ad}\left(d^{\prime},X\right)-\left[1-\frac{1\left\{ D=d^{\prime}\right\} }{1-g_{1d}^{0}\left(X\right)}\right]g_{4ad}^{0}\left(d^{\prime},X\right).
\end{align*}
To ease the notation, we express these nuisance 	parameters without their arguments in the following proof. Using the Minkowski inequality
yields 
\begin{align*}
	\left\Vert \psi_{d,d^{\prime},a}\left(W_{a},v\right)-\psi_{d,d^{\prime},a}\left(W_{a},v_{a}^{0}\right)\right\Vert _{P,2} & \leq\Pi_{1}+\Pi_{2}+\Pi_{3},
\end{align*}
where 
\begin{align*}
	\Pi_{1} & =\left\Vert \left[\frac{1\left\{ D=d\right\} \left(1-g_{2d}\right)}{\left(1-g_{1d}\right)g_{2d}}-\frac{1\left\{ D=d\right\} \left(1-g_{2d}^{0}\right)}{\left(1-g_{1d}^{0}\right)g_{2d}^{0}}\right]1\left\{ Y\leq a\right\} \right\Vert _{P,2},\\
	\Pi_{2} & =\left\Vert \left[\frac{1\left\{ D=d^{\prime}\right\} }{1-g_{1d}}-\frac{1\left\{ D=d\right\} \left(1-g_{2d}\right)}{\left(1-g_{1d}\right)g_{2d}}\right]\times g_{3a}-\left[\frac{1\left\{ D=d^{\prime}\right\} }{1-g_{1d}^{0}}-\frac{1\left\{ D=d\right\} \left(1-g_{2d}^{0}\right)}{\left(1-g_{1d}^{0}\right)g_{2d}^{0}}\right]\times g_{3a}^{0}\right\Vert _{P,2},\\
	\Pi_{3} & =\left\Vert \left(1-\frac{1\left\{ D=d^{\prime}\right\} }{1-g_{1d}}\right)g_{4ad}-\left(1-\frac{1\left\{ D=d^{\prime}\right\} }{1-g_{1d}^{0}}\right)g_{4ad}^{0}\right\Vert _{P,2}.
\end{align*}
In the following, Assumption 2.4 and the boundedness conditions (\ref{bound_1}),
(\ref{bound_2}), (\ref{bound_3}) and (\ref{bound_4}) are applied
to derive the relevant upper bounds. For the term $\Pi_{1}$, with
 probability $P$ at least $1-\Delta_{n}$, we have 
\begin{align*}
	\Pi_{1} & \leq\left\Vert \frac{g_{2d}^{0}-g_{2d}}{\left(1-g_{1d}\right)g_{2d}}\right\Vert _{P,2}+\left\Vert \frac{\left(1-g_{2d}^{0}\right)}{\left(1-g_{1d}\right)g_{2d}}-\frac{\left(1-g_{2d}^{0}\right)}{\left(1-g_{1d}^{0}\right)g_{2d}^{0}}\right\Vert _{P,2}\\
	& \leq\frac{1}{\varepsilon_{1}\varepsilon_{2}}\left\Vert g_{2}^{0}-g_{2}\right\Vert _{P,2}+\frac{1}{\varepsilon_{1}\varepsilon_{2}}\left\Vert g_{2}-g_{2}^{0}\right\Vert _{P,2}+\frac{1}{\varepsilon_{1}^{2}}\left\Vert g_{1}-g_{1}^{0}\right\Vert _{P,2}\\
	& \leq\left(\frac{2}{\varepsilon_{1}\varepsilon_{2}}+\frac{1}{\varepsilon_{1}^{2}}\right)\delta_{n}n^{-\frac{1}{4}}\lesssim\delta_{n}n^{-\frac{1}{4}}
\end{align*},
by making use of the fact that 
\begin{align*}
	\left|\frac{1}{\left(1-g_{1d}\right)g_{2d}}-\frac{1}{\left(1-g_{1d}^{0}\right)g_{2d}^{0}}\right| & =\left|\frac{g_{2d}^{0}-g_{2d}}{\left(1-g_{1d}\right)g_{2d}g_{2d}^{0}}+\frac{g_{1d}-g_{1d}^{0}}{\left(1-g_{1d}\right)\left(1-g_{1d}^{0}\right)g_{2d}^{0}}\right|\\
	& \leq\frac{1}{\varepsilon_{1}\varepsilon_{2}^{2}}\left|g_{2d}-g_{2d}^{0}\right|+\frac{1}{\varepsilon_{1}^{2}\varepsilon_{2}}\left|g_{1d}-g_{1d}^{0}\right|
\end{align*}
and Assumption $\left\Vert v-v_{a}^{0}\right\Vert _{P,2}\lesssim\delta_{n}n^{-\frac{1}{4}}$.
For the term $\Pi_{2}$, it is known that with probability $P$ at
least $1-\Delta_{n}$, for $a\in\mathcal{A}$, 
\begin{align*}
	\Pi_{2} & \leq\left\Vert \frac{\left(2g_{2d}-1\right)}{\left(1-g_{1d}\right)g_{2d}}\times\left(g_{3d}-g_{3d}^{0}\right)\right\Vert _{P,2}+\left\Vert \left[\frac{2g_{2d}-1}{\left(1-g_{1d}\right)g_{2d}}-\frac{2g_{2d}^{0}-1}{\left(1-g_{1d}^{0}\right)g_{2d}^{0}}\right]\times g_{3a}^{0}\right\Vert _{P,2}\\
	& \leq\frac{1-2\varepsilon_{2}}{\varepsilon_{1}\varepsilon_{2}}\left\Vert g_{3a}-g_{3a}^{0}\right\Vert _{P,2}+\frac{1-2\varepsilon_{2}}{\varepsilon_{1}\varepsilon_{2}^{2}}\left\Vert g_{2d}-g_{2d}^{0}\right\Vert _{P,2}+\frac{1-2\varepsilon_{2}}{\varepsilon_{1}^{2}\varepsilon_{2}}\left\Vert g_{1d}-g_{1d}^{0}\right\Vert _{P,2}+\frac{2}{\varepsilon_{1}\varepsilon_{2}}\left\Vert g_{2d}-g_{2d}^{0}\right\Vert _{P,2}\\
	& \leq\left(\frac{1-2\varepsilon_{2}}{\varepsilon_{1}\varepsilon_{2}}\varepsilon_{2}^{-\frac{1}{q}}+\frac{1-2\varepsilon_{2}}{\varepsilon_{1}\varepsilon_{2}^{2}}+\frac{1-2\varepsilon_{2}}{\varepsilon_{1}^{2}\varepsilon_{2}}+\frac{2}{\varepsilon_{1}\varepsilon_{2}}\right)\delta_{n}n^{-\frac{1}{4}}\lesssim\delta_{n}n^{-\frac{1}{4}},
\end{align*}
by using  
\[
\frac{2g_{2d}-1}{\left(1-g_{1d}\right)g_{2d}}\leq\frac{1-2\varepsilon_{2}}{\varepsilon_{1}\varepsilon_{2}},
\]
\begin{align*}
	\left|\frac{2g_{2d}-1}{\left(1-g_{1d}\right)g_{2d}}-\frac{2g_{2d}^{0}-1}{\left(1-g_{1d}^{0}\right)g_{2d}^{0}}\right| & =\left|\left[\frac{g_{2d}^{0}-g_{2d}}{\left(1-g_{1d}\right)g_{2d}g_{2d}^{0}}+\frac{g_{1d}-g_{1d}^{0}}{\left(1-g_{1d}\right)\left(1-g_{1d}^{0}\right)g_{2d}^{0}}\right]\times\left(2g_{2d}-1\right)+\frac{2\left(g_{2d}-g_{2d}^{0}\right)}{\left(1-g_{1d}^{0}\right)g_{2d}^{0}}\right|\\
	& \leq\frac{1-2\varepsilon_{2}}{\varepsilon_{1}\varepsilon_{2}^{2}}\left|g_{2d}-g_{2d}^{0}\right|+\frac{1-2\varepsilon_{2}}{\varepsilon_{1}^{2}\varepsilon_{2}}\left|g_{1d}-g_{1d}^{0}\right|+\frac{2}{\varepsilon_{1}\varepsilon_{2}}\left|g_{2d}-g_{2d}^{0}\right|,
\end{align*}
and Assumption $\left\Vert v-v_{a}^{0}\right\Vert _{P,2}\lesssim\delta_{n}n^{-\frac{1}{4}}$
as well as condition (\ref{bound_3}). For the term $\Pi_{3}$, we can show
that with probability $P$ at least $1-\Delta_{n}$ and for $a\in\mathcal{A}$,
\begin{align*}
	\Pi_{3} & \leq\left\Vert \frac{g_{1d}^{0}-g_{1d}}{\left(1-g_{1d}^{0}\right)\left(1-g_{1d}\right)}g_{4ad}^{0}\right\Vert _{P,2}+\left\Vert \frac{g_{4ad}^{0}}{1-g_{1d}}-\frac{g_{4ad}}{1-g_{1d}}\right\Vert _{P,2}+\left\Vert g_{4ad}-g_{4ad}^{0}\right\Vert _{P,2}\\
	& \leq\frac{1}{\varepsilon_{1}^{2}}\left\Vert g_{1d}^{0}-g_{1d}\right\Vert _{P,2}+\left(1+\frac{1}{\varepsilon_{1}}\right)\left\Vert g_{4ad}^{0}-g_{4ad}\right\Vert _{P,2}\\
	& \leq\left[\frac{1}{\varepsilon_{1}^{2}}+\varepsilon_{2}^{-\frac{1}{q}}\left(1+\frac{1}{\varepsilon_{1}}\right)\right]\delta_{n}n^{-\frac{1}{4}}\lesssim\delta_{n}n^{-\frac{1}{4}},
\end{align*}
by using $0\leq g_{4ad}^{0}\leq1$, Assumption
$\left\Vert v-v_{a}^{0}\right\Vert _{P,2}\lesssim\delta_{n}n^{-\frac{1}{4}}$
and condition (\ref{bound_4}). 

\subsubsection*{The case when $d=d^{\prime}$}

We have 
\begin{align*}
	\psi_{d,d,a}\left(W_{a},v\right)-\psi_{d,d,a}\left(W_{a},v_{a}^{0}\right) & =\left\{ \frac{1\left\{ D=d\right\} }{g_{1d}\left(X\right)}-\frac{1\left\{ D=d\right\} }{g_{1d}^{0}\left(X\right)}\right\} 1\left\{ Y\leq a\right\} \\
	& +\left[1-\frac{1\left\{ D=d^{\prime}\right\} }{g_{1d}}\right]g_{4ad}-\left[1-\frac{1\left\{ D=d\right\} }{g_{1d}^{0}}\right]g_{4ad}^{0},
\end{align*}
where $g_{4ad}=g_{4ad}\left(d,X\right)$ (not $g_{4ad}\left(d^{\prime},X\right)$).
Using the triangle inequality yields 
\begin{align*}
	\left\Vert \psi_{d,d,a}\left(W_{a},v\right)-\psi_{d,d,a}\left(W_{a},v_{a}^{0}\right)\right\Vert _{P,2} & \leq\Pi_{4}+\Pi_{5},
\end{align*}
where 
\begin{align*}
	\Pi_{4} & =\left\Vert \left[\frac{1\left\{ D=d\right\} }{g_{1d}\left(X\right)}-\frac{1\left\{ D=d\right\} }{g_{1d}^{0}\left(X\right)}\right]1\left\{ Y\leq a\right\} \right\Vert _{P,2},\\
	\Pi_{5} & =\left\Vert \left[1-\frac{1\left\{ D=d^{\prime}\right\} }{g_{1d}}\right]g_{4ad}-\left[1-\frac{1\left\{ D=d\right\} }{g_{1d}^{0}}\right]g_{4ad}^{0}\right\Vert _{P,2}.
\end{align*}
Following previous arguments, it can be shown that $\Pi_{4}\leq\varepsilon_{2}^{-2}\delta_{n}n^{-\frac{1}{4}}\lesssim\delta_{n}n^{-\frac{1}{4}}.$
Similar as for $\Pi_{3}$, we have for $\Pi_{5}$,  
\[
\Pi_{5}\leq\frac{1}{\varepsilon_{1}^{2}}\left\Vert g_{1d}^{0}-g_{1d}\right\Vert _{P,2}+\left(1+\frac{1}{\varepsilon_{1}}\right)\left\Vert g_{4ad}^{0}-g_{4ad}\right\Vert _{P,2}\leq\left[\frac{1}{\varepsilon_{1}^{2}}+\varepsilon_{2}^{-\frac{1}{q}}\left(1+\frac{1}{\varepsilon_{1}}\right)\right]\delta_{n}n^{-\frac{1}{4}}\lesssim\delta_{n}n^{-\frac{1}{4}}.
\]

Combining the previous results, we obtain that with probability $P$ 
at least $1-\Delta_{n}$ and for all $a\in\mathcal{A}$, 
\[
\left\Vert \psi_{d,d^{\prime},a}\left(W_{a},v\right)-\psi_{d,d^{\prime},a}\left(W_{a},v_{a}^{0}\right)\right\Vert _{P,2}\lesssim\delta_{n}n^{-\frac{1}{4}}
\]
holds for each $\left(d,d^{\prime}\right)\in\left\{ 0,1\right\} ^{2}$. Therefore, with probability $P$ at least $1-\Delta_{n}$ and
all $a\in\mathcal{A}$ and $v\in\mathcal{G}_{an}$, 
\[
\left\Vert \boldsymbol{\psi}_{a}\left(W_{a},v\right)-\boldsymbol{\psi}_{a}\left(W_{a},v_{a}^{0}\right)\right\Vert _{P,2}\lesssim\delta_{n}n^{-\frac{1}{4}}.
\]


\subsubsection*{Verifying Assumption A.1.4b $\left(\lambda_{n}^{\prime}\protect\leq\delta_{n}n^{-\frac{1}{2}}\right)$}

\subsubsection*{The case when $d\protect\neq d^{\prime}$}

We may write $\psi_{_{d,d^{\prime},a}}\left(W_{a};r\left(v-v_{a}^{0}\right)+v_{a}^{0}\right)$
as 
\begin{align*}
	\psi_{_{d,d^{\prime},a}}\left(W_{a};r\left(v-v_{a}^{0}\right)+v_{a}^{0}\right) & =\frac{1\left\{ D=d\right\} \left\{ 1-\left[r\left(g_{2d}-g_{2d}^{0}\right)+g_{2d}^{0}\right]\right\} }{\left\{ 1-\left[r\left(g_{1d}-g_{1d}^{0}\right)+g_{1d}^{0}\right]\right\} \left\{ r\left(g_{2d}-g_{2d}^{0}\right)+g_{2d}^{0}\right\} }\\
	& \times\left\{ Y_{a}-\left[r\left(g_{3d}-g_{3d}^{0}\right)+g_{3d}^{0}\right]\right\} \\
	& +\frac{1\left\{ D=d^{\prime}\right\} }{1-\left[r\left(g_{1d}-g_{1d}^{0}\right)+g_{1d}^{0}\right]}\\
	& \times\left\{ r\left[\left(g_{3d}-g_{4ad}\right)-\left(g_{3a}^{0}-g_{4ad}^{0}\right)\right]+\left(g_{3a}^{0}-g_{4ad}^{0}\right)\right\} \\
	& +\left[r\left(g_{4ad}-g_{4ad}^{0}\right)+g_{4ad}^{0}\right].
\end{align*}
Let 
\begin{align*}
	A_{1}\left(r\right) & =1-\left[r\left(g_{2d}-g_{2d}^{0}\right)+g_{2d}^{0}\right],\text{ }A_{2}\left(r\right)=1-\left[r\left(g_{1d}-g_{1d}^{0}\right)+g_{1d}^{0}\right],\\
	A_{3}\left(r\right) & =r\left(g_{2d}-g_{2d}^{0}\right)+g_{2d}^{0},\text{ }A_{4}\left(r\right)=Y_{a}-\left[r\left(g_{3a}-g_{3a}^{0}\right)+g_{3a}^{0}\right],\\
	A_{5}\left(r\right) & =r\left[\left(g_{3a}-g_{4ad}\right)-\left(g_{3a}^{0}-g_{4ad}^{0}\right)\right]+\left(g_{3a}^{0}-g_{4ad}^{0}\right).
\end{align*}
Notice that the functions $A_{i}\left(r\right)$, $i=1,\ldots,5$
are also functions of the random variables $\left(Y_{a},M,X\right)$.
For this reason, we may rewrite $\psi_{_{d,d^{\prime},a}}\left(W_{a};r\left(v-v_{a}^{0}\right)+v_{a}^{0}\right)$
 as 
\begin{align*}
	\psi_{_{d,d^{\prime},a}}\left(W_{a};r\left(v-v_{a}^{0}\right)+v_{a}^{0}\right) & =\frac{1\left\{ D=d\right\} A_{1}\left(r\right)}{A_{2}\left(r\right)A_{3}\left(r\right)}\times A_{4}\left(r\right)\\
	& +\frac{1\left\{ D=d^{\prime}\right\} }{A_{2}\left(r\right)}\times A_{5}\left(r\right)+\left[r\left(g_{4ad}-g_{4ad}^{0}\right)+g_{4ad}^{0}\right].
\end{align*}
After some calculations, we obtain 
\begin{align}
	\left.\frac{1}{2}E\left[\partial_{r}^{2}\left(\psi_{_{d,d^{\prime},a}}\left(W_{a};r\left(v-v_{a}^{0}\right)+v_{a}^{0}\right)\right)\right]\right|_{r=\bar{r}} & =-E\left[\frac{1\left\{ D=d\right\} A_{4}\left(\bar{r}\right)}{\left(A_{2}\left(\bar{r}\right)\right)^{2}A_{3}\left(\bar{r}\right)}\left(g_{1d}-g_{1d}^{0}\right)^{2}\right]\nonumber \\
	& +E\left[\frac{1\left\{ D=d\right\} A_{4}\left(\bar{r}\right)}{A_{2}\left(\bar{r}\right)\left(A_{3}\left(\bar{r}\right)\right)^{2}}\left(g_{2d}-g_{2d}^{0}\right)\left(g_{1d}-g_{1d}^{0}\right)\right]\nonumber \\
	& +E\left[\frac{1\left\{ D=d\right\} }{A_{2}\left(\bar{r}\right)A_{3}\left(\bar{r}\right)}\left(g_{2d}-g_{2d}^{0}\right)\left(g_{3a}-g_{3a}^{0}\right)\right]\nonumber \\
	& +E\left[\frac{1\left\{ D=d\right\} A_{1}\left(\bar{r}\right)A_{4}\left(\bar{r}\right)}{\left(A_{2}\left(\bar{r}\right)\right)^{3}A_{3}\left(\bar{r}\right)}\left(g_{1d}-g_{1d}^{0}\right)^{2}\right]\nonumber \\
	& -E\left[\frac{1\left\{ D=d\right\} A_{1}\left(\bar{r}\right)A_{4}\left(\bar{r}\right)}{\left(A_{2}\left(\bar{r}\right)\right)^{2}\left(A_{3}\left(\bar{r}\right)\right)^{2}}\left(g_{1d}-g_{1d}^{0}\right)\left(g_{2d}-g_{2d}^{0}\right)\right]\label{sec_der_Geteau}\\
	& +E\left[\frac{1\left\{ D=d\right\} A_{1}\left(\bar{r}\right)}{\left(A_{2}\left(\bar{r}\right)\right)^{2}A_{3}\left(\bar{r}\right)}\left(g_{1d}-g_{1d}^{0}\right)\left(g_{3a}-g_{3a}^{0}\right)\right]\nonumber \\
	& +E\left[\frac{1\left\{ D=d\right\} A_{1}\left(\bar{r}\right)A_{4}\left(\bar{r}\right)}{A_{2}\left(\bar{r}\right)\left(A_{3}\left(\bar{r}\right)\right)^{3}}\left(g_{2d}-g_{2d}^{0}\right)^{2}\right]\nonumber \\
	& +E\left[\frac{1\left\{ D=d\right\} A_{1}\left(\bar{r}\right)}{A_{2}\left(\bar{r}\right)\left(A_{3}\left(\bar{r}\right)\right)^{2}}\left(g_{2d}-g_{2d}^{0}\right)\left(g_{3a}-g_{3a}^{0}\right)\right]\nonumber \\
	& +E\left[\frac{1\left\{ D=d^{\prime}\right\} A_{5}\left(\bar{r}\right)}{\left(A_{2}\left(\bar{r}\right)\right)^{3}}\left(g_{1d}-g_{1d}^{0}\right)^{2}\right]\nonumber \\
	& +E\left[\frac{1\left\{ D=d^{\prime}\right\} A_{5}\left(\bar{r}\right)}{\left(A_{2}\left(\bar{r}\right)\right)^{2}}\left(g_{1d}-g_{1d}^{0}\right)\left[g_{3a}-g_{4ad}-\left(g_{3a}^{0}-g_{4ad}^{0}\right)\right]\right].\nonumber 
\end{align}
To bound the expectation of the second order derivative above, we can
use the properties of $A_{i}\left(r\right)$. Using Assumption 2.2 and
acknowledging that $\bar{r}\in\left(0,1\right)$, we have that
\begin{align*}
	\varepsilon_{2}<A_{1}\left(\bar{r}\right) & =1-\left[r\left(g_{2d}-g_{2d}^{0}\right)+g_{2d}^{0}\right]<1-\varepsilon_{2},\\
	\varepsilon_{1}<A_{2}\left(\bar{r}\right) & =1-\left[r\left(g_{1d}-g_{1d}^{0}\right)+g_{1d}^{0}\right]<1-\varepsilon_{1},\\
	\varepsilon_{2}<A_{3}\left(\bar{r}\right) & =r\left(g_{2d}-g_{2d}^{0}\right)+g_{2d}^{0}<1-\varepsilon_{1}
\end{align*}
hold with probability one. Also $\left|A_{4}\left(\bar{r}\right)\right|$
and $\left|A_{5}\left(\bar{r}\right)\right|$ are bounded by constants, since $\left|Y_{a}\right|$, $g_{3a}$, $g_{3a}^{0}$,
$g_{4ad}$ and $g_{4ad}^{0}$ are all bounded and $\bar{r}\in\left(0,1\right)$.
Based on these results and Assumption 2.4, it can be shown that the absolute
values of those terms on the right hand side of equation (\ref{sec_der_Geteau})
that involve interaction terms are all bounded by $\delta_{n}n^{-1/2}$.
We now consider the terms on the right hand side of equation (\ref{sec_der_Geteau})
that involve quadratic terms (the first, fourth, seventh and
ninth terms). By the assumption $\left\Vert v-v_{a}^{0}\right\Vert _{P,2}\leq\delta_{n}n^{-1/4}$,
we have $\left\Vert g_{1d}-g_{1d}^{0}\right\Vert _{P,2}\lesssim\delta_{n}n^{-1/4}$
and $\left\Vert g_{2d}-g_{2d}^{0}\right\Vert _{P,2}\lesssim\delta_{n}n^{-1/4}$.
Concerning the first term, 
\begin{align}
	\left|-E\left[\frac{1\left\{ D=d\right\} A_{4}\left(\bar{r}\right)}{\left(A_{2}\left(\bar{r}\right)\right)^{2}A_{3}\left(\bar{r}\right)}\left(g_{1d}-g_{1d}^{0}\right)^{2}\right]\right| & \leq\left|-E\left[\frac{1\left\{ D=d\right\} \left(Y_{a}-g_{3a}^{0}\right)}{\left(A_{2}\left(\bar{r}\right)\right)^{2}A_{3}\left(\bar{r}\right)}\left(g_{1d}-g_{1d}^{0}\right)^{2}\right]\right|\label{sec_der_Gateau1}\\
	& +\left|E\left[\frac{1\left\{ D=d\right\} \bar{r}\left(g_{3a}-g_{3a}^{0}\right)}{\left(A_{2}\left(\bar{r}\right)\right)^{2}A_{3}\left(\bar{r}\right)}\left(g_{1d}-g_{1d}^{0}\right)^{2}\right]\right|\nonumber \\
	& \leq\frac{\bar{r}\left(1-2\varepsilon_{1}\right)}{\varepsilon_{1}^{2}\varepsilon_{2}}\varepsilon_{2}^{-\frac{1}{q}}\delta_{n}n^{-\frac{1}{2}}\lesssim\delta_{n}n^{-\frac{1}{2}},\nonumber 
\end{align}
since
\begin{align*}
	E\left[\frac{1\left\{ D=d\right\} \left(Y_{a}-g_{3a}^{0}\right)}{\left(A_{2}\left(\bar{r}\right)\right)^{2}A_{3}\left(\bar{r}\right)}\left(g_{1d}-g_{1d}^{0}\right)^{2}\right] & =E\left[\frac{\left(g_{1}-g_{1}^{0}\right)^{2}}{\left(A_{2}\left(\bar{r}\right)\right)^{2}A_{3}\left(\bar{r}\right)}\times\left(P_{Y,D|M,X}^{0}\left(Y\leq a,D=d|M,X\right)-g_{2d}^{0}g_{3a}^{0}\right)\right]\\
	& =0,
\end{align*}
and
\begin{align*}
	\left|E\left[\frac{1\left\{ D=d\right\} \bar{r}\left(g_{3a}-g_{3a}^{0}\right)}{\left(A_{2}\left(\bar{r}\right)\right)^{2}A_{3}\left(\bar{r}\right)}\left(g_{1d}-g_{1d}^{0}\right)^{2}\right]\right| & \leq E\left[\left|\frac{1\left\{ D=d\right\} }{\left(A_{2}\left(\bar{r}\right)\right)^{2}A_{3}\left(\bar{r}\right)}\right|\bar{r}\left|g_{3a}-g_{3a}^{0}\right|\left(g_{1d}-g_{1d}^{0}\right)^{2}\right]\\
	& \leq\frac{\bar{r}\left(1-2\varepsilon_{1}\right)}{\varepsilon_{1}^{2}\varepsilon_{2}}\left\Vert g_{3a}-g_{3a}^{0}\right\Vert _{P,2}\left\Vert g_{1d}-g_{1d}^{0}\right\Vert _{P,2}\\
	& \leq\frac{\bar{r}\left(1-2\varepsilon_{1}\right)}{\varepsilon_{1}^{2}\varepsilon_{2}}\varepsilon_{2}^{-\frac{1}{q}}\delta_{n}n^{-\frac{1}{2}}\lesssim\delta_{n}n^{-\frac{1}{2}},
\end{align*}
by assuming that $\left\Vert v-v_{a}^{0}\right\Vert _{P,2}\leq\delta_{n}n^{-1/2}$
and $\left|g_{1d}-g_{1d}^{0}\right|<1-2\varepsilon_{1}$ with probability
one. Applying a similar argument to the fourth term,
\begin{align*}
	\left|E\left[\frac{1\left\{ D=d\right\} A_{1}\left(\bar{r}\right)A_{4}\left(\bar{r}\right)}{\left(A_{2}\left(\bar{r}\right)\right)^{3}A_{3}\left(\bar{r}\right)}\left(g_{1d}-g_{1d}^{0}\right)^{2}\right]\right| & \leq\left|E\left[\frac{A_{1}\left(\bar{r}\right)\left(g_{1d}-g_{1d}^{0}\right)^{2}}{\left(A_{2}\left(\bar{r}\right)\right)^{3}A_{3}\left(\bar{r}\right)}\times E\left[1\left\{ D=d\right\} \left(Y_{a}-g_{3a}^{0}\right)|M,X\right]\right]\right|\\
	& +E\left[\left|\frac{A_{1}\left(\bar{r}\right)}{\left(A_{2}\left(\bar{r}\right)\right)^{3}A_{3}\left(\bar{r}\right)}\right|\bar{r}\left|g_{3a}-g_{3a}^{0}\right|\left(g_{1d}-g_{1d}^{0}\right)^{2}\right]\\
	& \leq0+\frac{\bar{r}\left(1-\varepsilon_{2}\right)\left(1-2\varepsilon_{1}\right)}{\varepsilon_{1}^{3}\varepsilon_{2}}\varepsilon_{2}^{-\frac{1}{q}}\delta_{n}n^{-\frac{1}{2}}\lesssim\delta_{n}n^{-\frac{1}{2}}.
\end{align*}
For the ninth term, 
\begin{align*}
	\left|E\left[\frac{1\left\{ D=d^{\prime}\right\} A_{5}\left(\bar{r}\right)}{\left(A_{2}\left(\bar{r}\right)\right)^{3}}\left(g_{1d}-g_{1d}^{0}\right)^{2}\right]\right| & \leq\left|E\left[\frac{\left(g_{1d}-g_{1d}^{0}\right)^{2}}{\left(A_{2}\left(\bar{r}\right)\right)^{3}}E\left[1\left\{ D=d^{\prime}\right\} \left(g_{3a}^{0}-g_{4ad}^{0}\right)|X\right]\right]\right|\\
	& +E\left[\left|\frac{1\left\{ D=d^{\prime}\right\} }{\left(A_{2}\left(\bar{r}\right)\right)^{3}}\right|\bar{r}\left|\left(g_{3a}-g_{3a}^{0}\right)-\left(g_{4ad}-g_{4ad}^{0}\right)\right|\left(g_{1d}-g_{1d}^{0}\right)^{2}\right]\\
	& \leq\frac{\bar{r}\left(1-2\varepsilon_{1}\right)}{\varepsilon_{1}^{3}}\left(\varepsilon_{2}^{-\frac{1}{q}}\delta_{n}n^{-\frac{1}{2}}+\varepsilon_{1}^{-\frac{1}{q}}\delta_{n}n^{-\frac{1}{2}}\right)\lesssim\delta_{n}n^{-\frac{1}{2}},
\end{align*}
since the term 
\begin{align*}
	E\left[1\left\{ D=d^{\prime}\right\} \left(g_{3a}^{0}-g_{4ad}^{0}\right)|X\right] & =\int F_{Y|D,M,X}\left(a|d,m,X\right)f_{M|D|X}\left(m|d^{\prime}X\right)f_{D|X}\left(d^{\prime}|X\right)dm\\
	& -f_{D|X}\left(d^{\prime}|X\right)E\left[F_{Y|D,M,X}\left(a|d,M,X\right)|d^{\prime},X\right]\\
	& =0.
\end{align*}
For the seventh term, 
\begin{align*}
	\left|E\left[\frac{1\left\{ D=d\right\} A_{1}\left(\bar{r}\right)A_{4}\left(\bar{r}\right)}{A_{2}\left(\bar{r}\right)\left(A_{3}\left(\bar{r}\right)\right)^{3}}\left(g_{2d}-g_{2d}^{0}\right)^{2}\right]\right| & \leq0+\frac{\bar{r}\left(1-\varepsilon_{2}\right)\left(1-2\varepsilon_{2}\right)}{\varepsilon_{1}\varepsilon_{2}^{3}}\left\Vert g_{3a}-g_{3a}^{0}\right\Vert _{P,2}\left\Vert g_{1d}-g_{1d}^{0}\right\Vert _{P,2}\\
	& \leq\frac{\bar{r}\left(1-\varepsilon_{2}\right)\left(1-2\varepsilon_{2}\right)}{\varepsilon_{1}\varepsilon_{2}^{3}}\varepsilon_{2}^{-\frac{1}{q}}\delta_{n}n^{-\frac{1}{2}}\lesssim\delta_{n}n^{-\frac{1}{2}},
\end{align*}
by $\left|g_{2d}-g_{2d}^{0}\right|<1-2\varepsilon_{2}$ with probability
one.

\subsubsection*{The case when $d=d^{\prime}$}

Let 
\begin{align*}
	A_{6}\left(r\right) & =r\left(g_{1d}-g_{1d}^{0}\right)+g_{1d}^{0},\\
	A_{7}\left(r\right) & =Y_{a}-\left[r\left(g_{4ad}-g_{4ad}^{0}\right)+g_{4ad}^{0}\right],
\end{align*}
where $g_{4ad}=g_{4ad}\left(d,X\right)$. It holds that $\varepsilon_{1}<A_{6}\left(\bar{r}\right)<1-\varepsilon_{1}$
and $\left|A_{7}\left(\bar{r}\right)\right|$ for $\bar{r}\in\left[0,1\right]$
are bounded, since $\left|Y_{a}\right|$, $g_{4ad}$ and $g_{4ad}^{0}$
are bounded. Then, 
\begin{align*}
	\psi_{_{d,d,a}}\left(W_{a};r\left(v-v_{a}^{0}\right)+v_{a}^{0}\right) & =\frac{1\left\{ D=d\right\} }{r\left(g_{1d}-g_{1d}^{0}\right)+g_{1d}^{0}}\left\{ Y_{a}-\left[r\left(g_{4ad}-g_{4ad}^{0}\right)+g_{4ad}^{0}\right]\right\} +\left[r\left(g_{4ad}-g_{4ad}^{0}\right)+g_{4ad}^{0}\right]\\
	& =\frac{1\left\{ D=d\right\} }{A_{6}\left(\bar{r}\right)}A_{7}\left(\bar{r}\right)+\left[r\left(g_{4ad}-g_{4ad}^{0}\right)+g_{4ad}^{0}\right].
\end{align*}
After some calculations, we obtain
\begin{align*}
	\left.\frac{1}{2}E\left[\partial_{r}^{2}\left(\psi_{_{d,d,a}}\left(W_{a};r\left(v-v_{a}^{0}\right)+v_{a}^{0}\right)\right)\right]\right|_{r=\bar{r}} & =E\left[\frac{1\left\{ D=d\right\} A_{7}\left(\bar{r}\right)}{A_{6}\left(\bar{r}\right)^{3}}\left(g_{1d}-g_{1d}^{0}\right)^{2}\right]\\
	& +E\left[\frac{1\left\{ D=d\right\} }{A_{6}\left(\bar{r}\right)^{2}}\left(g_{1d}-g_{1d}^{0}\right)\left(g_{4ad}-g_{4ad}^{0}\right)\right].
\end{align*}
Considering the first term on the right hand side, we have that
\begin{align*}
	\left|E\left[\frac{1\left\{ D=d\right\} A_{7}\left(\bar{r}\right)}{A_{6}\left(\bar{r}\right)^{3}}\left(g_{1d}-g_{1d}^{0}\right)^{2}\right]\right| & \leq\left|E\left[\frac{1}{A_{6}\left(\bar{r}\right)^{3}}\left(g_{1d}-g_{1d}^{0}\right)^{2}\left(E\left[1\left\{ D=d\right\} Y_{a}|X\right]-g_{1d}^{0}g_{4ad}^{0}\right)\right]\right|\\
	& +E\left[\left|\frac{1\left\{ D=d\right\} \bar{r}}{A_{6}\left(\bar{r}\right)^{3}}\right|\left|g_{1d}-g_{1d}^{0}\right|\left|g_{1d}-g_{1d}^{0}\right|\left|g_{4ad}-g_{4ad}^{0}\right|\right]\\
	& \leq\frac{\bar{r}\left(1-2\varepsilon_{1}\right)}{\varepsilon_{1}^{3}}\varepsilon_{1}^{-\frac{1}{q}}\delta_{n}n^{-\frac{1}{2}}\lesssim\delta_{n}n^{-\frac{1}{2}},
\end{align*}
since $E\left[1\left\{ D=d\right\} Y_{a}|X\right]-g_{1d}^{0}g_{4ad}^{0}=0$
and

\begin{align*}
	\left\Vert g_{1d}-g_{1d}^{0}\right\Vert _{P,2}\left\Vert g_{4ad}-g_{4ad}^{0}\right\Vert _{P,2} & \leq\left\Vert g_{1d}\left(X\right)-g_{1d}^{0}\left(X\right)\right\Vert _{P,2}\left\Vert g_{4ad}\left(D,X\right)-g_{4ad}^{0}\left(D,X\right)\right\Vert _{P,2}\varepsilon_{1}^{-\frac{1}{q}}\\
	& \leq\varepsilon_{1}^{-\frac{1}{q}}\delta_{n}n^{-\frac{1}{2}}\lesssim\delta_{n}n^{-\frac{1}{2}}.
\end{align*}
Concerning the second term on the right hand side, 
\begin{align*}
	\left|E\left[\frac{1\left\{ D=d\right\} }{A_{6}\left(\bar{r}\right)^{2}}\left(g_{1d}-g_{1d}^{0}\right)\left(g_{4ad}-g_{4ad}^{0}\right)\right]\right| & \leq\frac{1}{\varepsilon_{1}^{2}}\left\Vert g_{1d}-g_{1d}^{0}\right\Vert _{P,2}\left\Vert g_{4ad}-g_{4ad}^{0}\right\Vert _{P,2}\\
	& \lesssim\delta_{n}n^{-\frac{1}{2}}.
\end{align*}

Finally, if $v=v_{a}^{0}$, it is trivial to see that $\left.E\left[\partial_{r}^{2}\left(\psi_{_{d,d^{\prime},a}}\left(W_{a};r\left(v-v_{a}^{0}\right)+v_{a}^{0}\right)\right)\right]\right|_{r=\bar{r}}=0$
for each $\left(d,d^{\prime}\right)\in\left\{ 0,1\right\} ^{2}$,
$a\in\mathcal{A}$ and $\bar{r}\in\left(0,1\right)$. Combining the previous results, it follows that with probability $P$ at least
$1-\Delta_{n}$, for $\bar{r}\in\left(0,1\right)$, all $a\in\mathcal{A}$
and $v\in\mathcal{G}_{an}\cup v_{a}^{0}$, we have that 
\[
\left|\left.E\left[\partial_{r}^{2}\left(\psi_{_{d,d^{\prime},a}}\left(W_{a};r\left(v-v_{a}^{0}\right)+v_{a}^{0}\right)\right)\right]\right|_{r=\bar{r}}\right|\lesssim\delta_{n}n^{-\frac{1}{2}},
\]
and this result holds for each $\left(d,d^{\prime}\right)\in\left\{ 0,1\right\} ^{2}$.
Therefore, with probability $P$ at least $1-\Delta_{n}$, 
\[
\left\Vert \left.E\left[\partial_{r}^{2}\left(\boldsymbol{\psi}_{a}\left(W_{a};r\left(v-v_{a}^{0}\right)+v_{a}^{0}\right)\right)\right]\right|_{r=\bar{r}}\right\Vert \lesssim\delta_{n}n^{-\frac{1}{2}}
\]
holds for $\bar{r}\in\left(0,1\right)$, all $a\in\mathcal{A}$ and
$v\in\mathcal{G}_{an}\cup v_{a}^{0}$.


\subsubsection*{Verifying Assumption A.1.5 (Smoothness condition)}

\subsubsection*{The case when $d\protect\neq d^{\prime}$}

We may write $\psi_{d,d^{\prime},a}\left(W_{a},v_{a}^{0}\right)-\psi_{d,d^{\prime},\bar{a}}\left(W_{\bar{a}},v_{\bar{a}}^{0}\right)$
as 
\begin{align*}
	\psi_{d,d^{\prime},a}\left(W_{a},v_{a}^{0}\right)-\psi_{d,d^{\prime},\bar{a}}\left(W_{\bar{a}},v_{\bar{a}}^{0}\right) & =\frac{1\left\{ D=d\right\} \left(1-g_{2d}^{0}\right)}{\left(1-g_{1d}^{0}\right)g_{2d}^{0}}\left(Y_{a}-Y_{\bar{a}}\right)\\
	& +\left[\frac{1\left\{ D=d^{\prime}\right\} }{1-g_{1d}^{0}}-\frac{1\left\{ D=d\right\} \left(1-g_{2d}^{0}\right)}{\left(1-g_{1d}^{0}\right)g_{2d}^{0}}\right]\left(g_{3a}^{0}-g_{3\bar{a}}^{0}\right)\\
	& +\left(1-\frac{1\left\{ D=d^{\prime}\right\} }{1-g_{1}^{0}\left(d,X\right)}\right)\left(g_{4ad}^{0}-g_{4\bar{a}d}^{0}\right).
\end{align*}
Using the Minkowski inequality yields 
\begin{align*}
	\left\Vert \psi_{d,d^{\prime},a}\left(W_{a},v_{a}^{0}\right)-\psi_{d,d^{\prime},\bar{a}}\left(W_{\bar{a}},v_{\bar{a}}^{0}\right)\right\Vert _{P,2} & \leq\Pi_{1}\left(a\right)+\Pi_{2}\left(a\right)+\Pi_{3}\left(a\right),
\end{align*}
where 
\begin{align*}
	\Pi_{1}\left(a\right) & =\left\Vert \frac{1\left\{ D=d\right\} \left(1-g_{2d}^{0}\right)}{\left(1-g_{1d}^{0}\right)g_{2d}^{0}}\left(Y_{a}-Y_{\bar{a}}\right)\right\Vert _{P,2},\\
	\Pi_{2}\left(a\right) & =\left\Vert \left[\frac{1\left\{ D=d^{\prime}\right\} }{1-g_{1d}^{0}}-\frac{1\left\{ D=d\right\} \left(1-g_{2d}^{0}\right)}{\left(1-g_{1d}^{0}\right)g_{2d}^{0}}\right]\left(g_{3a}^{0}-g_{3\bar{a}}^{0}\right)\right\Vert _{P,2},\\
	\Pi_{3}\left(a\right) & =\left\Vert \left(1-\frac{1\left\{ D=d^{\prime}\right\} }{1-g_{1}^{0}\left(d,X\right)}\right)\left(g_{4ad}^{0}-g_{4\bar{a}d}^{0}\right)\right\Vert _{P,2}.
\end{align*}
For $\Pi_{1}\left(a\right)$, we note that 
\[
\left|\frac{1\left\{ D=d\right\} \left(1-g_{2d}^{0}\right)}{\left(1-g_{1d}^{0}\right)g_{2d}^{0}}\right|\leq\left|\frac{1-g_{2d}^{0}}{\left(1-g_{1d}^{0}\right)g_{2d}^{0}}\right|\leq\frac{1-\varepsilon_{2}}{\varepsilon_{1}\varepsilon_{2}}
\]
with probability one, which implies that  
\begin{align*}
	\Pi_{1}\left(a\right) & \leq\frac{1-\varepsilon_{2}}{\varepsilon_{1}\varepsilon_{2}}\left\Vert Y_{a}-Y_{\bar{a}}\right\Vert _{P,2}.
\end{align*}
For $\Pi_{2}\left(a\right)$, we note that 
\[
\left|\frac{1\left\{ D=d^{\prime}\right\} }{1-g_{1d}^{0}}-\frac{1\left\{ D=d\right\} \left(1-g_{2d}^{0}\right)}{\left(1-g_{1d}^{0}\right)g_{2d}^{0}}\right|\leq\frac{1}{\varepsilon_{1}}+\frac{1-\varepsilon_{2}}{\varepsilon_{1}\varepsilon_{2}}
\]
with probability one, which implies that  
\begin{align*}
	\Pi_{2}\left(a\right) & \leq\left(\frac{1}{\varepsilon_{1}}+\frac{1-\varepsilon_{2}}{\varepsilon_{1}\varepsilon_{2}}\right)\left\Vert g_{3a}^{0}-g_{3\bar{a}}^{0}\right\Vert _{P,2}\leq\left(\frac{1}{\varepsilon_{1}}+\frac{1-\varepsilon_{2}}{\varepsilon_{1}\varepsilon_{2}}\right)\left\Vert Y_{a}-Y_{\bar{a}}\right\Vert _{P,2},
\end{align*}
because
\begin{align*}
	\left\Vert g_{3a}^{0}-g_{3\bar{a}}^{0}\right\Vert _{P,2} & \leq\left(E\left[1\left\{ D=d\right\} \left(Y_{a}-Y_{\bar{a}}\right)^{2}\right]\right)^{\frac{1}{2}}\leq\left\Vert Y_{a}-Y_{\bar{a}}\right\Vert _{P,2}.
\end{align*}
For $\Pi_{3}\left(a\right)$, we note that 
\[
\left|1-\frac{1\left\{ D=d^{\prime}\right\} }{1-g_{1d}^{0}}\right|\leq\frac{1+\varepsilon_{1}}{\varepsilon_{1}},
\]
which implies that
\begin{align*}
	\Pi_{3}\left(a\right)\leq & \frac{1+\varepsilon_{1}}{\varepsilon_{1}}\left\Vert g_{4ad}^{0}-g_{4\bar{a}d}^{0}\right\Vert _{P,2}\leq\frac{1+\varepsilon_{1}}{\varepsilon_{1}}\left\Vert Y_{a}-Y_{\bar{a}}\right\Vert _{P,2},
\end{align*}
because
\begin{align*}
	\left\Vert g_{4ad}^{0}-g_{4\bar{a}d}^{0}\right\Vert _{P,2} & \leq\left(E\left[1\left\{ D=d^{\prime}\right\} \left(E\left[Y_{a}-Y_{\bar{a}}|d,M,X\right]\right)^{2}\right]\right)^{\frac{1}{2}}\leq\left\Vert Y_{a}-Y_{\bar{a}}\right\Vert _{P,2}.
\end{align*}
Combining the previous results, we have 
\begin{align*}
	\left\Vert \psi_{d,d^{\prime},a}\left(W_{a},v_{a}^{0}\right)-\psi_{d,d^{\prime},\bar{a}}\left(W_{\bar{a}},v_{\bar{a}}^{0}\right)\right\Vert _{P,2} & \leq\left(\frac{1-\varepsilon_{2}}{\varepsilon_{1}\varepsilon_{2}}+\frac{1}{\varepsilon_{1}}+\frac{1-\varepsilon_{2}}{\varepsilon_{1}\varepsilon_{2}}+\frac{1+\varepsilon_{1}}{\varepsilon_{1}}\right)\left\Vert Y_{a}-Y_{\bar{a}}\right\Vert _{P,2}\\
	& \lesssim\left\Vert Y_{a}-Y_{\bar{a}}\right\Vert _{P,2}.
\end{align*}
By Assumption 2.1, for each $\left(d,d^{\prime}\right)\in\left\{ 0,1\right\} ^{2}$,
we then obtain

\[
\sup_{P\in\mathcal{P}}\left\Vert \psi_{d,d^{\prime},a}\left(W_{a},v_{a}^{0}\right)-\psi_{d,d^{\prime},\bar{a}}\left(W_{\bar{a}},v_{\bar{a}}^{0}\right)\right\Vert _{P,2}\lesssim\sup_{P\in\mathcal{P}}\left\Vert Y_{a}-Y_{\bar{a}}\right\Vert _{P,2}=0
\]
as $d_{\mathcal{A}}\left(a,\bar{a}\right)\rightarrow0$. 

\subsubsection*{The case when $d=d^{\prime}$}

\begin{align*}
	\psi_{d,d,a}\left(W_{a},v_{a}^{0}\right)-\psi_{d,d,\bar{a}}\left(W_{\bar{a}},v_{\bar{a}}^{0}\right) & =\frac{1\left\{ D=d\right\} }{g_{1d}^{0}}\left(1\left\{ Y\leq a\right\} -1\left\{ Y\leq\bar{a}\right\} \right)\\
	& +\left[1-\frac{1\left\{ D=d\right\} }{g_{1d}^{0}}\right]\left(g_{4ad}^{0}-g_{4\bar{a}d}^{0}\right),
\end{align*}
where $g_{4ad}^{0}=g_{4ad}^{0}\left(d,X\right)$. By the triangle
inequality,  
\begin{align*}
	\left\Vert \psi_{d,d,a}\left(W_{a},v_{a}^{0}\right)-\psi_{d,d,\bar{a}}\left(W_{\bar{a}},v_{\bar{a}}^{0}\right)\right\Vert _{P,2} & \leq\Pi_{4}\left(a\right)+\Pi_{5}\left(a\right)\lesssim\left\Vert Y_{a}-Y_{\bar{a}}\right\Vert _{P,2},
\end{align*}
where 
\begin{align*}
	\Pi_{4}\left(a\right) & =\left\Vert \frac{1\left\{ D=d\right\} }{g_{1d}^{0}}\left(1\left\{ Y\leq a\right\} -1\left\{ Y\leq\bar{a}\right\} \right)\right\Vert _{P,2}\leq\frac{1}{\varepsilon_{1}}\left\Vert Y_{a}-Y_{\bar{a}}\right\Vert _{P,2}\\
	\Pi_{5}\left(a\right) & =\left\Vert \left[1-\frac{1\left\{ D=d\right\} }{g_{1d}^{0}}\right]\left(g_{4ad}^{0}-g_{4\bar{a}d}^{0}\right)\right\Vert _{P,2}\leq\frac{1+\varepsilon_{1}}{\varepsilon_{1}}\left\Vert Y_{a}-Y_{\bar{a}}\right\Vert _{P,2}
\end{align*}
By Assumption 2.1, for each $\left(d,d^{\prime}\right)\in\left\{ 0,1\right\} ^{2}$,
we obtain 
\[
\sup_{P\in\mathcal{P}}\left\Vert \psi_{d,d,a}\left(W_{a},v_{a}^{0}\right)-\psi_{d,d,\bar{a}}\left(W_{\bar{a}},v_{\bar{a}}^{0}\right)\right\Vert _{P,2}\lesssim\sup_{P\in\mathcal{P}}\left\Vert Y_{a}-Y_{\bar{a}}\right\Vert _{P,2}=0
\]
as $d_{\mathcal{A}}\left(a,\bar{a}\right)\rightarrow0$.

Combining the previous results, it follows 
\[
\sup_{P\in\mathcal{P}}\left\Vert \boldsymbol{\psi}_{a}\left(W_{a},v_{a}^{0}\right)-\boldsymbol{\psi}_{a}\left(W_{\bar{a}},v_{\bar{a}}^{0}\right)\right\Vert _{P,2}=0
\]
as $d_{\mathcal{A}}\left(a,\bar{a}\right)\rightarrow0$.


\subsubsection*{Verifying Assumption A.1.6}

\subsubsection*{The case when $d\protect\neq d^{\prime}$}

Let $\mathcal{G}_{1}^{0}=\left\{ g_{1d}^{0}\left(X\right):d\in\left\{ 0,1\right\} \right\} $
$\mathcal{G}_{2}^{0}=\left\{ g_{2d}^{0}\left(M,X\right):d\in\left\{ 0,1\right\} \right\} $,
\begin{align*}
	\mathcal{G}_{3}^{0} & =\left\{ g_{3a}^{0}\left(d,M,X\right):a\in\mathcal{A},d\in\left\{ 0,1\right\} \right\} ,\\
	\mathcal{G}_{4}^{0} & =\left\{ g_{4ad}^{0}\left(d^{\prime},X\right):a\in\mathcal{A},\left\{ d,d^{\prime}\right\} \in\left\{ 0,1\right\} ^{2}\right\} ,
\end{align*}
$\mathcal{G}_{5}=\left\{ Y_{a}:a\in\mathcal{A}\right\} $, $\mathcal{G}_{6}=\left\{ 1\left\{ D=d\right\} :d\in\left\{ 0,1\right\} \right\} $.
The union $\cup_{j=1}^{4}\mathcal{G}_{j}^{0}$ forms the set $\mathcal{G}_{a}$
as defined above. By our assumptions, $g_{1d}^{0}\left(X\right)$ and $g_{2d}^{0}\left(M,X\right)$
are bounded within the interval $\left(0,1\right)$ with probability
one. $g_{3a}^{0}\left(d,M,X\right)=F_{Y|D,M,X}\left(a|d,M,X\right)$
is a conditional c.d.f. and $g_{4}^{0}\left(a,d,d^{\prime}X\right)=E\left[g_{3a}^{0}\left(d,M,X\right)|d^{\prime},X\right]$
is a conditional expectation of a c.d.f., which are also
bounded for $a\in\mathcal{A}$ with probability one. The functions
$Y_{a}=1\left\{ Y\leq a\right\} $ and $1\left\{ D=d\right\} $ are
indicator functions and are bounded with probability one. In conclusion,
functions in the sets $\mathcal{G}_{j}^{0}$, $j=1,\ldots,4$ are
uniformly bounded and their envelop functions are all bounded by some
constant. By Assumption 1 and Lemma L.2 of \citet{BCFH_2017},
it can be shown that uniform covering numbers of functions in $\mathcal{G}_{3}^{0}$
and $\mathcal{G}_{4}^{0}$ are bounded\textbf{ }by $\log\left(\text{e}/\epsilon\right)\vee0$
multiplied by some constants. Uniform covering numbers of functions
in $\mathcal{G}_{1}^{0}$,$\mathcal{G}_{2}^{0}$, $\mathcal{G}_{5}$
and $\mathcal{G}_{6}$ are also bounded by $\log\left(\text{e}/\epsilon\right)\vee0$
multiplied by some constants. The function $\psi_{d,d^{\prime},a}\left(W_{a},v_{a}^{0}\right)$
is formed based on a union of functions in the sets $\mathcal{G}_{j}^{0}$,
$j=1,\ldots,4$, $\mathcal{G}_{5}$ and $\mathcal{G}_{6}$. Let $\Psi_{d,d^{\prime}}^{0}=\left\{ \psi_{d,d^{\prime},a}\left(W_{a},v_{a}^{0}\right),a\in\mathcal{A}\right\} ,$
where $\left(d,d^{\prime}\right)\in\left\{ 0,1\right\} ^{2}$. Fixing
$\left(d,d^{\prime}\right)$, we have 
\begin{align*}
	\left|\psi_{d,d^{\prime},a}\left(W_{a},v_{a}^{0}\right)\right| & \leq\left|\frac{\left(1-g_{2d}^{0}\right)}{\left(1-g_{1d}^{0}\right)g_{2}^{0}}\right|\left|Y_{a}\right|+\left|\frac{1\left\{ D=d^{\prime}\right\} }{1-g_{1}^{0}}-\frac{1\left\{ D=d\right\} \left(1-g_{2d}^{0}\right)}{\left(1-g_{1d}^{0}\right)g_{2}^{0}}\right|\left|g_{3a}^{0}\right|\\
	& +\left|1-\frac{1\left\{ D=d^{\prime}\right\} }{1-g_{1d}^{0}}\right|\left|g_{4ad}^{0}\right|\\
	& \leq\frac{1-\varepsilon_{2}}{\varepsilon_{1}\varepsilon_{2}}+\left(\frac{1}{\varepsilon_{1}}+\frac{1-\varepsilon_{2}}{\varepsilon_{1}\varepsilon_{2}}\right)\left|g_{3a}^{0}\right|+\left(\frac{1+\varepsilon_{1}}{\varepsilon_{1}}\right)\left|g_{4ad}^{0}\right|.
\end{align*}
The envelop function of $f\in\Psi_{d,d^{\prime}}^{0}$ is defined
as $\psi{}_{d,d^{\prime}}^{0}\left(W\right):=\sup_{a\in\mathcal{A},v\in\left(\cup_{j=1}^{4}\mathcal{G}_{j}^{0}\right)}\left|\psi_{d,d^{\prime},a}\left(W_{a},v\right)\right|$
and we have 
\begin{align*}
	\psi{}_{d,d^{\prime}}^{0}\left(W\right) & \leq\frac{1-\varepsilon_{2}}{\varepsilon_{1}\varepsilon_{2}}+\left(\frac{1}{\varepsilon_{1}}+\frac{1-\varepsilon_{2}}{\varepsilon_{1}\varepsilon_{2}}\right)\sup_{a\in\mathcal{A},g_{3a}^{0}\in\mathcal{G}_{3}^{0}}\left|g_{3a}^{0}\left(d,M,X\right)\right|\\
	& +\left(\frac{1+\varepsilon_{1}}{\varepsilon_{1}}\right)\sup_{a\in\mathcal{A},g_{4}^{0}\in\mathcal{G}_{4}^{0}}\left|g_{4ad}^{0}\left(d^{\prime},X\right)\right|.
\end{align*}
Using the facts that for $q\geq4$, $\left\Vert g_{3a}^{0}\left(d,M,X\right)\right\Vert _{P,q}\leq\left\Vert g_{3a}^{0}\left(D,M,X\right)\right\Vert _{P,q}\varepsilon_{2}^{-\frac{1}{q}}$
and $\left\Vert g_{4ad}^{0}\left(d^{\prime},X\right)\right\Vert _{P,q}\leq\left\Vert g_{4ad}^{0}\left(D,X\right)\right\Vert _{P,q}\varepsilon_{1}^{-\frac{1}{q}}$,
it follows that 
\begin{align*}
	\left\Vert \psi{}_{d,d^{\prime}}^{0}\left(W\right)\right\Vert _{P,q} & \leq\frac{1-\varepsilon_{2}}{\varepsilon_{1}\varepsilon_{2}}+\left(\frac{1}{\varepsilon_{1}}+\frac{1-\varepsilon_{2}}{\varepsilon_{1}\varepsilon_{2}}\right)\sup_{a\in\mathcal{A},g_{3a}^{0}\in\mathcal{G}_{3}^{0}}\left\Vert g_{3a}^{0}\left(D,M,X\right)\right\Vert _{P,q}\varepsilon_{2}^{-\frac{1}{q}}\\
	& +\left(\frac{1+\varepsilon_{1}}{\varepsilon_{1}}\right)\sup_{a\in\mathcal{A},g_{4}^{0}\in\mathcal{G}_{4}^{0}}\left\Vert g_{4ad}^{0}\left(D,X\right)\right\Vert _{P,q}\varepsilon_{1}^{-\frac{1}{q}}\\
	& \leq\frac{1-\varepsilon_{2}}{\varepsilon_{1}\varepsilon_{2}}+\left(\frac{1}{\varepsilon_{1}}+\frac{1-\varepsilon_{2}}{\varepsilon_{1}\varepsilon_{2}}\right)\varepsilon_{2}^{-\frac{1}{q}}+\left(\frac{1+\varepsilon_{1}}{\varepsilon_{1}}\right)\varepsilon_{1}^{-\frac{1}{q}}<\infty,
\end{align*}
and this property holds for $P\in\mathcal{P}$. Therefore, $\sup_{P\in\mathcal{P}}\left\Vert \psi{}_{d,d^{\prime}}^{0}\left(W\right)\right\Vert _{P,q}<\infty$
for $q\geq4$ and $f\in\Psi_{d,d^{\prime}}^{0}$ is uniformly bounded
and has a uniform covering entropy bounded by $\log\left(\text{e}/\epsilon\right)\vee0$
up to multiplication by a constant. 

\subsubsection*{The case when $d=d^{\prime}$}

The function $\psi_{d,d,a}\left(W_{a},v_{a}^{0}\right)$ is formed
by a union of functions in the sets $\mathcal{G}_{1}^{0}$, $\mathcal{G}_{4}^{0}$,
$\mathcal{G}_{5}$ and $\mathcal{G}_{6}$. Let $\Psi_{d,d}^{0}=\left\{ \psi_{d,d,a}\left(W_{a},v_{a}^{0}\right),a\in\mathcal{A}\right\} ,$
where $d\in\left\{ 0,1\right\} $. Fixing $d$, we have 
\begin{align*}
	\left|\psi_{d,d,a}\left(W_{a},v_{a}^{0}\right)\right| & \leq\left|\frac{1\left\{ D=d\right\} }{g_{1d}^{0}}\right|\left|Y_{a}\right|+\left|\frac{1\left\{ D=d\right\} }{g_{1d}^{0}}\right|\left|g_{4ad}^{0}\right|+\left|g_{4ad}^{0}\right|\\
	& \leq\frac{1}{\varepsilon_{1}}+\left(\frac{1}{\varepsilon_{1}}+1\right)\left|g_{4ad}^{0}\right|,
\end{align*}
where $g_{4ad}^{0}=g_{4ad}^{0}\left(d,X\right)$. The envelop function
of $f\in\Psi_{d,d}^{0}$ is defined as $\psi{}_{d,d}^{0}\left(W\right):=\sup_{a\in\mathcal{A},v\in\left(\mathcal{G}_{1}^{0}\cup\mathcal{G}_{4}^{0}\right)}\left|\psi_{d,d,a}\left(W_{a},v\right)\right|$
and we have 
\[
\psi{}_{d,d}^{0}\left(W\right)\leq\frac{1}{\varepsilon_{1}}+\left(\frac{1}{\varepsilon_{1}}+1\right)\sup_{a\in\mathcal{A},g_{4}^{0}\in\mathcal{G}_{4}^{0}}\left|g_{4ad}^{0}\left(d,X\right)\right|.
\]
For $q\geq4$ and $\left\Vert g_{4ad}^{0}\left(d,X\right)\right\Vert _{P,q}\leq\left\Vert g_{4ad}^{0}\left(D,X\right)\right\Vert _{P,q}\varepsilon_{1}^{-\frac{1}{q}}$,
we have 
\begin{align*}
	\left\Vert \psi{}_{d,d}^{0}\left(W\right)\right\Vert _{P,q} & \leq\frac{1}{\varepsilon_{1}}+\left(\frac{1}{\varepsilon_{1}}+1\right)\sup_{a\in\mathcal{A},g_{4}^{0}\in\mathcal{G}_{4}^{0}}\left\Vert g_{4ad}^{0}\left(d,X\right)\right\Vert _{P,q}\\
	& \leq\frac{1}{\varepsilon_{1}}+\left(\frac{1}{\varepsilon_{1}}+1\right)\sup_{a\in\mathcal{A},g_{4}^{0}\in\mathcal{G}_{4}^{0}}\left\Vert g_{4ad}^{0}\left(D,X\right)\right\Vert _{P,q}\varepsilon_{1}^{-\frac{1}{q}}\\
	&\leq \frac{1}{\varepsilon_{1}}+\left(\frac{1}{\varepsilon_{1}}+1\right)\varepsilon_{1}^{-\frac{1}{q}}<\infty,
\end{align*}
and this property holds for $P\in\mathcal{P}$. Therefore, $\sup_{P\in\mathcal{P}}\left\Vert \psi{}_{d,d}^{0}\left(W\right)\right\Vert _{P,q}<\infty$
for $q\geq4$ and $f\in\Psi_{d,d}^{0}$ is uniformly bounded and has
a uniform covering entropy bounded by $\log\left(\text{e}/\epsilon\right)\vee0$
up to multiplication by a constant.

Combining the previous results, let $\Psi^{0}=\Psi_{1,1}^{0}\cup\Psi_{1,0}^{0}\cup\Psi_{0,1}^{0}\cup\Psi_{0,0}^{0}$
be a union of $\Psi_{d,d^{\prime}}^{0}$, $\left(d,d^{\prime}\right)\in\left\{ 0,1\right\} ^{2}$.
Since $\Psi^{0}$ is a union of $\Psi_{d,d^{\prime}}^{0}$ for $\left(d,d^{\prime}\right)\in\left\{ 0,1\right\} ^{2}$,
it is a finite union of classes of functions which are uniformly bounded
and have the uniform entropies bounded by $\log\left(\text{e}/\epsilon\right)\vee0$
up to multiplication by a constant. For this reason, $\Psi^{0}$ has a uniform
covering number $\sup_{P\in\mathcal{P}}\sup_{Q}\log N\left(\epsilon,\Psi^{0},\left\Vert .\right\Vert _{Q,2}\right)\lesssim\log\left(\text{e}/\epsilon\right)\vee0$
and its envelop function 
\[
\Psi^{0}\left(W\right)=\sup_{\left(d,d^{\prime}\right)\in\left\{ 0,1\right\} ^{2},a\in\mathcal{A},v\in\cup_{j=1}^{4}\mathcal{G}_{j}^{0}}\left|\psi_{d,d^{\prime},a}\left(W_{a},v\right)\right|
\]
is also bounded. Furthermore, the uniform covering integral of $\Psi^{0}$
satisfies: 
\begin{align*}
	\int_{0}^{1}\sqrt{\sup_{Q}\log N\left(\epsilon,\Psi^{0},\left\Vert .\right\Vert _{Q,2}\right)}d\epsilon & \leq\sqrt{C}\int_{0}^{1}\sqrt{1-\log\epsilon}d\epsilon\leq\sqrt{C}\int_{0}^{1}\frac{1}{\sqrt{\epsilon}}d\epsilon<\infty
\end{align*}
by the fact that $1-\log\epsilon\leq1/\epsilon$ for all $\epsilon>0$
and $\int_{0}^{1}\epsilon^{-b}d\epsilon<\infty$ for $b<1$.

The second condition in (B.1) of \citet{BCFH_2017} can also be proven
(which corresponds to Assumption A1.4). The goal of such a proof
is the same as the aim of bounding $m_{n}$ in \citet{FHLLS_2022}
and \citet{CCDDHNR_2018}. But in their scenarios, they need to consider
$Y,$ $E\left[Y|d,M,X\right]$ and $E\left[E\left[Y|d,M,X\right]|d^{\prime},X\right]$,
which requires additional moment conditions, say $\left\Vert Y\right\Vert _{P,q}$,
$\left\Vert E\left[Y|d,M,X\right]\right\Vert _{P,q}$ and $\left\Vert E\left[E\left[Y|d,M,X\right]|d^{\prime},X\right]\right\Vert _{P,q}$.
Our scenario is less challenging since the functions $g_{3a}\left(d,M,X\right)$
(a c.d.f.) and $g_{4ad}\left(d^{\prime},X\right)$ (an expectation
of a c.d.f.) are bounded.

\subsubsection*{Verifying Assumption A.1.7}
\subsubsection*{The case when $d\protect\neq d^{\prime}$}

We define the following sets of functions: 
\begin{align*}
	\mathcal{G}_{1}\left(d\right) & :=\left\{ \begin{array}{c}
		x\longmapsto h_{1}\left(f\left(x\right)^{\top}\boldsymbol{\beta}_{1}\right):\left\Vert \boldsymbol{\beta}_{1}\right\Vert _{0}\leq s_{1}\\
		\left\Vert h_{1}\left(f\left(X\right)^{\top}\boldsymbol{\beta}_{1}\right)-g_{1d}^{0}\left(X\right)\right\Vert _{P,2}\lesssim\delta_{n}n^{-\frac{1}{4}}\\
		\left\Vert h_{1}\left(f\left(X\right)^{\top}\boldsymbol{\beta}_{1}\right)-g_{1d}^{0}\left(X\right)\right\Vert _{P,\infty}\lesssim C
	\end{array}\right\} ,\\
	\mathcal{G}_{2}\left(d\right) & :=\left\{ \begin{array}{c}
		\left(m,x\right)\longmapsto h_{2}\left(f\left(m,x\right)^{\top}\boldsymbol{\beta}_{2}\right):\left\Vert \boldsymbol{\beta}_{2}\right\Vert _{0}\leq s_{2}\\
		\left\Vert h_{2}\left(f\left(M,X\right)^{\top}\boldsymbol{\beta}_{2}\right)-g_{2d}^{0}\left(M,X\right)\right\Vert _{P,2}\lesssim\delta_{n}n^{-\frac{1}{4}}\\
		\left\Vert h_{2}\left(f\left(M,X\right)^{\top}\boldsymbol{\beta}_{2}\right)-g_{2d}^{0}\left(M,X\right)\right\Vert _{P,\infty}\lesssim C
	\end{array}\right\} ,\\
	\mathcal{G}_{3}\left(d\right) & :=\left\{ \begin{array}{c}
		\left(d,m,x\right)\longmapsto h_{3}\left(f\left(d,m,x\right)^{\top}\boldsymbol{\beta}_{3}\right):\left\Vert \boldsymbol{\beta}_{3}\right\Vert _{0}\leq s_{3}\\
		\left\Vert h_{3}\left(f\left(d,M,X\right)^{\top}\boldsymbol{\beta}_{3}\right)-g_{3a}^{0}\left(d,M,X\right)\right\Vert _{P,2}\lesssim\delta_{n}n^{-\frac{1}{4}}\\
		\left\Vert h_{3}\left(f\left(d,M,X\right)^{\top}\boldsymbol{\beta}_{3}\right)-g_{3a}^{0}\left(d,M,X\right)\right\Vert _{P,\infty}\lesssim C
	\end{array}\right\} ,\\
	\mathcal{G}_{4}\left(d^{\prime}\right) & :=\left\{ \begin{array}{c}
		\left(d^{\prime},x\right)\longmapsto h_{4}\left(f\left(d^{\prime},x\right)^{\top}\boldsymbol{\beta}_{4}\right):\left\Vert \boldsymbol{\beta}_{4}\right\Vert _{0}\leq s_{4}\\
		\left\Vert h_{4}\left(f\left(d^{\prime},x\right)^{\top}\boldsymbol{\beta}_{4}\right)-g_{4ad}^{0}\left(d^{\prime},X\right)\right\Vert _{P,2}\lesssim\delta_{n}n^{-\frac{1}{4}}\\
		\left\Vert h_{4}\left(f\left(d^{\prime},X\right)^{\top}\boldsymbol{\beta}_{4}\right)-g_{4ad}^{0}\left(d^{\prime},X\right)\right\Vert _{P,\infty}\lesssim C
	\end{array}\right\} ,
\end{align*}
where $\boldsymbol{\beta}_{i}$, $i=1,\ldots,4$ are vectors of coefficients
on different sets of conditioning variables and $\left\Vert \right\Vert _{0}$ denotes
the $l^{0}$ norm. By Assumption 2.3b, $\text{dim}\left(f\left(x\right)\right)=p\times1$,
$\text{dim}\left(f\left(m,x\right)\right)=\left(p+1\right)\times1$,
$\text{dim}\left(f\left(d,m,x\right)\right)=\left(p+2\right)\times1$,
and $\text{dim}\left(f\left(d^{\prime},x\right)\right)=\left(p+1\right)\times1$.
From Assumption 2.4 follows that with probability $P$  no less than $1-\Delta_{n}$,
$1-\hat{g}_{1d}\left(X\right)\in1-\mathcal{G}_{1}\left(d\right)$,
$\hat{g}_{2d}\left(M,X\right)\in\mathcal{G}_{2}\left(d\right)$, $\left(1-\hat{g}_{2d}\left(M,X\right)\right)\in1-\mathcal{G}_{2}\left(d\right)$.
$\hat{g}_{3a}\left(d,M,X\right)\in\mathcal{G}_{3}\left(d\right)$
and $\hat{g}_{4ad}\left(d^{\prime},X\right)\in\mathcal{G}_{4}\left(d^{\prime}\right)$.
Notice that the union $\left(1-\mathcal{G}_{1}\left(d\right)\right)\cup\mathcal{G}_{2}\left(d\right)\cup\left(1-\mathcal{G}_{2}\left(d\right)\right)\cup\mathcal{G}_{3}\left(d\right)\cup\mathcal{G}_{4}\left(d^{\prime}\right)$
forms the set $\mathcal{G}_{an}$. We consider the following sets of functions:
\begin{align*}
	\mathcal{H}_{1} & =\left\{ x\longmapsto h_{1}\left(f\left(x\right)^{\top}\boldsymbol{\beta}_{1}\right):\left\Vert \boldsymbol{\beta}_{1}\right\Vert _{0}\leq s_{1},h_{1}\in\mathcal{H}_{1}^{*}\right\} ,\\
	\mathcal{H}_{2} & =\left\{ \left(m,x\right)\longmapsto h_{2}\left(f\left(m,x\right)^{\top}\boldsymbol{\beta}_{2}\right):\left\Vert \boldsymbol{\beta}_{2}\right\Vert _{0}\leq s_{2},h_{2}\in\mathcal{H}_{2}^{*}\right\} ,\\
	\mathcal{H}_{3} & =\left\{ \left(d,m,x\right)\longmapsto h_{3}\left(f\left(d,m,x\right)^{\top}\boldsymbol{\beta}_{3}\right):\left\Vert \boldsymbol{\beta}_{3}\right\Vert _{0}\leq s_{3},h_{3}\in\mathcal{H}_{3}^{*}\right\} ,\\
	\mathcal{H}_{4} & =\left\{ \left(d^{\prime},x\right)\longmapsto h_{4}\left(f\left(d^{\prime},x\right)^{\top}\boldsymbol{\beta}_{4}\right):\left\Vert \boldsymbol{\beta}_{4}\right\Vert _{0}\leq s_{4},h_{1}\in\mathcal{H}_{4}^{*}\right\} ,
\end{align*}
where $\mathcal{H}_{i}^{*}$, $i=1,\ldots,4$ are sets containing a finite number of monotonically increasing, continuously differentiable
link functions, possibly bounded within a certain interval (say $\left(0,1\right)$).
For example, \citet{BCFH_2017} chose some commonly used
link functions: $\left\{ \mathbf{I}\text{d},\varLambda,1-\varLambda,\Phi,1-\Phi\right\} $,
where $\mathbf{I}\text{d}$ is the identity function, $\varLambda$
is the logistic link, and $\Phi$ is the probit link. In our case,
each $\mathcal{H}_{i}^{*}$ is also a subset of $\left\{ \mathbf{I}\text{d},\varLambda,1-\varLambda,\Phi,1-\Phi\right\} $.
Obviously, $1-\mathcal{G}_{1}\left(d\right)\subseteq1-\mathcal{H}_{1}$,
$\mathcal{G}_{2}\left(d\right)\subseteq\mathcal{H}_{2}$, $1-\mathcal{G}_{2}\left(d\right)\subseteq1-\mathcal{H}_{2}$
$\mathcal{G}_{3}\left(d\right)\subseteq\mathcal{H}_{3}$, and $\mathcal{G}_{4}\left(d^{\prime}\right)\subseteq\mathcal{H}_{4}$
by Assumption 2.3a. For functions in $1-\mathcal{G}_{1}\left(d\right)$
$\mathcal{G}_{2}\left(d\right)$, $1-\mathcal{G}_{2}\left(d\right)$
and $\mathcal{G}_{3}\left(d\right)$, their envelope functions are
constant and bounded. As shown in \citet{BCFH_2017}, for the set $1-\mathcal{H}_{1}$, $f\left(x\right)^{\top}\boldsymbol{\beta}_{1}$
is VC-subgraph function with VC dimension bounded by some constant
($s_{1}$), and $1-\mathcal{H}_{1}$ is a union of at most $\binom{p}{s_{1}}$
of such functions. Therefore, 
\[
\log\sup_{Q}N\left(\epsilon,\mathcal{G}_{1},\left\Vert .\right\Vert _{Q}\right)\lesssim\left(s_{1}\log p+s_{1}\log\text{e}/\epsilon\right)\vee0.
\]
Using a similar argument, we obtain 
\begin{align*}
	\log\sup_{Q}N\left(\epsilon,\mathcal{G}_{2}\cup(1-\mathcal{G}_{2}),\left\Vert .\right\Vert _{Q}\right) & \lesssim\left(s_{2}\log\left(p+1\right)+s_{2}\log\text{e}/\epsilon\right)\vee0,\\
	\log\sup_{Q}N\left(\epsilon,\mathcal{G}_{3},\left\Vert .\right\Vert _{Q}\right) & \lesssim\left(s_{3}\log\left(p+2\right)+s_{3}\log\text{e}/\epsilon\right)\vee0.
\end{align*}
Concerning the functions in $\mathcal{G}_{4}\left(d\right)$, they have an
additive linear form $f\left(d^{\prime},X\right)^{\top}\boldsymbol{\beta}_{4}=aD+X^{\top}\mathbf{b}$,
$\boldsymbol{\beta}_{4}=\left(a,\mathbf{b}\right)^{\top}$. Its envelope
function is bounded by $\boldsymbol{\beta}_{4}$ being bounded and Assumption
2.3e, implying that $\left\Vert X\right\Vert _{P,q}$ is bounded. Concerning the set $\mathcal{H}_{4}$ and 
as shown in \citet{BCFH_2017}, $f\left(d^{\prime},X\right)^{\top}\boldsymbol{\beta}_{4}$
is a VC-subgraph function with VC dimension bounded by some constant
($s_{4}$), and $\mathcal{H}_{4}$ is a union of at most $\binom{p+1}{s_{4}}$
of such functions. Therefore, 
\[
\log\sup_{Q}N\left(\epsilon,\mathcal{G}_{4},\left\Vert \right\Vert _{Q}\right)\lesssim\left(s_{4}\log\left(p+1\right)+s_{4}\log\text{e}/\epsilon\right)\vee0.
\]
Combining the previous results, it follows that 
\[
\log\sup_{Q}N\left(\epsilon,\mathcal{G}_{an},\left\Vert \right\Vert _{Q}\right)\lesssim\left(s\log p+s\log\text{e}/\epsilon\right)\vee0,
\]
where $s_{1}+s_{2}+s_{3}+s_{4}\leq s$. The set of functions
\[
\mathcal{F}_{2,n}=\left\{ \psi_{d,d^{\prime},a}\left(W_{a},v\right)-\psi_{d,d^{\prime},a}\left(W_{a},v_{a}^{0}\right):\left(d,d^{\prime}\right)\in\left\{ 0,1\right\} ^{2},a\in\mathcal{A},v\in\mathcal{G}_{an}\right\} 
\]
is a Lipschitz transformation of function sets $\mathcal{G}_{i}^{0}$ $i=1,\ldots,4$
, $\mathcal{G}_{5}$, $\mathcal{G}_{6}$ defined in the previous proof,
and $\mathcal{G}_{an}$, with bounded Lipschitz coefficients and with
a constant envelope. Therefore, 
\[
\log\sup_{Q}N\left(\epsilon,\mathcal{F}_{2,n},\left\Vert \right\Vert _{Q}\right)\lesssim\left(s\log p+s\log\frac{\text{e}}{\epsilon}\right)\vee0.
\]
With probability $P$ $1-o\left(1\right)$, we have 
\[
\sup_{\left(d,d^{\prime}\right)\in\left\{ 0,1\right\} ^{2},a\in\mathcal{A}}\left|\mathbb{G}_{N,k}\left(\psi_{d,d^{\prime},a}\left(W_{a};\hat{v}_{k,a}\right)-\psi_{d,d^{\prime},a}\left(W_{a};v_{a}^{0}\right)\right)\right|\leq\sup_{f\in\mathcal{F}_{2}}\left|\mathbb{G}_{N,k}f\right|
\]
Furthermore, we have proven that $r_{n}\lesssim\delta_{n}n^{-1/4}$
and therefore, $\sup_{f\in\mathcal{F}_{2}}\left\Vert f\right\Vert _{P,2}\lesssim r_{n}\lesssim\delta_{n}n^{-1/4}$
holds. Let $L$ be a constant and $L\geq\text{e}$, using the maximum inequality A.1 of Lemma 6.2 of \cite{CCDDHNR_2018} by
setting the parameters $\sigma=C^{\prime\prime}\delta_{n}n^{-1/4}$. $C^{\prime\prime}>1$ is some constant, $a=b=p$, where $\log p=o\left(n^{-1/3}\right)=o\left(K^{-1/3}N^{-1/3}\right)$
and $v=s$ in this maximum inequality. We obtain that 
\begin{align*}
	\sup_{f\in\mathcal{F}_{2,n}}\left|\mathbb{G}_{N,k}f\right| & \lesssim\delta_{n}n^{-\frac{1}{4}}\sqrt{s\log\left(p\vee L\vee\sigma^{-1}\right)}+\frac{s}{\sqrt{N}}\log\left(p\vee L\vee\sigma^{-1}\right)\\
	& \lesssim\delta_{n}n^{-\frac{1}{4}}\sqrt{s\log\left(p\vee n\right)}+\sqrt{Ks^{2}\log^{2}\left(p\vee n\right)n^{-1}}\\
	& \lesssim\delta_{n}\delta_{n}^{\frac{1}{4}}+\delta_{n}^{\frac{1}{2}}\log^{-1}n\lesssim\delta_{n}^{\frac{1}{2}}.
\end{align*}
by applying similar arguments as in the proof of Theorem A.1 and by Assumption
2.3c. 
\subsubsection*{The case when $d=d^{\prime}$}
As this case is a special case of $d\neq d^{\prime}$, the same results apply. 
\end{proof}

\begin{proof}[Proof of Theorem 2] The proof makes use of
	Theorem A.2. As in Theorem A.2, let $U_{n,P}^{*}:=\left(\mathbb{G}_{n}\xi\boldsymbol{\psi}_{a}\left(W_{a};v_{a}^{0}\right)-\boldsymbol{\theta}_{a}^{0}\right)_{a\in\mathcal{A}}$.
	To show that $Z_{n,P}^{*}\rightsquigarrow_{B}Z_{P}$ uniformly over
	$P\in\mathcal{P}_{n}$ in $l^{\infty}\left(\mathcal{A}\right)^{4}$,
	we first show that $\left\Vert Z_{n,P}^{*}-U_{n,P}^{*}\right\Vert =o_{P}\left(1\right)$
	and then show that $U_{n,P}^{*}\rightsquigarrow_{B}Z_{P}$ uniformly
	over $P\in\mathcal{P}_{n}$ in $l^{\infty}\left(\mathcal{A}\right)^{4}$.
	Let $Z_{n,P}^{*}\left(a\right):=\mathbb{G}_{n}\xi\left(\boldsymbol{\psi}_{a}\left(W_{a};\hat{v}_{k,a}\right)-\hat{\boldsymbol{\theta}}_{a}\right)$
	and $U_{n,P}^{*}\left(a\right):=\mathbb{G}_{n}\xi\left(\boldsymbol{\psi}_{a}\left(W_{a};v_{a}^{0}\right)-\boldsymbol{\theta}_{a}^{0}\right)$.
	We notice that since $E\left[\xi\right]=0$ and $\xi$ and $W_{a}$
	are independent, $E\left[\xi\left(\boldsymbol{\psi}_{a}\left(W_{a};v_{a}^{0}\right)-\boldsymbol{\theta}_{a}^{0}\right)\right]=0$
	and 
	\begin{align*}
		Z_{n,P}^{*}\left(a\right) & =\frac{1}{\sqrt{n}}\sum_{i=1}^{n}\xi_{i}\left(\boldsymbol{\psi}_{a}\left(W_{a,i};\hat{v}_{a}^{0}\right)-\hat{\boldsymbol{\theta}}_{a}^{0}\right),\\
		U_{n,P}^{*}\left(a\right) & =\frac{1}{\sqrt{n}}\sum_{i=1}^{n}\xi_{i}\left(\boldsymbol{\psi}_{a}\left(W_{a,i};v_{a}^{0}\right)-\boldsymbol{\theta}_{a}^{0}\right).
	\end{align*}
	It follows that
	\[
	\sup_{a\in\mathcal{A}}\left\Vert Z_{n,P}^{*}\left(a\right)-U_{n,P}^{*}\left(a\right)\right\Vert \leq\Pi_{1}+\Pi_{2},
	\]
	where 
	\begin{align*}
		\Pi_{1} & =\sup_{a\in\mathcal{A}}\left\Vert \mathbb{G}_{n}\xi\left(\boldsymbol{\psi}_{a}\left(W_{a,i};\hat{v}_{a}^{0}\right)-\boldsymbol{\psi}_{a}\left(W_{a,i};v_{a}^{0}\right)\right)\right\Vert \\
		& =\sup_{a\in\mathcal{A}}\left\Vert \sqrt{n}\frac{1}{K}\sum_{k=1}^{K}\frac{1}{\sqrt{N}}\left[\mathbb{G}_{N,k}\xi\left(\boldsymbol{\psi}_{a}\left(W_{a};\hat{v}_{k,a}^{0}\right)-\boldsymbol{\psi}_{a}\left(W_{a};v_{a}^{0}\right)\right)\right]\right\Vert \\
		& \leq\sqrt{n}\frac{1}{K}\sum_{k=1}^{K}\frac{1}{\sqrt{N}}\sup_{a\in\mathcal{A}}\left\Vert \mathbb{G}_{N,k}\xi\left(\boldsymbol{\psi}_{a}\left(W_{a};\hat{v}_{k,a}^{0}\right)-\boldsymbol{\psi}_{a}\left(W_{a};v_{a}^{0}\right)\right)\right\Vert ,
	\end{align*}
	and 
	\[
	\Pi_{2}=\sup_{a\in\mathcal{A}}\left\Vert \frac{1}{\sqrt{n}}\sum_{i=1}^{n}\xi_{i}\left(\hat{\boldsymbol{\theta}}_{a}^{0}-\boldsymbol{\theta}_{a}^{0}\right)\right\Vert \leq\sup_{a\in\mathcal{A}}\left\Vert \hat{\boldsymbol{\theta}}_{a}^{0}-\boldsymbol{\theta}_{a}^{0}\right\Vert \left|\mathbb{G}_{n}\xi\right|.
	\]
	The term $\Pi_{2}$ is $O_{p}\left(n^{-1/2}\right)$, since $\sup_{a\in\mathcal{A}}\left\Vert \hat{\boldsymbol{\theta}}_{a}^{0}-\boldsymbol{\theta}_{a}^{0}\right\Vert =O_{p}\left(n^{-1/2}\right)$
	by Theorem 1 and $\left|\mathbb{G}_{n}\xi\right|=O_{p}\left(1\right)$.
	Concerning the term $\Pi_{1}$, recall the class of functions used in the
	proof of Theorem 1:
	\[
	\mathcal{F}_{2,n}=\left\{ \psi_{d,d^{\prime},a}\left(W_{a};v\right)-\psi_{d,d^{\prime},a}\left(W_{a};v_{a}^{0}\right):\left(d,d^{\prime}\right)\in\left\{ 0,1\right\} ^{2},a\in\mathcal{A},v\in\mathcal{G}_{an}\right\} .
	\]
	with the envelop function $F_{2,n}$ is $\left|\xi\right|$ times
	a constant. In the proof of Theorem 1, we have established that the
	covering entropy of $\mathcal{F}_{2}$ obeys 
	\[
	\log\sup_{Q}N\left(\epsilon,\mathcal{F}_{2,n},\left\Vert \right\Vert _{Q}\right)\lesssim\left(s\log p+s\log\frac{\text{e}}{\epsilon}\right)\vee0.
	\]
	Furthermore, using Lemma L.1 in the appendix of \citet{BCFH_2017},
	multiplication of this class by $\xi$ does not change the entropy
	bound modulo an absolute constant, and therefore its covering entropy
	is bounded by the same order as $\log\sup_{Q}N\left(\epsilon,\mathcal{F}_{2,n},\left\Vert \right\Vert _{Q}\right)$,
	\[
	\log\sup_{Q}N\left(\epsilon\left\Vert F_{2,n}\right\Vert _{Q,2},\xi\mathcal{F}_{2,n},\left\Vert .\right\Vert _{Q,2}\right)\lesssim\left(s\log p+s\log\frac{\text{e}}{\epsilon}\right)\vee0.
	\]
	Next, we use the result $\left(E\left[\max_{i\in I_{k}}\xi_{i}^{2}\right]\right)^{1/2}\lesssim\log N$
	by $E\left[\exp\left(\left|\xi\right|\right)\right]<\infty$, and
	the maximum inequality A.1 of Lemma 6.2 of \citet{CCDDHNR_2018} by
	setting the envelope function $F_{2,n}=C^{\prime\prime\prime}\left|\xi\right|$,
	$\sigma=C^{\prime\prime}\delta_{n}n^{-1/4}$, where $C^{\prime\prime}>1$
	is some constant, $a=b=p$, where $\log p=o\left(n^{-1/3}\right)=o\left(K^{-1/3}N^{-1/3}\right)$
	and $v=s$, with  probability $P$ $1-o\left(1\right)$. We have
	\begin{align*}
		\sup_{f\in\xi\mathcal{F}_{2,n}}\left|\mathbb{G}_{N,k}f\right| & \lesssim\delta_{n}n^{-\frac{1}{4}}\sqrt{s\log\left(p\vee L\vee\sigma^{-1}\right)}+\frac{s\log N}{\sqrt{N}}\log\left(p\vee L\vee\sigma^{-1}\right)\\
		& \lesssim\delta_{n}n^{-\frac{1}{4}}\sqrt{s\log\left(p\vee L\vee\sigma^{-1}\right)}+K^{\frac{1}{2}}\frac{s\left(\log n-\log K\right)}{\sqrt{n}}\log\left(p\vee L\vee\sigma^{-1}\right)\\
		& \lesssim\delta_{n}n^{-\frac{1}{4}}\sqrt{s\log\left(p\vee n\right)}+\sqrt{\frac{s^{2}\log^{2}\left(p\vee n\right)\log^{2}n}{n}}\\
		& \lesssim\delta_{n}\delta_{n}^{\frac{1}{4}}+\delta_{n}^{\frac{1}{2}}\lesssim\delta_{n}^{\frac{1}{2}}=o_{p}\left(1\right),
	\end{align*}
	by using $\sup_{f\in\xi\mathcal{F}_{2,n}}\left\Vert f\right\Vert _{P,2}=\sup_{f\in\mathcal{F}_{2,n}}\left\Vert f\right\Vert _{P,2}\leq r_{n}\lesssim\delta_{n}n^{-1/4}$
	and Assumption 2.3c. With probability $P$ $1-o\left(1\right)$ and for 
	$\hat{v}_{k,a}\in\mathcal{V}_{an}$, it can be shown that 
	\[
	\sup_{a\in\mathcal{A}}\left\Vert \mathbb{G}_{N,k}\xi\left(\boldsymbol{\psi}_{a}\left(W_{a};\hat{v}_{k,a}\right)-\boldsymbol{\psi}_{a}\left(W_{a};v_{a}^{0}\right)\right)\right\Vert \lesssim\sup_{f\in\xi\mathcal{F}_{2,n}}\left|\mathbb{G}_{N,k}f\right|.
	\]
	Therefore we conclude that with probability $P$  $1-o\left(1\right)$,
	\[
	\frac{1}{\sqrt{N}}\sup_{a\in\mathcal{A}}\left\Vert \mathbb{G}_{N,k}\xi\left(\boldsymbol{\psi}_{a}\left(W_{a};\hat{v}_{k,a}\right)-\boldsymbol{\psi}_{a}\left(W_{a};v_{a}^{0}\right)\right)\right\Vert \lesssim K^{\frac{1}{2}}n^{-\frac{1}{2}}o_{p}\left(1\right)\lesssim o_{p}\left(n^{-\frac{1}{2}}\right),
	\]
	and since $K$ is fixed and finite, 
	\begin{align*}
		\sup_{a\in\mathcal{A}}\left\Vert \mathbb{G}_{n}\xi\left(\boldsymbol{\psi}_{a}\left(W_{a};\hat{v}_{k,a}\right)-\boldsymbol{\psi}_{a}\left(W_{a};v_{a}^{0}\right)\right)\right\Vert  & \lesssim\sqrt{n}o_{p}\left(n^{-\frac{1}{2}}\right)=o_{p}\left(1\right),
	\end{align*}
	which implies that $\left\Vert Z_{n,P}^{*}-U_{n,P}^{*}\right\Vert =o_{p}\left(1\right)$.
	Next, we notice that $U_{n,P}^{*}$ is associated with the class of functions
	$\xi f$, where $f\in\mathcal{F}_{0}$. As shown in the proof of Theorem
	1, the class of $\mathcal{F}_{0}$ is Donsker uniformly in $P\in\mathcal{P}_{n}$
	under the required assumptions. Therefore we can invoke Theorem B.2
	of \citet{BCFH_2017} 
	and conclude that $U_{n,P}^{*}\rightsquigarrow_{B}Z_{P}$. Indeed,
	$U_{n,P}^{*}$ and $Z_{P}$ both are Gaussian processes, and share
	the same (zero) mean and the same covariance matrix. Finally, using
	a similar argument as in step 2 for proving Theorem 5.2 in the appendix
	of \citet{BCFH_2017}, it follows that $Z_{n,P}^{*}\rightsquigarrow_{B}Z_{P}$.
	Let $\text{BL}_{1}\left(l^{\infty}\left(\mathcal{A}\right)\right)$
	be the space of functions mapping the space of functions in $l^{\infty}\left(\mathcal{A}\right)$
	to $\left[0,1\right]$ with a Lipschitz norm of at most 1. Let $E_{B_{n}}$
	denote the expectation over the multiplier weights $\left(\xi_{i}\right)_{i=1}^{n}$
	when holding the data $\left(W_{i}\right)_{i=1}^{n}$ fixed. Following
	step 2 for proving Theorem 5.2 in the appendix of \citet{BCFH_2017},
	we obtain the following inequality: 
	\begin{align*}
		\sup_{h\in\text{BL}_{1}\left(l^{\infty}\left(\mathcal{A}\right)\right)}\left|E_{B_{n}}\left[h\left(Z_{n,P}^{*}\right)\right]-E_{P}\left[h\left(Z_{P}\right)\right]\right| & \leq\sup_{h\in\text{BL}_{1}\left(l^{\infty}\left(\mathcal{A}\right)\right)}\left|E_{B_{n}}\left[h\left(U_{n,P}^{*}\right)\right]-E_{P}\left[h\left(Z_{P}\right)\right]\right|\\
		& +E_{B_{n}}\left[\left\Vert Z_{n,P}^{*}-U_{n,P}\right\Vert \land2\right].
	\end{align*}
	The first term vanishes as asserted by using Theorem B.2 of \citet{BCFH_2017},
	since we have proven that $U_{n,P}^{*}\rightsquigarrow_{B}Z_{P}$.
	The second term is $o_{P}\left(1\right)$ since $E\left[\left\Vert Z_{n,P}^{*}-U_{n,P}\right\Vert \land2\right]=E\left[E_{B_{n}}\left[\left\Vert Z_{n,P}^{*}-U_{n,P}\right\Vert \land2\right]\right]\rightarrow0$
	by using the Markov inequality,  
	\begin{align*}
		P\left(E_{B_{n}}\left[\left\Vert Z_{n,P}^{*}-U_{n,P}^{*}\right\Vert \land2\right]\geq\varepsilon\right) & \leq\frac{E\left[E_{B_{n}}\left[\left\Vert Z_{n,P}^{*}-U_{n,P}^{*}\right\Vert \land2\right]\right]}{\varepsilon}\\
		& =\frac{E\left[\left\Vert Z_{n,P}^{*}-U_{n,P}^{*}\right\Vert \land2\right]}{\varepsilon}.
	\end{align*}
	As shown above, $\left\Vert Z_{n,P}^{*}-U_{n,P}^{*}\right\Vert =o_{p}\left(1\right)$,
	which implies that $\left\Vert Z_{n,P}^{*}-U_{n,P}^{*}\right\Vert \land2=o_{P}\left(1\right)$.
	Therefore, $\sup_{h\in BL_{1}\left(l^{\infty}\left(\mathcal{A}\right)\right)}\left|E_{B_{n}}\left[h\left(Z_{n,P}^{*}\right)\right]-E_{P}\left[h\left(Z_{P}\right)\right]\right|$
	vanishes and we obtain that $Z_{n,P}^{*}\rightsquigarrow_{B}Z_{P}$.
\end{proof}

Concerning the proof of Theorem 3, we first introduce the definition of uniform
Hadamard differentiability in the appendix further below. The proof relies on 
Theorems B.3 and B.4 of \citet{BCFH_2017} (restated as Theorem A.4
and A.5 in their appendix), which show that when an estimator satisfies
uniform validity, this property also holds for a transformation of
this estimator if uniform Hadamard tangential differentiability of
the transformation holds.

\begin{proof}[Proof of Theorem 3]: Since the $\phi_{\boldsymbol{\theta}}$
	satisfies uniform Hadamard tangential differentiable, and as shown
	in Theorem 1 and 2, both $Z_{n,P}\rightsquigarrow Z_{P}$ and $Z_{n,P}^{*}\rightsquigarrow Z_{P}$
	in $l^{\infty}\left(\mathcal{A}\right)^{4}$ uniformly in $P\in\mathcal{P}_{n}$.
	Therefore the proof can be completed by using Theorems A.4 and A.5
	(which are restated results of Theorems B.3 and B.4 of \citet{BCFH_2017}).
\end{proof}


\subsection{General Theorems for Uniform Validity of a K-Fold Cross Fitting Estimator Based on the EIF}
We subsequently derive some useful theorems for establishing the uniform validity of the proposed K-fold cross-fitting estimator based on the efficient influence function (EIF) 
under specific conditions (see below). We recall the notation used in Section 3 and assume that the parameter of interest, $F_{Y(d,M(d^{\prime}))}^{0}(a)$, is identified as 
\begin{equation}
	F_{Y(d,M(d^{\prime}))}^{0}(a)=\theta_{d,d^{\prime},a}^{0},\label{identification1}
\end{equation}
for $(d,d^{\prime})\in\{0,1\}^{2}$, where $\theta_{d,d^{\prime},a}^{0}:=E\left[\psi_{d,d^{\prime},a}\left(W_{a};v_{a}^{0}\right)\right]$ is an expectation of $\psi_{d,d^{\prime},a}$ evaluated with the true nuisance parameters. The estimator of 
$\theta_{d,d^{\prime},a}^{0}$ is the $K$-fold cross-fitting estimator
\[
\hat{\theta}_{d,d^{\prime},a}=\frac{1}{K}\sum_{k=1}^{K}\hat{\theta}_{d,d^{\prime},a}^{\left(k\right)},
\]
where $\hat{\theta}_{d,d^{\prime},a}^{\left(k\right)}=N^{-1}\sum_{i\in I_{k}}\psi_{d,d^{\prime},a}\left(W_{a,i};\hat{v}_{k,a}\right)$. Let $\boldsymbol{\psi}_{a}\left(W_{a},v\right)$ denote a vector containing the elements $\psi_{d,d^{\prime},a}\left(W_{a};v\right)$,
$\left(d,d^{\prime}\right)\in\left\{ 0,1\right\} ^{2}$. Let
$\boldsymbol{\theta}_{a}^{0}$, $\hat{\boldsymbol{\theta}}_{a}$, $\hat{\boldsymbol{\theta}}_{a}^{\left(k\right)}$ and $\mathbf{F}^{0}(a)$ denote vectors containing $\theta_{d,d^{\prime},a}^{0}$, $\hat{\theta}_{d,d^{\prime},a}$, $\hat{\theta}_{d,d^{\prime},a}^{(k)}$ and $F_{Y\left(d,M\left(d^{\prime}\right)\right)}^{0}\left(a\right)$ over different $\left(d,d^{\prime}\right)\in\{0,1\}^{2}$. It holds that $\boldsymbol{\theta}_{a}^{0}=E\left[\boldsymbol{\psi}_{a}\left(W_{a};v_{a}^{0}\right)\right]$
and $\hat{\boldsymbol{\theta}}_{a}=K^{-1}\sum_{k=1}^{K}\hat{\boldsymbol{\theta}}_{a}^{\left(k\right)}$. If equation (\ref{identification1}) holds for all $a\in\mathcal{A}$ and $\left(d,d^{\prime}\right)\in \{0,1\}^{2}$, $\mathbf{F}^{0}(a)=\boldsymbol{\theta}_{a}^{0}$.

The main results are stated in Theorems A.1 to A.3. Establishing these theorems relies on imposing the following high level assumptions on $\psi_{d,d^{\prime},a}\left(W_{a};v\right)$.
\begin{assumption} Consider a random element $W$, taking values
	in a measure space $\left(\mathcal{W},\mathcal{X}_{W}\right)$, with the
	law determined by a probability measure $P\in\mathcal{P}_{n}$. The
	observed data $\left(\left(W_{a,i}\right)_{a\in\mathcal{A}}\right)_{i=1}^{n}$
	consist of $n$ i.i.d.\ copies of a random element $\left(W_{a}\right)_{a\in\mathcal{A}}$
	which is generated as a suitably measurable transformation with respect
	to $W$ and $a$. Uniformly for all $3\leq n_{0}\leq n$ and $P\in\mathcal{P}_{n}$, 
	
	\item[1.] The true parameter $F_{Y(d,M(d^{\prime}))}^{0}(a)$ satisfies
	equation (\ref{identification1}), $\theta_{d,d^{\prime},a}^{0}$ is interior relative to $\Theta_{a}\subset\Theta\subset\mathbb{R}$
	for all $a\in\mathcal{A}$, $\left(d,d^{\prime}\right)\in\left\{ 0,1\right\} ^{2}$
	and $\Theta$ is a compact set. 
	
	\item[2.] For $a\in\mathcal{A}$, the map $v\longmapsto E\left[\boldsymbol{\psi}_{a}\left(W_{a};v\right)\right]$
	is twice continuously Gateaux-differentiable on $\mathcal{V}_{a}$.
	
	\item[3.] The function $\boldsymbol{\psi}_{a}\left(W_{a};v\right)$
	satisfies the following Neyman $\lambda_{n}$ near-orthogonality condition
	at $v=v_{a}^{0}$ with respect to $a\in\mathcal{A}$ and $v\in\mathcal{V}_{an}\cup\left\{ v_{a}^{0}\right\} $:
	\[
	\lambda_{n}:=\sup_{a\in\mathcal{A},v\in\mathcal{V}_{an}\cup\left\{ v_{a}^{0}\right\} }\left\Vert \partial_{v}E\left[\boldsymbol{\psi}_{a}\left(W_{a};v_{a}^{0}\right)\left[v-v_{a}^{0}\right]\right]\right\Vert \leq\delta_{n}n^{-\frac{1}{2}},
	\]
	where $\delta_{n}$ is a sequence converging to zero from above at
	a speed at most polynomial in $n$, e.g., $\delta_{n}\geq n^{-c}$
	for some $c>0$.
	
	\item[4.] The following moment conditions hold: 
	\begin{align*}
		r_{n} & :=\sup_{a\in\mathcal{A},v\in\mathcal{V}_{an}}\left\Vert \boldsymbol{\psi}_{a}\left(W_{a};v\right)-\boldsymbol{\psi}_{a}\left(W_{a};v_{a}^{0}\right)\right\Vert _{P,2}\leq\delta_{n}n^{-\frac{1}{4}}\\
		\lambda_{n}^{\prime} & :=\sup_{a\in\mathcal{A},v\in\mathcal{V}_{an}\cup\left\{ v_{a}^{0}\right\} ,\tilde{r}\in\left(0,1\right)}\left\Vert \partial_{r}^{2}E\left[\boldsymbol{\psi}_{a}\left(W_{a};r\left(v-v_{a}^{0}\right)+v_{a}^{0}\right)\right]|_{r=\tilde{r}}\right\Vert \leq\delta_{n}n^{-\frac{1}{2}},
	\end{align*}
	
	\item[5.] The following smoothness condition holds for each
	$(d,d^{\prime})\in\{0,1\}^{2}$: 
	\[
	\sup_{d_{\mathcal{A}}\left(a,\bar{a}\right)\leq\delta}E\left[\left\{ \left(\psi_{d,d^{\prime},a}\left(W_{a};v_{a}^{0}\right)-\theta_{d,d^{\prime},a}^{0}\right)-\left(\psi_{d,d^{\prime},\bar{a}}\left(W_{\bar{a}};v_{\bar{a}}^{0}\right)-\theta_{d,d^{\prime},\bar{a}}^{0}\right)\right\} ^{2}\right]\leq C\delta^{c_{1}},
	\]
	where $c_{1}$ is a constant.
	
	\item[6.] The set of functions
	\[
	\mathcal{F}_{0}=\left\{ \psi_{d,d^{\prime},a}\left(W_{a};v_{a}^{0}\right)-\theta_{d,d^{\prime},a}^{0}:\left(d,d^{\prime}\right)\in\left\{ 0,1\right\} ^{2},a\in\mathcal{A}\right\},
	\]
	expressed as a function of $W$, is suitably measurable, and has an envelope
	function 
	\[
	F_{0}\left(W\right)=\sup_{\left(d,d^{\prime}\right)\in\left\{ 0,1\right\} ^{2},a\in\mathcal{A},v\in\mathcal{V}_{a},\theta\in\Theta_{a}}\left|\psi_{d,d^{\prime},a}\left(W_{a};v\right)-\theta\right|
	\],
	which is measurable with respect to $W$, and $\left\Vert F_{0}\left(W\right)\right\Vert _{P,q}\leq C$,
	where $q\geq4$ is a fixed constant. Its uniform covering entropy
	satisfies 
	\[
	\log\sup_{Q}N\left(\epsilon\left\Vert F_{0}\right\Vert _{Q,2},\mathcal{F}_{0},\left\Vert .\right\Vert _{Q,2}\right)\leq C\log\left(\text{e}/\epsilon\right)\vee0,
	\]
	where $C>0$ is a constant, $\text{e}$ denotes $\exp(1)$ and
	$0<\epsilon\leq1$.
	
	\item[7.] The set of functions 
	\[
	\mathcal{F}_{1}=\left\{ \psi_{d,d^{\prime},a}\left(W_{a};v\right)-\theta:\left(d,d^{\prime}\right)\in\left\{ 0,1\right\} ^{2},a\in\mathcal{A},v\in\mathcal{V}_{an},\theta\in\Theta_{a}\right\} 
	\]
	is suitably measurable and has an envelope function 
	\[
	F_{1}\left(W\right)=\sup_{\left(d,d^{\prime}\right)\in\left\{ 0,1\right\} ^{2},a\in\mathcal{A},v\in\mathcal{V}_{an},\theta\in\Theta_{a}}\left|\psi_{d,d^{\prime},a}\left(W_{a};v\right)-\theta\right|
	\]
	which is measurable with respect to $W$, and $F_{1}\left(W\right)\leq F_{0}\left(W\right)$.
	Its uniform covering entropy satisfies 
	\[
	\log\sup_{Q}N\left(\epsilon\left\Vert F_{1}\right\Vert _{Q,2},\mathcal{F}_{1},\left\Vert .\right\Vert _{Q,2}\right)\leq v\log\left(b/\varepsilon\right)\vee0,
	\]
	where $v\geq1$ and $b\geq\max\{\text{e},N\}$ and $0<\epsilon\leq1$.
	
	\item[8.] Nuisance parameter estimation:\\
	Let $K$ be a fixed integer and $\Delta_{n}$ and $\tau_{n}$ be
	a sequence of positive constants converging to zero at a speed of at
	most polynomial $n$. The following conditions hold for each $n\geq3$
	and all $P\in\mathcal{P}_{n}$. Given a random subset $I_{k}$, $k=1,\ldots,K$
	of size $n/K$, the estimated nuisance parameter $\left\{ \hat{v}_{k,a,g}\right\} _{g=1}^{G}\in\mathcal{V}_{an}$
	with probability at least $1-\Delta_{n}$, where $\mathcal{V}_{an}$
	is the set of measurable maps $\left\{ v_{g}\right\} _{g=1}^{G}\in\mathcal{V}_{a}$
	such that for each $g$, $\left\Vert v_{g}-v_{a,g}^{0}\right\Vert _{P,2}\leq\tau_{n}$
	and $\sqrt{n}\tau_{n}^{2}\leq\delta_{n}$. Therefore and when denoting by $\mathcal{E}_{n}$
	the event that $\hat{v}_{k,a}\in\mathcal{V}_{an}$ for all
	$k=1,\ldots,K$, the probability of $\mathcal{E}_{n}$ is not smaller than $1-K\Delta_{n}=1-o\left(1\right)$. \end{assumption}

Let $\mathbb{G}_{n}$ denote an empirical process $\mathbb{G}_{n}f\left(W\right)=\sqrt{n}\left(E_{n}f\left(W\right)-E\left[f\left(W\right)\right]\right)$,
where $f$ is any $P\in\mathcal{P}_{n}$ integrable function on the
set $\mathcal{W}$. Let $\mathbb{G}_{P}f\left(W\right)$ denote the
limiting process of $\mathbb{G}_{n}f\left(W\right)$, which is a Gaussian
process with zero mean and a finite covariance matrix $E\left[\left(f\left(W\right)-E\left[f\left(W\right)\right]\right)\left(f\left(W\right)-E\left[f\left(W\right)\right]\right)^{\top}\right]$
under probability $P$ (the $P$-Brownian bridge). Using the previous notation and assumptions, we obtain the following result.

\begin{theorem}
	If Assumptions A.1.1 to A.1.8 hold, the K-fold cross-fitting estimator
	$\hat{\boldsymbol{\theta}}_{a}$ for estimating $\mathbf{F}^{0}(a)$ satisfies 
	\[
	\sqrt{n}\left(\hat{\boldsymbol{\theta}}_{a}-\mathbf{F}^{0}(a)\right)_{a\in\mathcal{A}}=Z_{n,P}+o_{P}\left(1\right),
	\]
	in $l^{\infty}\left(\mathcal{A}\right)^{4}$, uniformly in $P\in\mathcal{P}_{n}$,
	where $Z_{n,P}:=\left( \mathbb{G}_{n}\left(\boldsymbol{\psi}_{a}\left(W_{a};v_{a}^{0}\right)-\boldsymbol{\theta}_{a}^{0}\right)\right) _{a\in\mathcal{A}}$.
	Furthermore, 
	\[
	Z_{n,P}\rightsquigarrow Z_{P}
	\]
	in $l^{\infty}\left(\mathcal{A}\right)^{4}$, uniformly in $P\in\mathcal{P}_{n}$,
	where $Z_{P}:=\left( \mathbb{G}_{P}\left(\boldsymbol{\psi}_{a}\left(W_{a};v_{a}^{0}\right)-\boldsymbol{\theta}_{a}^{0}\right)\right) _{a\in\mathcal{A}}$
	and the paths of $a\longmapsto\mathbb{G}_{P}\left(\boldsymbol{\psi}_{a}\left(W_{a};v_{a}^{0}\right)-\boldsymbol{\theta}_{a}^{0}\right)$
	are a.s. uniformly continuous on $\left(\mathcal{A},d_{\mathcal{A}}\right)$,
	and 
	\begin{align*}
		\sup_{P\in\mathcal{P}_{n}}E\left[\sup_{a\in\mathcal{A}}\left\Vert \mathbb{G}_{P}\left(\boldsymbol{\psi}_{a}\left(W_{a};v_{a}^{0}\right)-\boldsymbol{\theta}_{a}^{0}\right)\right\Vert \right] & <\infty,\\
		\lim_{\delta\rightarrow0}\sup_{P\in\mathcal{P}_{n}}E\left[\sup_{d_{\mathcal{A}}\left(a,\bar{a}\right)}\left\Vert \mathbb{G}_{P}\left(\boldsymbol{\psi}_{a}\left(W_{a};v_{a}^{0}\right)-\boldsymbol{\theta}_{a}^{0}\right)-\mathbb{G}_{P}\left(\boldsymbol{\psi}_{\bar{a}}\left(W_{\bar{a}};v_{\bar{a}}^{0}\right)-\boldsymbol{\theta}_{\bar{a}}^{0}\right)\right\Vert \right] & =0.
	\end{align*}
\end{theorem}

Under Assumptions A.1.1 to A.1.8, we can also establish the uniform validity
of the multiplier bootstrap. Recall the multiplier bootstrap estimator:
\[
\hat{\theta}_{d,d^{\prime},a}^{*}=\hat{\theta}_{d,d^{\prime},a}+\frac{1}{n}\sum_{i=1}^{n}\xi_{i}\left(\psi_{d,d^{\prime},a}\left(W_{a,i};\hat{v}_{k,a}\right)-\hat{\theta}_{d,d^{\prime},a}\right),
\]
where $\xi$ is a random variable independent of $W_{a}$ that satisfies
$E\left[\xi\right]=0$, $Var\left(\xi\right)=1$ and $E\left[\exp\left(\left|\xi\right|\right)\right]<\infty$.
By the independence of $\xi$ and $W_{a}$, $E\left[\xi\psi_{d,d^{\prime},a}\left(W_{a};\hat{v}_{k,a}\right)\right]=E\left[\xi\right]E\left[\psi_{d,d^{\prime},a}\left(W_{a};\hat{v}_{k,a}\right)\right]=0$,
and also $E\left[\xi\hat{\theta}_{d,d^{\prime},a}\right]=0$. Therefore,
\begin{align*}
	\sqrt{n}\left(\hat{\theta}_{d,d^{\prime},a}^{*}-\hat{\theta}_{d,d^{\prime},a}\right) & =\frac{1}{\sqrt{n}}\sum_{i=1}^{n}\xi_{i}\left(\psi_{d,d^{\prime},a}\left(W_{a,i};\hat{v}_{k,a}\right)-\hat{\theta}_{d,d^{\prime},a}\right)\\
	& =\mathbb{G}_{n}\xi\left(\psi_{d,d^{\prime},a}\left(W_{a};\hat{v}_{k,a}\right)-\hat{\theta}_{d,d^{\prime},a}\right).
\end{align*}
Let $\hat{\boldsymbol{\theta}}_{a}^{*}$ denote a vector containing the multiplier bootstrap estimators $\hat{\theta}_{d,d^{\prime},a}^{*}$ over different $\left(d,d^{\prime}\right)\in \{0,1\}^{2}$. We may rewrite the
previous result in a vector form as 
\[
\sqrt{n}\left(\hat{\boldsymbol{\theta}}_{a}^{*}-\hat{\boldsymbol{\theta}}_{a}\right)=\mathbb{G}_{n}\xi\left(\boldsymbol{\psi}_{a}\left(W_{a};\hat{v}_{k,a}\right)-\hat{\boldsymbol{\theta}}_{a}\right).
\]
Furthermore, let $Z_{n,P}^{*}:=\left(\mathbb{G}_{n}\xi\left(\boldsymbol{\psi}_{a}\left(W_{a};\hat{v}_{k,a}\right)-\hat{\boldsymbol{\theta}}_{a}\right)\right)_{a\in\mathcal{A}}$ to postulate the following theorem. 

\begin{theorem}
	
	If Assumptions A.1.1 through A.1.8 hold, the large sample law $Z_{P}$
	of $Z_{n,P}$ in Theorem A.1 can be consistently approximated by the
	bootstrap law $Z_{n,P}^{*}$: 
	\[
	Z_{n,P}^{*}\rightsquigarrow_{B}Z_{P}
	\]
	uniformly over $P\in\mathcal{P}_{n}$ in $l^{\infty}\left(\mathcal{A}\right)^{4}$.
	
\end{theorem}

Let $\phi_{\tau}\left(F_{X}\right):=\inf\left\{ a\in\mathbb{R}:F_{X}\left(a\right)\geq\tau\right\} $
be the $\tau$-th quantile function of a random variable $X$ associated
with a c.d.f. $F_{X}$. The von Mises expansion of $\phi_{\tau}\left(F_{X}\right)$
(p.292 in \citet{Vaart_1998}) is given by: 
\[
\phi_{\tau}\left(E_{n}\right)-\phi_{\tau}\left(E\right)=\frac{1}{\sqrt{n}}\phi_{\tau,E}^{\prime}\left(\mathbb{G}_{n}\right)+\ldots+\frac{1}{m!}\frac{1}{n^{m/2}}\phi_{\tau,E}^{\left(k\right)}\left(\mathbb{G}_{n}\right)+\ldots,
\]
where $\phi_{\tau,E}^{\prime}\left(.\right)$ is a linear derivative
map and $\mathbb{G}_{n}$ denotes an empirical process $\mathbb{G}_{n}f\left(W\right)=\sqrt{n}\left(E_{n}f\left(W\right)-E\left[f\left(W\right)\right]\right)$.
Let $\phi_{\boldsymbol{\theta}}^{\prime}:=\left(\phi_{\tau,\boldsymbol{\theta}}^{\prime}\right)_{\tau\in\mathcal{T}}$,
where $\boldsymbol{\theta}=\left(\boldsymbol{\theta}_{a}\right)_{a\in\mathcal{A}}$.
Let $Q_{Y\left(d,M\left(d^{\prime}\right)\right)}^{0}\left(\tau\right):=\inf\left\{ a\in\mathbb{R}:\theta_{d,d^{\prime},a}^{0}\geq\tau\right\} $,
$\hat{Q}_{Y\left(d,M\left(d^{\prime}\right)\right)}\left(\tau\right):=\inf\left\{ a\in\mathbb{R}:\hat{\theta}_{d,d^{\prime},a}\geq\tau\right\} $
and $\hat{Q}_{Y\left(d,M\left(d^{\prime}\right)\right)}^{*}\left(\tau\right):=\inf\left\{ a\in\mathbb{R}:\hat{\theta}_{d,d^{\prime},a}^{*}\geq\tau\right\} $.
Let $\mathbf{Q}_{\tau}^{0}$, $\hat{\mathbf{Q}}_{\tau}$ and $\mathbf{\hat{Q}}_{\tau}^{*}$
denote the corresponding vectors containing $Q_{Y\left(d,M\left(d^{\prime}\right)\right)}^{0}\left(\tau\right)$,
$\hat{Q}_{Y\left(d,M\left(d^{\prime}\right)\right)}\left(\tau\right)$
and $\hat{Q}_{Y\left(d,M\left(d^{\prime}\right)\right)}^{*}\left(\tau\right)$
for different $(d,d^{\prime})\in\{0,1\}^{2}$, respectively. 
We then obtain the following result of uniform validity of quantile
estimation, which can be proven by invoking the functional delta theorems
(Theorems B.3 and B.4) of \citet{BCFH_2017}.

\begin{theorem}
	Under Assumptions A.1.1 to A.1.8, 
	\begin{align*}
		\sqrt{n}\left(\hat{\mathbf{Q}}_{\tau}-\mathbf{Q}_{\tau}^{0}\right)_{\tau\in\mathcal{T}} & \rightsquigarrow T_{P}:=\phi_{\boldsymbol{\theta}}^{\prime}\left(Z_{P}\right),\\
		\sqrt{n}\left(\hat{\mathbf{Q}}_{\tau}^{*}-\hat{\mathbf{Q}}_{\tau}\right)_{\tau\in\mathcal{T}} & \rightsquigarrow_{B}T_{P}:=\phi_{\boldsymbol{\theta}}^{\prime}\left(Z_{P}\right).
	\end{align*}
	uniformly over $P\in\mathcal{P}_{n}$ in $l^{\infty}\left(\mathcal{T}\right)^{4}$,
	where $\mathcal{T}\subset(0,1)$, $T_{P}$ is a zero mean tight Gaussian process for each $P\in\mathcal{P}_{n}$
	and $Z_{P}:=\left\{ \mathbb{G}_{P}\left(\boldsymbol{\psi}_{a}\left(W_{a};v_{a}^{0}\right)-\boldsymbol{\theta}_{a}^{0}\right)\right\} _{a\in\mathcal{A}}$.
\end{theorem}

\subsection{Proofs of Theorems A.1 to A.3}

In the proofs of Theorems A.1 to A.3, we will use the notation $\lesssim$
to denote ``less than equal a constant times'': $b\lesssim c$ denotes
$b\leq Bc$, where $B$ is a constant depending on Assumptions
A.1.1 to A.1.8, but not on $n_{0}\leq n$ and $P\in\mathcal{P}_{n}$.
We assume $n_{0}\leq n$ since the results are all asymptotic.

\begin{proof}[Proof of Theorem A.1] It is sufficient to establish
	the result over any sequence of induced probability measure $P_{n}\in\mathcal{P}_{n}$.
	But we will write $P=P_{n}$ to simplify the notation. Furthermore, we
	fix any $k=1,...,K$. From the definition of $\hat{\boldsymbol{\theta}}_{a}$ and under Assumption A.1.1, we obtain 
	\begin{align*}
		\sqrt{n}\left(\hat{\boldsymbol{\theta}}_{a}-\mathbf{F}^{0}(a)\right) & =  \sqrt{n}\left[\frac{1}{K}\sum_{k=1}^{K}\hat{\boldsymbol{\theta}}_{a}^{\left(k\right)}-\boldsymbol{\theta}_{a}^{0}\right]\\
		& =  \sqrt{n}\left\{ \frac{1}{K}\sum_{k=1}^{K}E_{N,k}\left[\boldsymbol{\psi}_{a}\left(W_{a};\hat{v}_{k,a}\right)-E\left[\boldsymbol{\psi}_{a}\left(W_{a};v_{a}^{0}\right)\right]\right]\right\} \\
		& =  \sqrt{n}\left\{ \frac{1}{K}\sum_{k=1}^{K}E_{N,k}\left[\boldsymbol{\psi}_{a}\left(W_{a};\hat{v}_{k,a}\right)\right]-E\left[\boldsymbol{\psi}_{a}\left(W_{a};v_{a}^{0}\right)\right]\right\} \\
		& =  \sqrt{n}\left\{ \frac{1}{K}\sum_{k=1}^{K}E_{N,k}\left[\boldsymbol{\psi}_{a}\left(W_{a};\hat{v}_{k,a}\right)\right]-\frac{1}{n}\sum_{i=1}^{n}\boldsymbol{\psi}_{a}\left(W_{a,i};v_{a}^{0}\right)\right.\\
		&   +\left.\frac{1}{n}\sum_{i=1}^{n}\left(\boldsymbol{\psi}_{a}\left(W_{a,i};v_{a}^{0}\right)-\boldsymbol{\theta}_{a}^{0}\right)-\left(E\left[\boldsymbol{\psi}_{a}\left(W_{a};v_{a}^{0}\right)-\boldsymbol{\theta}_{a}^{0}\right]\right)\right\} \\
		& =  \underbrace{\sqrt{n}\left\{ \frac{1}{K}\sum_{k=1}^{K}E_{N,k}\left[\boldsymbol{\psi}_{a}\left(W_{a};\hat{v}_{k,a}\right)\right]-\frac{1}{n}\sum_{i=1}^{N}\boldsymbol{\psi}_{a}\left(W_{a,i};v_{a}^{0}\right)\right\} }_{m_{1,n}\left(a\right)}\\
		& +\mathbb{G}_{n}\left(\boldsymbol{\psi}_{a}\left(W_{a};v_{a}^{0}\right)-\boldsymbol{\theta}_{a}^{0}\right).
	\end{align*}
	Therefore, proving $\sqrt{n}\left(\hat{\boldsymbol{\theta}}_{a}-\boldsymbol{\theta}_{a}^{0}\right)_{a\in\mathcal{A}}=Z_{n,P}+o_{P}\left(1\right)$
	uniformly over $P\in\mathcal{P}_{n}$ in $l^{\infty}\left(\mathcal{A}\right)^{4}$ is equivalent to showing that $\left(m_{1,n}\left(a\right)\right)_{a\in\mathcal{A}}=o_{P}\left(1\right)$
	uniformly over $P\in\mathcal{P}_{n}$ in $l^{\infty}\left(\mathcal{A}\right)^{4}$.
	Notice that 
	\begin{align*}
		\frac{1}{K}\sum_{k=1}^{K}E_{N,k}\left[\boldsymbol{\psi}_{a}\left(W_{a};\hat{v}_{k,a}\right)\right]-\frac{1}{n}\sum_{i=1}^{N}\boldsymbol{\psi}_{a}\left(W_{a,i};v_{a}^{0}\right) & =\frac{1}{K}\sum_{k=1}^{K}\left\{ E_{N,k}\left[\boldsymbol{\psi}_{a}\left(W_{a};\hat{v}_{k,a}\right)\right]-E_{N,k}\left[\boldsymbol{\psi}_{a}\left(W_{a};v_{a}^{0}\right)\right]\right\} 
	\end{align*}
	and 
	\begin{align*}
		\sup_{a\in\mathcal{A}}\left\Vert m_{1,n}\left(a\right)\right\Vert  & =\sqrt{n}\sup_{a\in\mathcal{A}}\left\Vert \frac{1}{K}\sum_{k=1}^{K}\left\{ E_{N,k}\left[\boldsymbol{\psi}_{a}\left(W_{a};\hat{v}_{k,a}\right)\right]-E_{N,k}\left[\boldsymbol{\psi}_{a}\left(W_{a};v_{a}^{0}\right)\right]\right\} \right\Vert \\
		& \leq\sqrt{n}\sup_{a\in\mathcal{A}}\frac{1}{K}\sum_{k=1}^{K}\left\Vert E_{N,k}\left[\boldsymbol{\psi}_{a}\left(W_{a};\hat{v}_{k,a}\right)\right]-E_{N,k}\left[\boldsymbol{\psi}_{a}\left(W_{a};v_{a}^{0}\right)\right]\right\Vert \\
		& \le\sqrt{n}\frac{1}{K}\sum_{k=1}^{K}\sup_{a\in\mathcal{A}}\left\Vert E_{N,k}\left[\boldsymbol{\psi}_{a}\left(W_{a};\hat{v}_{k,a}\right)\right]-E_{N,k}\left[\boldsymbol{\psi}_{a}\left(W_{a};v_{a}^{0}\right)\right]\right\Vert .
	\end{align*}
	Therefore, it suffices to show that 
	\[
	\sup_{a\in\mathcal{A}}\underbrace{\left\Vert E_{N,k}\left[\boldsymbol{\psi}_{a}\left(W_{a};\hat{v}_{k,a}\right)\right]-E_{N,k}\left[\boldsymbol{\psi}_{a}\left(W_{a};v_{a}^{0}\right)\right]\right\Vert }_{m_{2,N,k}\left(a\right)}=o_{P}\left(n^{-1/2}\right)
	\]
	holds, since $K$ is finite and fixed. Next, 
	\begin{align*}
		m_{2,N,k}\left(a\right) & =\left\Vert \frac{1}{N}\sum_{i\in I_{k}}\boldsymbol{\psi}_{a}\left(W_{a,i};\hat{v}_{k,a}\right)-\frac{1}{N}\sum_{i\in I_{k}}\boldsymbol{\psi}_{a}\left(W_{a,i};v_{a}^{0}\right)\right\Vert \\
		& =\left\Vert \frac{1}{N}\sum_{i\in I_{k}}\boldsymbol{\psi}_{a}\left(W_{a,i};\hat{v}_{k,a}\right)-\frac{1}{N}\sum_{i\in I_{k}}E\left[\boldsymbol{\psi}_{a}\left(W_{a,i};\hat{v}_{k,a}\right)|\left(W_{a,j}\right)_{j\in I_{k}^{c}}\right]\right.\\
		& -\left[\frac{1}{N}\sum_{i\in I_{k}}\boldsymbol{\psi}_{a}\left(W_{a,i};v_{a}^{0}\right)-\frac{1}{N}\sum_{i\in I_{k}}E\left[\boldsymbol{\psi}_{a}\left(W_{a,i};v_{a}^{0}\right)\right]\right]\\
		& \left.+\frac{1}{N}\sum_{i\in I_{k}}E\left[\boldsymbol{\psi}_{a}\left(W_{a,i};\hat{v}_{k,a}\right)|\left(W_{a,j}\right)_{j\in I_{k}^{c}}\right]-\frac{1}{N}\sum_{i\in I_{k}}E\left[\boldsymbol{\psi}_{a}\left(W_{a,i};v_{a}^{0}\right)\right]\right\Vert \\
		& \leq\frac{1}{\sqrt{N}}\underbrace{\left\Vert \mathbb{G}_{N,k}\left(\boldsymbol{\psi}_{a}\left(W_{a};\hat{v}_{k,a}\right)-\boldsymbol{\psi}_{a}\left(W_{a};v_{a}^{0}\right)\right)\right\Vert }_{m_{3,N,k}\left(a\right)}\\
		& +\underbrace{\left\Vert \frac{1}{N}\sum_{i\in I_{k}}E\left[\boldsymbol{\psi}_{a}\left(W_{a,i};\hat{v}_{k,a}\right)|\left(W_{a,j}\right)_{j\in I_{k}^{c}}\right]-E\left[\boldsymbol{\psi}_{a}\left(W_{a};v_{a}^{0}\right)\right]\right\Vert }_{m_{4,N,k}\left(a\right)},
	\end{align*}
	where $\mathbb{G}_{N,k}$ is an empirical process defined as 
	\[
	\mathbb{G}_{N,k}f\left(W\right)=\sqrt{N}\left(\frac{1}{N}\sum_{i\in I_{k}}f\left(W_{i}\right)-\int f\left(w\right)dP\right),
	\]
	and $f$ is any $P$ integrable function on $\mathcal{W}$. We note that 
	\[
	E\left[\boldsymbol{\psi}_{a}\left(W_{a,i};\hat{v}_{k,a}\right)|\left(W_{a,j}\right)_{j\in I_{k}^{c}}\right]=E\left[\boldsymbol{\psi}_{a}\left(W_{a,i};\hat{v}_{k,a}\right)\right]
	\]
	for $i\in I_{k}$, since conditional on $\left(W_{a,j}\right)_{j\in I_{k}^{c}}$,
	$\hat{v}_{k,a}$ is a constant and $\left(W_{a,i}\right)_{i\in I_{k}}$
	and $\left(W_{a,j}\right)_{i\in I_{k}^{c}}$ are independent. Then, 
	\[
	\sup_{a\in\mathcal{A}}m_{2,N,k}\left(a\right)\leq\frac{\sup_{a\in\mathcal{A}}m_{3,N,k}\left(a\right)}{\sqrt{N}}+\sup_{a\in\mathcal{A}}m_{4,N,k}\left(a\right).
	\]
	In order to bound $\sup_{a\in\mathcal{A}}m_{3,N,k}\left(a\right)$,
	we define the following class of functions:
	\[
	\mathcal{F}_{2}^{\prime}=\left\{ \psi_{d,d^{\prime},a}\left(W_{a};v\right)-\theta-\left(\psi_{d,d^{\prime},a}\left(W_{a};v_{a}^{0}\right)-\theta_{d,d^{\prime},a}^{0}\right):\left(d,d^{\prime}\right)\in\left\{ 0,1\right\} ^{2},a\in\mathcal{A},v\in\mathcal{V}_{an},\theta\in\Theta_{an}\right\} ,
	\]
	where $\Theta_{an}:=\left\{ \theta\in\Theta_{a}:\left|\theta-\theta_{d,d^{\prime},a}^{0}\right|\leq C\tau_{n}\right\} $.
	Notice that the envelope function of $\mathcal{F}_{2}^{\prime}$,
	denoted by 
	\begin{align*}
		F_{2}^{\prime}\left(W\right) & =\sup_{\left(d,d^{\prime}\right)\in\left\{ 0,1\right\} ^{2},a\in\mathcal{A},v\in\mathcal{V}_{an},\theta\in\Theta_{an}}\left|\psi_{d,d^{\prime},a}\left(W_{a};v\right)-\theta-\left(\psi_{d,d^{\prime},a}\left(W_{a};v_{a}^{0}\right)-\theta_{d,d^{\prime},a}^{0}\right)\right|\\
		& \leq\sup_{\left(d,d^{\prime}\right)\in\left\{ 0,1\right\} ^{2},a\in\mathcal{A},v\in\mathcal{V}_{an},\theta\in\Theta_{an}}\left|\psi_{d,d^{\prime},a}\left(W_{a};v\right)-\theta\right|\\
		& +\sup_{\left(d,d^{\prime}\right)\in\left\{ 0,1\right\} ^{2},a\in\mathcal{A},v\in\mathcal{V}_{a},\theta\in\Theta_{an}}\left|\psi_{d,d^{\prime},a}\left(W_{a};v\right)-\theta\right|\\
		& \leq F_{1}\left(W\right)+F_{0}\left(W\right)\\
		& \leq2F_{0}\left(W\right).
	\end{align*}
	The uniform covering entropy of $\mathcal{F}_{2}^{\prime}$: $\log\sup_{Q}N\left(\epsilon\left\Vert F_{2}^{\prime}\right\Vert _{Q,2},\mathcal{F}_{2}^{\prime},\left\Vert .\right\Vert _{Q,2}\right)$
	satisfies 
	\[
	\log\sup_{Q}N\left(\epsilon\left\Vert F_{2}^{\prime}\right\Vert _{Q,2},\mathcal{F}_{2}^{\prime},\left\Vert .\right\Vert _{Q,2}\right)\lesssim2v\left(\log\left(a/\epsilon\right)\right)\vee0.
	\]
	Next, consider another class of functions: 
	\[
	\mathcal{F}_{2}=\left\{ \psi_{d,d^{\prime},a}\left(W_{a};v\right)-\psi_{d,d^{\prime},a}\left(W_{a};v_{a}^{0}\right):\left(d,d^{\prime}\right)\in\left\{ 0,1\right\} ^{2},a\in\mathcal{A},v\in\mathcal{V}_{an},\theta\in\Theta_{an}\right\} .
	\]
	$\mathcal{F}_{2}$ is a subset of $\mathcal{F}_{2}^{\prime}$ in which
	we choose $C=0$ in $\Theta_{an}$, and for this reason, its envelope function
	$F_{2}\left(W\right)$ is bounded by $F_{2}^{\prime}\left(W\right)$.
	Therefore, the uniform covering entropy of $\mathcal{F}_{2}$: $\log\sup_{Q}N\left(\epsilon\left\Vert F_{2}\right\Vert _{Q,2},\mathcal{F}_{2},\left\Vert .\right\Vert _{Q,2}\right)$
	also satisfies 
	\[
	\log\sup_{Q}N\left(\epsilon\left\Vert F_{2}\right\Vert _{Q,2},\mathcal{F}_{2},\left\Vert .\right\Vert _{Q,2}\right)\lesssim2v\left(\log\left(a/\epsilon\right)\right)\vee0.
	\]
	With probability $P$ $1-K\Delta_{n}=1-o\left(1\right)$, we have
	\[
	\sup_{\left(d,d^{\prime}\right)\in\left\{ 0,1\right\} ^{2},a\in\mathcal{A}}\left|\mathbb{G}_{N,k}\left(\psi_{d,d^{\prime},a}\left(W_{a};\hat{v}_{k,a}\right)-\psi_{d,d^{\prime},a}\left(W_{a};v_{a}^{0}\right)\right)\right|\leq\sup_{f\in\mathcal{F}_{2}}\left|\mathbb{G}_{N,k}f\right|,
	\]
	where $\psi_{d,d^{\prime},a}\left(W_{a};\hat{v}_{k,a}\right)-\psi_{d,d^{\prime},a}\left(W_{a};v_{a}^{0}\right)$
	is an element of $\boldsymbol{\psi}_{a}\left(W_{a};\hat{v}_{k,a}\right)-\boldsymbol{\psi}_{a}\left(W_{a};v_{a}^{0}\right)$
	in $m_{3,N,k}\left(a\right)$. Furthermore, it can be shown that $\sup_{f\in\mathcal{F}_{2}}\left\Vert f\right\Vert _{P,2}\leq r_{n}$,
	where 
	\[
	r_{n}:=\sup_{a\in\mathcal{A},v\in\mathcal{V}_{an}}\left\Vert \boldsymbol{\psi}_{a}\left(W_{a};v\right)-\boldsymbol{\psi}_{a}\left(W_{a};v_{a}^{0}\right)\right\Vert _{P,2}.
	\]
	Then using the maximum inequality A.1 of Lemma 6.2 of \citet{CCDDHNR_2018}
	by setting the used parameter $\sigma=C^{\prime}\delta_{n}n^{-1/4}$,
	where $C^{\prime}>1$ is a constant, and $a=b=N$ in this maximum
	inequality, we have 
	\begin{align*}
		\sup_{f\in\mathcal{F}_{2}}\left|\mathbb{G}_{N,k}f\right| & \lesssim\delta_{n}n^{-\frac{1}{4}}\sqrt{\log\left(N\vee\sigma^{-1}\right)}+N^{\frac{1}{q}-\frac{1}{2}}\log\left(N\vee\sigma^{-1}\right)\\
		& \lesssim\delta_{n}+K^{\frac{1}{2}-\frac{1}{q}}\frac{\log n}{n^{\frac{1}{2}-\frac{1}{q}}}=o\left(1\right)
	\end{align*}
	If follows that with probability $P$ $1-o\left(1\right)$,
	\[
	\sup_{\left(d,d^{\prime}\right)\in\left\{ 0,1\right\} ^{2},a\in\mathcal{A}}\left|\mathbb{G}_{N,k}\left(\psi_{d,d^{\prime},a}\left(W_{a};\hat{v}_{k,a}\right)-\psi_{d,d^{\prime},a}\left(W_{a};v_{a}^{0}\right)\right)\right|\leq\sup_{f\in\mathcal{F}_{2}}\left|\mathbb{G}_{N,k}f\right|\lesssim o\left(1\right),
	\]
	and $\sup_{a\in\mathcal{A}}m_{3,N,k}\left(a\right)\lesssim o\left(1\right)$.
	Then, $N^{-1/2}\sup_{a\in\mathcal{A}}m_{3,N,k}\left(a\right)=n^{-1/2}K^{1/2}\sup_{a\in\mathcal{A}}m_{3,N,k}\left(a\right)=o_{P}\left(n^{-1/2}\right)$
	since $K$ is fixed and finite.
	
	Concerning the term $m_{4,N,k}\left(a\right)$, let 
	\[
	\varphi_{k}\left(r\right)=E\left[\boldsymbol{\psi}_{a}\left(W_{a};r\left(\hat{v}_{k,a}-v_{a}^{0}\right)+v_{a}^{0}\right)|\left(W_{a,j}\right)_{j\in I_{k}^{c}}\right]-E\left[\boldsymbol{\psi}_{a}\left(W_{a};v_{a}^{0}\right)\right],
	\]
	where $r\in\left(0,1\right)$. Notice that $\varphi_{k}\left(0\right)=0$,
	since 
	\[
	E\left[\boldsymbol{\psi}_{a}\left(W_{a};v_{a}^{0}\right)|\left(W_{a,j}\right)_{j\in I_{k}^{c}}\right]=E\left[\boldsymbol{\psi}_{a}\left(W_{a};v_{a}^{0}\right)\right],
	\]
	and $\varphi_{k}\left(1\right)=E\left[\boldsymbol{\psi}_{a}\left(W_{a};\hat{v}_{k,a}\right)|\left(W_{a,j}\right)_{j\in I_{k}^{c}}\right]-E\left[\boldsymbol{\psi}_{a}\left(W_{a};v_{a}^{0}\right)\right]$.
	By applying a Taylor expansion to $\varphi_{k}\left(r\right)$ around $0$,
	\[
	\varphi_{k}\left(r\right)=\varphi_{k}\left(0\right)+\varphi_{k}^{\prime}\left(0\right)r+\frac{1}{2}\varphi_{k}^{\prime\prime}\left(\bar{r}\right)r^{2}
	\]
	for some $\bar{r}\in\left(0,1\right)$. Then, $m_{4,N,k}\left(a\right)\leq\left\Vert \varphi_{k}\left(1\right)\right\Vert =\left\Vert \varphi_{k}^{\prime}\left(0\right)+1/2\varphi_{k}^{\prime\prime}\left(\bar{r}\right)\right\Vert $,
	where 
	\begin{align*}
		\varphi_{k}^{\prime}\left(0\right) & =\partial_{v}E\left[\boldsymbol{\psi}_{a}\left(W_{a};v_{a}^{0}\right)\left[\hat{v}_{k,a}-v_{a}^{0}\right]|\left(W_{a,j}\right)_{j\in I_{k}^{c}}\right]=\partial_{v}E\left[\boldsymbol{\psi}_{a}\left(W_{a};v_{a}^{0}\right)\left[\hat{v}_{k,a}-v_{a}^{0}\right]\right],\\
		\varphi_{k}^{\prime\prime}\left(\bar{r}\right) & =\partial_{r}^{2}E\left[\boldsymbol{\psi}_{a}\left(W_{a};\bar{r}\left(\hat{v}_{k,a}-v_{a}^{0}\right)+v_{a}^{0}\right)|\left(W_{a,j}\right)_{j\in I_{k}^{c}}\right],\bar{r}\in\left(0,1\right).
	\end{align*}
	Therefore, $\sup_{a\in\mathcal{A}}m_{4,N,k}\left(a\right)\leq\sup_{a\in\mathcal{A}}\left\Vert \varphi_{k}^{\prime}\left(0\right)\right\Vert +1/2\sup_{a\in\mathcal{A}}\left\Vert \varphi_{k}^{\prime\prime}\left(\bar{r}\right)\right\Vert $.
	Furthermore, $\sup_{a\in\mathcal{A}}\left\Vert \varphi_{k}^{\prime}\left(0\right)\right\Vert $
	and $\sup_{a\in\mathcal{A}}\left\Vert \varphi_{k}^{\prime\prime}\left(\bar{r}\right)\right\Vert $
	are bounded by the following terms, respectively: 
	\[
	\sup_{a\in\mathcal{A},v\in\mathcal{V}_{an}\cup\left\{ v_{a}^{0}\right\} }\left\Vert \partial_{v}E\left[\boldsymbol{\psi}_{a}\left(W_{a};v_{a}^{0}\right)\left[v-v_{a}^{0}\right]\right]\right\Vert =o\left(n^{-\frac{1}{2}}\right),
	\]
	\[
	\sup_{a\in\mathcal{A},v\in\mathcal{V}_{an}\cup\left\{ v_{a}^{0}\right\} ,\tilde{r}\in\left(0,1\right)}\left\Vert \partial_{r}^{2}E\left[\boldsymbol{\psi}_{a}\left(W_{a};\bar{r}\left(v-v_{a}^{0}\right)+v_{a}^{0}\right)\right]\right\Vert =o\left(n^{-\frac{1}{2}}\right).
	\]
	It follows that $\sup_{a\in\mathcal{A}}m_{4,N,k}\left(a\right)=o_{P}\left(n^{-1/2}\right).$
	Combining the previous results, we obtain that $\sup_{a\in\mathcal{A}}m_{2,N,k}\left(a\right)=o_{P}\left(n^{-1/2}\right)$,
	and $\sup_{a\in\mathcal{A}}\left\Vert m_{1,n}\left(a\right)\right\Vert \leq n^{1/2}o_{P}\left(n^{-1/2}\right)=o_{P}\left(1\right)$
	uniformly over $P\in\mathcal{P}_{n}$ in $l^{\infty}\left(\mathcal{A}\right)^{4}$.
	
	Finally, to show that $Z_{n,P}\rightsquigarrow Z_{P}$ in $l^{\infty}\left(\mathcal{A}\right)^{4}$
	uniformly in $P\in\mathcal{P}_{n}$, we may exploit the properties of
	functions in $\mathcal{F}_{0}$. Recall that
	$\mathcal{F}_{0}$ is suitably measurable, has an envelop function
	$F_{0}\left(W\right)$ that is measurable with respect to $\mathcal{W}$
	and satisfies $\left\Vert F_{0}\right\Vert _{P,q}=\left(E\left|F_{0}\right|^{q}\right)^{1/q}\leq C$,
	where $q\geq4$ is a fixed number. By Assumption A.1.5, functions in
	$\mathcal{F}_{0}$ satisfy 
	\[
	\sup_{P\in\mathcal{P}_{n}}E\left[\left\{ \left(\psi_{d,d^{\prime},a}\left(W_{a};v_{a}^{0}\right)-\theta_{d,d^{\prime},a}^{0}\right)-\left(\psi_{d,d^{\prime},\bar{a}}\left(W_{\bar{a}};v_{\bar{a}}^{0}\right)-\theta_{d,d^{\prime},\bar{a}}^{0}\right)\right\} ^{2}\right]\rightarrow0,
	\]
	as $d_{\mathcal{A}}\left(a,\bar{a}\right)\rightarrow0$. By Assumption
	A.1.6, uniform covering entropy of $\mathcal{F}_{0}$ satisfies 
	\[
	\sup_{Q}\log N\left(\epsilon\left\Vert F_{0}\right\Vert _{Q,2},\mathcal{F}_{0},\left\Vert .\right\Vert _{Q,2}\right)\leq C\log\left(\text{e}/\epsilon\right)\vee0.
	\]
	In fact, the uniform covering integral satisfies 
	\begin{align*}
		\int_{0}^{1}\sqrt{\sup_{Q}\log N\left(\epsilon\left\Vert F_{0}\right\Vert _{Q,2},\mathcal{F}_{0},\left\Vert .\right\Vert _{Q,2}\right)}d\epsilon & \leq\sqrt{C}\int_{0}^{1}\sqrt{1-\log\epsilon}d\epsilon\\
		& \leq\sqrt{C}\int_{0}^{1}\frac{1}{\sqrt{\epsilon}}d\epsilon<\infty,
	\end{align*}
	which follows from the result that $1-\log\epsilon\leq1/\epsilon$ for all $\epsilon>0$
	and $\int_{0}^{1}\epsilon^{-b}d\epsilon<\infty$ for $b<1$. Therefore, 
	we may invoke Theorem B.1 of \citet{BCFH_2017} 
	to obtain the result. The class of $\mathcal{F}_{0}$ is Donsker
	uniformly in $P\in\mathcal{P}_{n}$ because $\left\Vert F_{0}\right\Vert _{P,q}$
	is bounded, the entropy condition holds, and Assumption A.1.5 implies
	that $\sup_{P\in\mathcal{P}_{n}}\left\Vert \boldsymbol{\psi}_{a}\left(W_{a},v_{a}^{0}\right)-\boldsymbol{\theta}_{a}^{0}-\left(\boldsymbol{\psi}_{\bar{a}}\left(W_{\bar{a}},v_{\bar{a}}^{0}\right)-\boldsymbol{\theta}_{\bar{a}}^{0}\right)\right\Vert _{P,2}\rightarrow0$
	as $d_{\mathcal{A}}\left(a,\bar{a}\right)\rightarrow0$. \end{proof}

\begin{proof}[Proof of Theorem A.2] It is sufficient to establish
	the result over any sequence of probability measure $P_{n}\in\mathcal{P}_{n}$.
	Again, we will write $P=P_{n}$ for simplifying the notation. Let $U_{n,P}^{*}:=\left(\mathbb{G}_{n}\xi\left(\boldsymbol{\psi}_{a}\left(W_{a};v_{a}^{0}\right)-\boldsymbol{\theta}_{a}^{0}\right)\right)_{a\in\mathcal{A}}$.
	To show that $Z_{n,P}^{*}\rightsquigarrow_{B}Z_{P}$ uniformly over
	$P\in\mathcal{P}_{n}$ in $l^{\infty}\left(\mathcal{A}\right)^{4}$,
	we first show that $\left\Vert Z_{n,P}^{*}-U_{n,P}^{*}\right\Vert =o_{P}\left(1\right)$
	and then show that $U_{n,P}^{*}\rightsquigarrow_{B}Z_{P}$ uniformly
	over $P\in\mathcal{P}_{n}$ in $l^{\infty}\left(\mathcal{A}\right)^{4}$.
	To prove $\left\Vert Z_{n,P}^{*}-U_{n,P}^{*}\right\Vert =o_{P}\left(1\right)$
	uniformly over $P\in\mathcal{P}_{n}$ in $l^{\infty}\left(\mathcal{A}\right)^{4}$,
	we use a similar argument as for proving equation (E.7) in the appendix
	of \citet{BCFH_2017}. Let $Z_{n,P}^{*}\left(a\right):=\mathbb{G}_{n}\xi\left(\boldsymbol{\psi}_{a}\left(W_{a};\hat{v}_{k,a}\right)-\hat{\boldsymbol{\theta}}_{a}\right)$
	and $U_{n,P}^{*}\left(a\right):=\mathbb{G}_{n}\xi\left(\boldsymbol{\psi}_{a}\left(W_{a};v_{a}^{0}\right)-\boldsymbol{\theta}_{a}^{0}\right)$.
	We first notice that since $E\left[\xi\right]=0$ and $\xi$ and $W_{a}$
	are independent, $E\left[\xi\left(\boldsymbol{\psi}_{a}\left(W_{a};v_{a}^{0}\right)-\boldsymbol{\theta}_{a}^{0}\right)\right]=0$
	and 
	\begin{align*}
		Z_{n,P}^{*}\left(a\right) & =\frac{1}{\sqrt{n}}\sum_{i=1}^{n}\xi_{i}\left(\boldsymbol{\psi}_{a}\left(W_{a,i};\hat{v}_{a}^{0}\right)-\hat{\boldsymbol{\theta}}_{a}^{0}\right),\\
		U_{n,P}^{*}\left(a\right) & =\frac{1}{\sqrt{n}}\sum_{i=1}^{n}\xi_{i}\left(\boldsymbol{\psi}_{a}\left(W_{a,i};v_{a}^{0}\right)-\boldsymbol{\theta}_{a}^{0}\right).
	\end{align*}
	Therefore, we have
	\[
	\sup_{a\in\mathcal{A}}\left\Vert Z_{n,P}^{*}\left(a\right)-U_{n,P}^{*}\left(a\right)\right\Vert \leq\Pi_{1}+\Pi_{2},
	\]
	where 
	\begin{align*}
		\Pi_{1} & =\sup_{a\in\mathcal{A}}\left\Vert \mathbb{G}_{n}\xi\left(\boldsymbol{\psi}_{a}\left(W_{a,i};\hat{v}_{a}^{0}\right)-\boldsymbol{\psi}_{a}\left(W_{a,i};v_{a}^{0}\right)\right)\right\Vert \\
		& =\sup_{a\in\mathcal{A}}\left\Vert \sqrt{n}\frac{1}{K}\sum_{k=1}^{K}\frac{1}{\sqrt{N}}\left[\mathbb{G}_{N,k}\xi\left(\boldsymbol{\psi}_{a}\left(W_{a};\hat{v}_{k,a}^{0}\right)-\boldsymbol{\psi}_{a}\left(W_{a};v_{a}^{0}\right)\right)\right]\right\Vert \\
		& \leq\sqrt{n}\frac{1}{K}\sum_{k=1}^{K}\frac{1}{\sqrt{N}}\sup_{a\in\mathcal{A}}\left\Vert \mathbb{G}_{N,k}\xi\left(\boldsymbol{\psi}_{a}\left(W_{a};\hat{v}_{k,a}^{0}\right)-\boldsymbol{\psi}_{a}\left(W_{a};v_{a}^{0}\right)\right)\right\Vert ,
	\end{align*}
	and
	\[
	\Pi_{2}=\sup_{a\in\mathcal{A}}\left\Vert \frac{1}{\sqrt{n}}\sum_{i=1}^{n}\xi_{i}\left(\hat{\boldsymbol{\theta}}_{a}^{0}-\boldsymbol{\theta}_{a}^{0}\right)\right\Vert \leq\sup_{a\in\mathcal{A}}\left\Vert \hat{\boldsymbol{\theta}}_{a}^{0}-\boldsymbol{\theta}_{a}^{0}\right\Vert \left|\mathbb{G}_{n}\xi\right|.
	\]
	The term $\Pi_{2}$ is $O_{p}\left(n^{-1/2}\right)$, since $\sup_{a\in\mathcal{A}}\left\Vert \hat{\boldsymbol{\theta}}_{a}^{0}-\boldsymbol{\theta}_{a}^{0}\right\Vert =O_{p}\left(n^{-1/2}\right)$
	by Theorem A.1 and $\left|\mathbb{G}_{n}\xi\right|=O_{p}\left(1\right)$. Concerning
	the term $\Pi_{1}$, recall the class of functions used in the proof
	of Theorem A.1, 
	\[
	\mathcal{F}_{2}=\left\{ \psi_{d,d^{\prime},a}\left(W_{a};v\right)-\psi_{d,d^{\prime},a}\left(W_{a};v_{a}^{0}\right):\left(d,d^{\prime}\right)\in\left\{ 0,1\right\} ^{2},a\in\mathcal{A},v\in\mathcal{V}_{an}\right\},
	\]
	as well as its envelope function $F_{2}\leq2F_{0}$ and the covering entropy:
	\[
	\log\sup_{Q}N\left(\epsilon\left\Vert F_{2}\right\Vert _{Q,2},\mathcal{F}_{2},\left\Vert .\right\Vert _{Q,2}\right)\lesssim2v\left(\log\left(a/\epsilon\right)\right)\vee0.
	\]
	Using Lemma K.1 in the Appendix of \citet{BCFH_2017},
	multiplication of this class by $\xi$ does not change the entropy
	bound modulo an absolute constant, and therefore its covering entropy
	\[
	\log\sup_{Q}N\left(\epsilon\left\Vert \xi F_{2}\right\Vert _{Q,2},\xi\mathcal{F}_{2},\left\Vert .\right\Vert _{Q,2}\right)
	\]
	is bounded by the same order as $\log\sup_{Q}N\left(\epsilon\left\Vert F_{2}\right\Vert _{Q,2},\mathcal{F}_{2},\left\Vert .\right\Vert _{Q,2}\right)$.
	Next, notice that $\left(E\left[\max_{i\in I_{k}}\xi_{i}^{2}\right]\right)^{1/2}\lesssim\log N$
	by $E\left[\exp\left(\left|\xi\right|\right)\right]<\infty$. By the
	independence of $\xi_{i}$ and $W_{i}$, we have 
	\[
	\left\Vert \max_{i\in I_{k}}\xi_{i}F_{0}\left(W_{i}\right)\right\Vert _{P,2}\leq\left\Vert \max_{i\in I_{k}}\xi_{i}\right\Vert _{P,2}\left\Vert \max_{i\in I_{k}}F_{0}\left(W_{i}\right)\right\Vert _{P,2}\lesssim N^{\frac{1}{q}}\log N,
	\]
	which holds for $k=1,\ldots,K$. Using the maximum inequality
	A.1 of Lemma 6.2 of \citet{CCDDHNR_2018}, with probability $P$ $1-o\left(1\right)$,
	we obtain 
	\begin{align*}
		\sup_{f\in\xi\mathcal{F}_{2}}\left|\mathbb{G}_{N.k}f\right| & \lesssim\delta_{n}n^{-\frac{1}{4}}\sqrt{\log\left(N\vee\sigma^{-1}\right)}+\frac{N^{\frac{1}{q}}\log N}{\sqrt{N}}\log\left(N\vee\sigma^{-1}\right)\\
		& \lesssim\delta_{n}n^{-\frac{1}{4}}\sqrt{\log\left(N\vee\sigma^{-1}\right)}+K^{\frac{1}{2}-\frac{1}{q}}\frac{\left(\log n-\log K\right)}{n^{\frac{1}{2}-\frac{1}{q}}}\log\left(N\vee\sigma^{-1}\right)\\
		& \lesssim\delta_{n}+n^{\frac{1}{q}-\frac{1}{2}}\log n\log\left(n\vee\sigma^{-1}\right)=o_{p}\left(1\right),
	\end{align*}
	by using the fact that $\sup_{f\in\xi\mathcal{F}_{2}}\left\Vert f\right\Vert _{P,2}=\sup_{f\in\mathcal{F}_{2}}\left\Vert f\right\Vert _{P,2}\leq r_{n}$
	and setting the parameters $\sigma=C^{\prime}\delta_{n}n^{-\frac{1}{4}}$
	and $a=b=N$ in this maximum inequality. With probability $P$ $1-o\left(1\right)$ and for $\hat{v}_{k,a}\in\mathcal{V}_{an}$, it can be shown that 
	\[
	\sup_{a\in\mathcal{A}}\left\Vert \mathbb{G}_{N,k}\xi\left(\boldsymbol{\psi}_{a}\left(W_{a};\hat{v}_{k,a}\right)-\boldsymbol{\psi}_{a}\left(W_{a};v_{a}^{0}\right)\right)\right\Vert \lesssim\sup_{f\in\xi\mathcal{F}_{2}}\left|\mathbb{G}_{N,k}f\right|.
	\]
	Therefore, it follows with  probability $P$ $1-o\left(1\right)$ that
	\[
	\frac{1}{\sqrt{N}}\sup_{a\in\mathcal{A}}\left\Vert \mathbb{G}_{N,k}\xi\left(\boldsymbol{\psi}_{a}\left(W_{a};\hat{v}_{k,a}\right)-\boldsymbol{\psi}_{a}\left(W_{a};v_{a}^{0}\right)\right)\right\Vert \lesssim K^{\frac{1}{2}}n^{-\frac{1}{2}}o_{p}\left(1\right)\lesssim o_{p}\left(n^{-\frac{1}{2}}\right),
	\]
	and since $K$ is fixed and finite, 
	\begin{align*}
		\Pi_{1}=\sup_{a\in\mathcal{A}}\left\Vert \mathbb{G}_{n}\xi\left(\boldsymbol{\psi}_{a}\left(W_{a};\hat{v}_{k,a}\right)-\boldsymbol{\psi}_{a}\left(W_{a};v_{a}^{0}\right)\right)\right\Vert  & \lesssim\sqrt{n}o_{p}\left(n^{-\frac{1}{2}}\right)=o_{p}\left(1\right).
	\end{align*}
	Combining the previous results, we obtain that $\left\Vert Z_{n,P}^{*}-U_{n,P}^{*}\right\Vert =o_{p}\left(1\right)$.
	Next, notice that $U_{n,P}^{*}$ is associated with the class of functions
	$\xi f$, where $f\in\mathcal{F}_{0}$ is defined in Assumption A.1.6.
	As shown in the proof of Theorem A.1, the class of $\mathcal{F}_{0}$
	is Donsker, uniformly in $P\in\mathcal{P}_{n}$ under the imposed
	assumptions. Therefore, we may invoke Theorem B.2 of of \citet{BCFH_2017}
	and obtain $U_{n,P}^{*}\rightsquigarrow_{B}Z_{P}$. Indeed, both $U_{n,P}^{*}$ and $Z_{P}$ are Gaussian processes that share the same (zero) mean and the same covariance matrix. Finally,
	using a similar argument in step 2 for proving Theorem 5.2 in the
	appendix of \citet{BCFH_2017}, we can obtain $Z_{n,P}^{*}\rightsquigarrow_{B}Z_{P}$.
	Let $\text{BL}_{1}\left(l^{\infty}\left(\mathcal{A}\right)\right)$
	be the space of functions mapping the space of functions in $l^{\infty}\left(\mathcal{A}\right)$
	to $\left[0,1\right]$, with the Lipschitz norm being at most 1. Let $E_{B_{n}}$
	denote the expectation over the multiplier weights $\left(\xi_{i}\right)_{i=1}^{n}$ when 
	holding the data $\left(W_{i}\right)_{i=1}^{n}$ fixed. Following
	step 2 for proving Theorem 5.2 in the appendix of \citet{BCFH_2017},
	we obtain the following inequality: 
	\begin{align*}
		\sup_{h\in\text{BL}_{1}\left(l^{\infty}\left(\mathcal{A}\right)\right)}\left|E_{B_{n}}\left[h\left(Z_{n,P}^{*}\right)\right]-E_{P}\left[h\left(Z_{P}\right)\right]\right| & \leq\sup_{h\in\text{BL}_{1}\left(l^{\infty}\left(\mathcal{A}\right)\right)}\left|E_{B_{n}}\left[h\left(U_{n,P}^{*}\right)\right]-E_{P}\left[h\left(Z_{P}\right)\right]\right|\\
		& +E_{B_{n}}\left[\left\Vert Z_{n,P}^{*}-U_{n,P}\right\Vert \land2\right].
	\end{align*}
	The first term vanishes by Theorem B.2 of \citet{BCFH_2017},
	since we have proven that $U_{n,P}^{*}\rightsquigarrow_{B}Z_{P}$.
	The second term is $o_{P}\left(1\right)$, because $E\left[\left\Vert Z_{n,P}^{*}-U_{n,P}\right\Vert \land2\right]=E\left[E_{B_{n}}\left[\left\Vert Z_{n,P}^{*}-U_{n,P}\right\Vert \land2\right]\right]\rightarrow0$
	by the Markov inequality 
	\begin{align*}
		P\left(E_{B_{n}}\left[\left\Vert Z_{n,P}^{*}-U_{n,P}^{*}\right\Vert \land2\right]\geq\varepsilon\right) & \leq\frac{E\left[E_{B_{n}}\left[\left\Vert Z_{n,P}^{*}-U_{n,P}^{*}\right\Vert \land2\right]\right]}{\varepsilon}\\
		& =\frac{E\left[\left\Vert Z_{n,P}^{*}-U_{n,P}^{*}\right\Vert \land2\right]}{\varepsilon}.
	\end{align*}
	As shown above, $\left\Vert Z_{n,P}^{*}-U_{n,P}^{*}\right\Vert =o_{p}\left(1\right)$,
	which implies that $\left\Vert Z_{n,P}^{*}-U_{n,P}^{*}\right\Vert \land2=o_{P}\left(1\right)$.
	Therefore, $\sup_{h\in BL_{1}\left(l^{\infty}\left(\mathcal{A}\right)\right)}\left|E_{B_{n}}\left[h\left(Z_{n,P}^{*}\right)\right]-E_{P}\left[h\left(Z_{P}\right)\right]\right|$
	vanishes and it follows that $Z_{n,P}^{*}\rightsquigarrow_{B}Z_{P}$.
\end{proof}

The proof of Theorem A.3 relies on the \textit{uniform Hadamard differentiability}
\citep{BCFH_2017} of the quantile function. The definition of uniform
Hadamard differentiability is as follows.

\begin{definition}\textbf{Uniform Hadamard Tangential Differentiability,
		\citet{BCFH_2017}}: Let $\mathbb{E}$ and $\mathbb{D}$ be normed
	spaces. Consider a map $\phi:\mathbb{D}_{\phi}\longmapsto\mathbb{E}$,
	where $\mathbb{D}_{\phi}\subset\mathbb{D}$ and the range of $\phi$
	is a subset of $\mathbb{E}$. Let $\mathbb{D}_{0}\subset\mathbb{D}$
	be a normed space, and $\mathbb{D}_{\rho}\subset\mathbb{D}_{\phi}$
	be a compact metric space. Let $h\longmapsto\phi_{\rho}^{\prime}\left(h\right)$
	be the linear derivative map associated with $\phi$, where $h\in\mathbb{D}_{0}$
	and $\rho\in\mathbb{D}_{\rho}$. The linearity of $\phi_{\rho}^{\prime}\left(h\right)$
	holds for each $\rho$. Then the map $\phi:\mathbb{D}_{\phi}\longmapsto\mathbb{E}$
	is called Hadamard differentiable uniformly in $\rho\in\mathbb{D}_{\rho}$
	tangentially to $\mathbb{D}_{0}$ with derivative map $h\longmapsto\phi_{\rho}^{\prime}\left(h\right)$,
	if 
	\begin{align*}
		\left|\frac{\phi\left(\rho_{n}+t_{n}h_{n}\right)-\phi\left(\rho_{n}\right)}{t_{n}}-\phi_{\rho}^{\prime}\left(h\right)\right| & \rightarrow0,\\
		\left|\phi_{\rho_{n}}^{\prime}\left(h_{n}\right)-\phi_{\rho}^{\prime}\left(h\right)\right| & \rightarrow0\text{ as }n\rightarrow0.
	\end{align*}
	for all convergence sequences $\rho_{n}\rightarrow\rho$, $t_{n}\rightarrow0$
	in $\mathbb{R}$ and $h_{n}\rightarrow h$, such that $\rho_{n}+t_{n}h_{n}\in\mathbb{D}_{\phi}$
	for every $n$. \end{definition}

As pointed out by \citet{BCFH_2017}, the quantile function is uniformly
Hadamard-differentiable if we set $\mathbb{D}=l^{\infty}\left(T\right)$,
where $T=\left[\epsilon,1-\epsilon\right],\epsilon>0$, $\mathbb{D}_{\phi}$
is a set of c\`{a}dl\`{a}g functions on $T$, $\mathbb{D}_{0}=\text{UC}\left(T\right)$,
$\mathbb{D}_{\rho}$ being a compact subset of $C^{1}\left(T\right)$
such that each $\rho$ satisfies $\partial\rho\left(u\right)/\partial u>0$.\footnote{$\text{UC}\left(T\right)$ denotes a set of uniformly continuous functions
	from $T$ to $\mathbb{R}$, and $C^{1}\left(T\right)$ denotes a set
	of continuous differentiable functions from $T$ to $\mathbb{R}$.} Notice that this setting rules out the case that $Y\left(d,M\left(d^{\prime}\right)\right)$
is a discrete random variable. Also if $\mathbb{D}_{\rho}=\mathbb{D}_{\phi}$,
the quantile function is not Hadamard-differentiable uniformly in
$\mathbb{D}_{\rho}$ in the sense of our definition. This
is different from the definition of uniformly differentiability given
in \citet{Vaart_1998} which requires $\mathbb{D}_{\rho}=\mathbb{D}_{\phi}$.
Since our estimation is for infinite dimension, it is essential to
restrict $\mathbb{D}_{\rho}$ to be much smaller than $\mathbb{D}_{\phi}$
and endow $\mathbb{D}_{\rho}$ to have a much stronger metric than
the metric induced by the norm of $\mathbb{D}$. However, here the
estimated $\hat{\rho}$ can satisfy $\hat{\rho}\in\mathbb{D}_{\phi}$, but
$\hat{\rho}\notin\mathbb{D}_{\rho}$ (for example when $\hat{\rho}$ is
an empirical c.d.f.), even though the population values of $\rho\in\mathbb{D}_{\rho}$
and $\partial\rho\left(u\right)/\partial u>0$ should hold. With the
definition of uniform Hadamard differentiability, we in a next step restate
Theorems B.3 and B.4 of \citet{BCFH_2017} as follows. \begin{theorem}\textbf{Functional
		delta method uniformly in $P\in\mathcal{P}$, \citet{BCFH_2017}:}
	Let $\phi:\mathbb{D}_{\phi}\subset\mathbb{D}\longmapsto\mathbb{E}$
	be Hadamard differentiable uniformly in $\rho\in\mathbb{D}_{\rho}\subset\mathbb{D}_{\phi}$
	tangentially to $\mathbb{D}_{0}$ with derivative map $h\longmapsto\phi_{\rho}^{\prime}\left(h\right)$.
	Let $\hat{\rho}_{n,P}$ be a sequence of stochastic processes taking
	values in $\mathbb{D}_{\phi}$, where each $\hat{\rho}_{n,P}$ is
	an estimator of the parameter $\rho=\rho_{P}\in\mathbb{D}_{\rho}$.
	Suppose there exists a sequence of constants $r_{n}\rightarrow\infty$
	such that $Z_{n,P}:=r_{n}\left(\hat{\rho}_{n,P}-\rho_{P}\right)\rightsquigarrow Z_{P}$
	in $\mathbb{D}$ uniformly in $P\in\mathcal{P}_{n}$. The limit process
	$Z_{P}$ is separable and takes its values in $\mathbb{D}_{0}$ for
	all $P\in\mathcal{P}=\bigcup_{n\geq n_{0}}\mathcal{P}_{n}$, where
	$n_{0}$ is fixed. Moreover, the set of stochastic processes $\left\{ Z_{P}:P\in\mathcal{P}\right\} $
	is relatively compact in the topology of weak convergence in $\mathbb{D}_{0}$,
	that is, every sequence in this set can be split into weakly convergent
	subsequences. Then, $r_{n}\left(\phi\left(\hat{\rho}_{n,P}\right)-\phi\left(\rho\right)\right)\rightsquigarrow\phi_{\rho_{P}}^{\prime}\left(Z_{P}\right)$
	in $\mathbb{E}$ uniformly in $P\in\mathcal{P}_{n}$. If $\left(\rho,h\right)\longmapsto\phi_{\rho_{P}}^{\prime}\left(h\right)$
	is defined and continuous on the whole of $\mathbb{D}_{\rho}\times\mathbb{D}$,
	then the sequence $r_{n}\left(\phi\left(\hat{\rho}_{n,P}\right)-\phi\left(\rho\right)\right)\rightsquigarrow\phi_{\rho_{P}}^{\prime}\left(r_{n}\left(\hat{\rho}_{n,P}-\rho\right)\right)$
	converges to zero in outer probability uniformly in $P\in\mathcal{P}_{n}$.
	Moreover, the set of stochastic processes $\left\{ \phi_{\rho_{P}}^{\prime}\left(Z_{P}\right):P\in\mathcal{P}\right\} $
	is relatively compact in the topology of weak convergence in $\mathbb{E}$.
\end{theorem}

\begin{theorem}\textbf{Functional delta method uniformly in $P\in\mathcal{P}$
		for the bootstrap and other simulation methods, \citet{BCFH_2017}}: Assume
	that the conditions in Theorem A.4 hold. Let $\hat{\rho}_{n,P}$ and $\hat{\rho}_{n,P}^{*}$
	be maps as previously indicated, taking values in $\mathbb{D}_{\phi}$ such
	that $Z_{n,P}:=r_{n}\left(\hat{\rho}_{n,P}-\rho_{P}\right)\rightsquigarrow Z_{P}$
	and $Z_{n,P}^{*}:=r_{n}\left(\hat{\rho}_{n,P}^{*}-\rho_{P}\right)\rightsquigarrow_{B}Z_{P}$
	in $\mathbb{D}$ uniformly in $P\in\mathcal{P}_{n}$. Then, $r_{n}\left(\phi\left(\hat{\rho}_{n,P}^{*}\right)-\phi\left(\hat{\rho}_{n,P}\right)\right)\rightsquigarrow_{B}\phi_{\rho_{P}}^{\prime}\left(Z_{P}\right)$
	uniformly in $P\in\mathcal{P}_{n}$. \end{theorem}

\begin{proof}[Proof of Theorem A.3] Function $\phi_{\boldsymbol{\theta}}$
	satisfies uniform Hadamard tangential differentiability and both $Z_{n,P}\rightsquigarrow Z_{P}$ and $Z_{n,P}^{*}\rightsquigarrow Z_{P}$
	in $l^{\infty}\left(\mathcal{A}\right)^{4}$ uniformly in $P\in\mathcal{P}_{n}$, as shown
	in Theorems A.1 and A.2. Therefore, the proof follows by applying Theorems A.4 and A.5. \end{proof}


\subsection{Tables and Figures}

\begin{table}[H]
	\centering
	\caption{Descriptive statistics for the empirical application (Job Corps data)}
	\scalebox{0.8}{ %
		\begin{tabular}{ccccccccccc}
			\hline 
			& All  & $D=1$  & $D=0$  & Diff  & p-value  & $M=0$  & $M=1$  & $M=2$  & $M=3$  & $M=4$\tabularnewline
			\hline 
			sample size  & 9,240  & 6,574  & 2,666  & -  & -  & 311  & 3,495  & 4,004  & 1,298  & 141\tabularnewline
			\textbf{Outcome} $Y$  &  &  &  &  &  &  &  &  &  & \tabularnewline
			weekly earnings in third year  & 172.93  & 173.39  & 171.82  & 1.57  & 0.68  & 159.70  & 188.92  & 165.37  & 159.22  & 145.98 \tabularnewline
			\textbf{Treatment} $D$  &  &  &  &  &  &  &  &  &  & \tabularnewline
			training in first year  & 0.71  & 1.00  & 0.00  & -  & -  & 0.42  & 0.74  & 0.72  & 0.70  & 0.70 \tabularnewline
			\textbf{Mediator} $M$  &  &  &  &  &  &  &  &  &  & \tabularnewline
			general health after first year  & 1.72  & 1.74  & 1.70  & 0.04  & 0.03  & 0.00  & 1.00  & 2.00  & 3.00  & 4.00 \tabularnewline
			\textbf{Covariates} $X$  &  &  &  &  &  &  &  &  &  & \tabularnewline
			female  & 0.44  & 0.45  & 0.42  & 0.03  & 0.01  & 0.41  & 0.39  & 0.46  & 0.50  & 0.51 \tabularnewline
			age  & 18.44  & 18.21  & 18.99  & -0.77  & 0.00  & 18.55  & 18.43  & 18.42  & 18.42  & 18.86 \tabularnewline
			white  & 0.26  & 0.25  & 0.30  & -0.05  & 0.00  & 0.33  & 0.24  & 0.28  & 0.28  & 0.24 \tabularnewline
			black  & 0.49  & 0.50  & 0.47  & 0.03  & 0.01  & 0.41  & 0.52  & 0.48  & 0.49  & 0.44 \tabularnewline
			Hispanic  & 0.17  & 0.17  & 0.16  & 0.01  & 0.28  & 0.21  & 0.17  & 0.17  & 0.16  & 0.21 \tabularnewline
			education  & 9.96  & 9.91  & 10.07  & -0.15  & 0.00  & 9.93  & 10.00  & 9.94  & 9.91  & 9.76 \tabularnewline
			education missing  & 0.02  & 0.01  & 0.03  & -0.01  & 0.00  & 0.00  & 0.02  & 0.02  & 0.01  & 0.03 \tabularnewline
			GED degree  & 0.05  & 0.04  & 0.07  & -0.03  & 0.00  & 0.06  & 0.05  & 0.05  & 0.05  & 0.05 \tabularnewline
			high school degree  & 0.20  & 0.17  & 0.25  & -0.08  & 0.00  & 0.13  & 0.21  & 0.20  & 0.17  & 0.16 \tabularnewline
			English mother tongue & 0.85  & 0.84  & 0.86  & -0.02  & 0.02  & 0.83  & 0.85  & 0.84  & 0.88  & 0.80 \tabularnewline
			cohabiting or married  & 0.06  & 0.05  & 0.08  & -0.03  & 0.00  & 0.07  & 0.05  & 0.06  & 0.08  & 0.09 \tabularnewline
			has one or more children  & 0.20  & 0.18  & 0.24  & -0.06  & 0.00  & 0.23  & 0.20  & 0.19  & 0.21  & 0.21 \tabularnewline
			ever worked before JC  & 0.14  & 0.15  & 0.13  & 0.02  & 0.02  & 0.15  & 0.15  & 0.14  & 0.15  & 0.15 \tabularnewline
			average weekly gross earnings  & 19.42  & 18.24  & 22.32  & -4.08  & 0.05  & 37.52  & 19.91  & 17.65  & 16.97  & 39.69 \tabularnewline
			household size  & 3.43  & 3.47  & 3.34  & 0.13  & 0.01  & 3.30  & 3.40  & 3.47  & 3.45  & 3.37 \tabularnewline
			household size missing  & 0.02  & 0.01  & 0.03  & -0.02  & 0.00  & 0.01  & 0.02  & 0.02  & 0.01  & 0.03 \tabularnewline
			mum's education  & 9.31  & 9.45  & 8.98  & 0.47  & 0.00  & 8.61  & 9.35  & 9.41  & 9.24  & 8.00 \tabularnewline
			mum's education missing  & 0.19  & 0.18  & 0.21  & -0.03  & 0.00  & 0.24  & 0.19  & 0.18  & 0.20  & 0.28 \tabularnewline
			dad's education  & 7.04  & 7.16  & 6.73  & 0.44  & 0.00  & 6.45  & 7.14  & 7.10  & 6.74  & 6.66 \tabularnewline
			dad's education missing  & 0.39  & 0.38  & 0.41  & -0.02  & 0.03  & 0.44  & 0.39  & 0.38  & 0.40  & 0.44 \tabularnewline
			received welfare as child  & 1.92  & 1.93  & 1.89  & 0.04  & 0.23  & 1.86  & 1.89  & 1.92  & 1.98  & 1.97 \tabularnewline
			welfare info missing  & 0.07  & 0.07  & 0.08  & -0.01  & 0.06  & 0.10  & 0.07  & 0.07  & 0.07  & 0.07 \tabularnewline
			general health at baseline  & 1.65  & 1.65  & 1.65  & -0.01  & 0.65  & 1.69  & 1.40  & 1.73  & 2.00  & 2.12 \tabularnewline
			health at baseline missing  & 0.02  & 0.01  & 0.03  & -0.01  & 0.00  & 0.00  & 0.02  & 0.02  & 0.01  & 0.03 \tabularnewline
			smoker  & 0.81  & 0.81  & 0.80  & 0.01  & 0.63  & 0.86  & 0.75  & 0.82  & 0.88  & 0.86 \tabularnewline
			smoker info missing  & 0.48  & 0.49  & 0.46  & 0.03  & 0.03  & 0.42  & 0.52  & 0.47  & 0.41  & 0.43 \tabularnewline
			alcohol consumption & 1.79  & 1.78  & 1.84  & -0.06  & 0.15  & 1.96  & 1.67  & 1.85  & 1.90  & 1.88 \tabularnewline
			alcohol info missing  & 0.43  & 0.43  & 0.41  & 0.03  & 0.03  & 0.38  & 0.46  & 0.41  & 0.39  & 0.38 \tabularnewline
			\hline 
	\end{tabular}} \label{table3} 
\end{table}

\clearpage

\begin{sidewaysfigure}[ht]
	\centering
	\mbox{\subfigure{\includegraphics[height=5cm, width = 5cm]{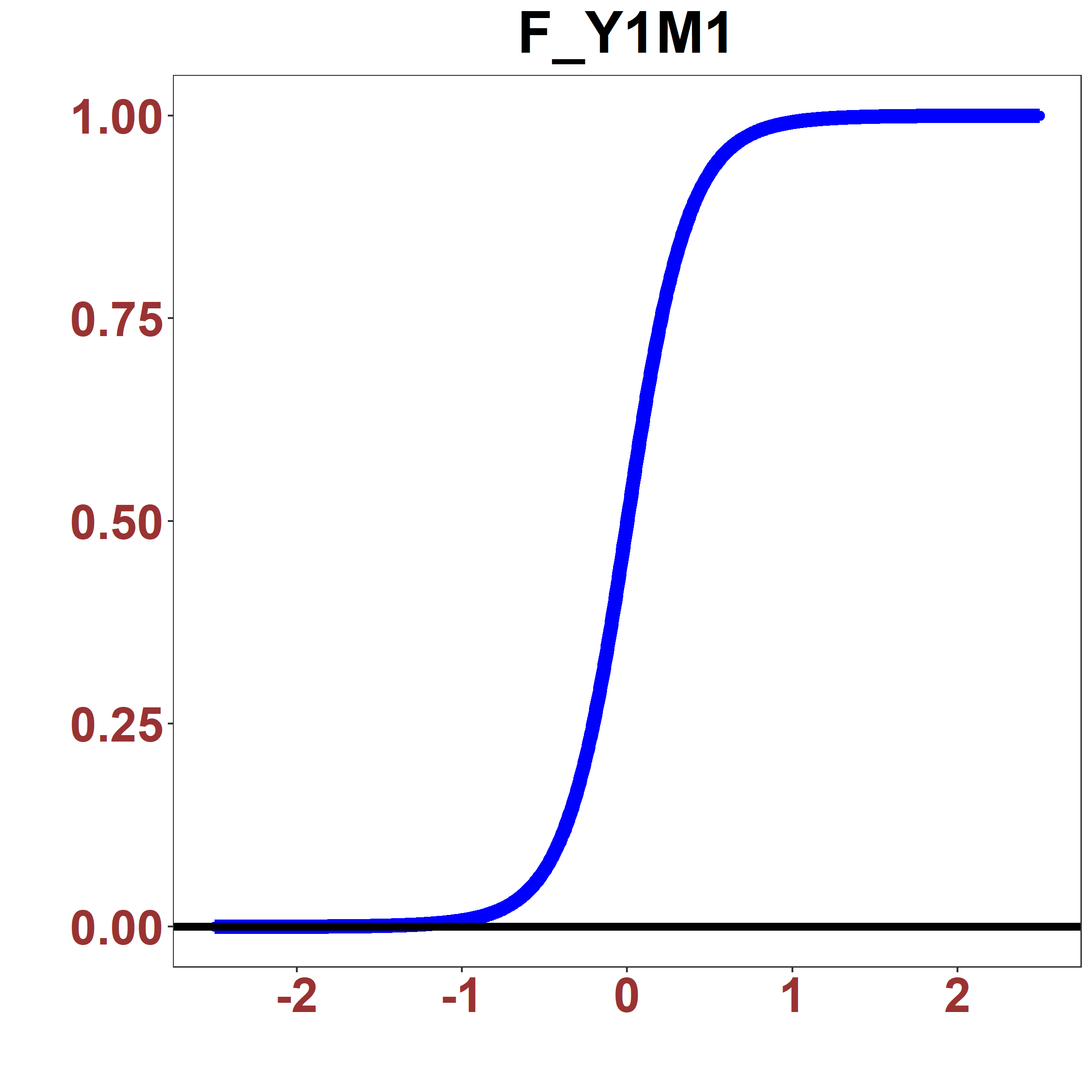}}		
		\subfigure{\includegraphics[height=5cm, width = 5cm]{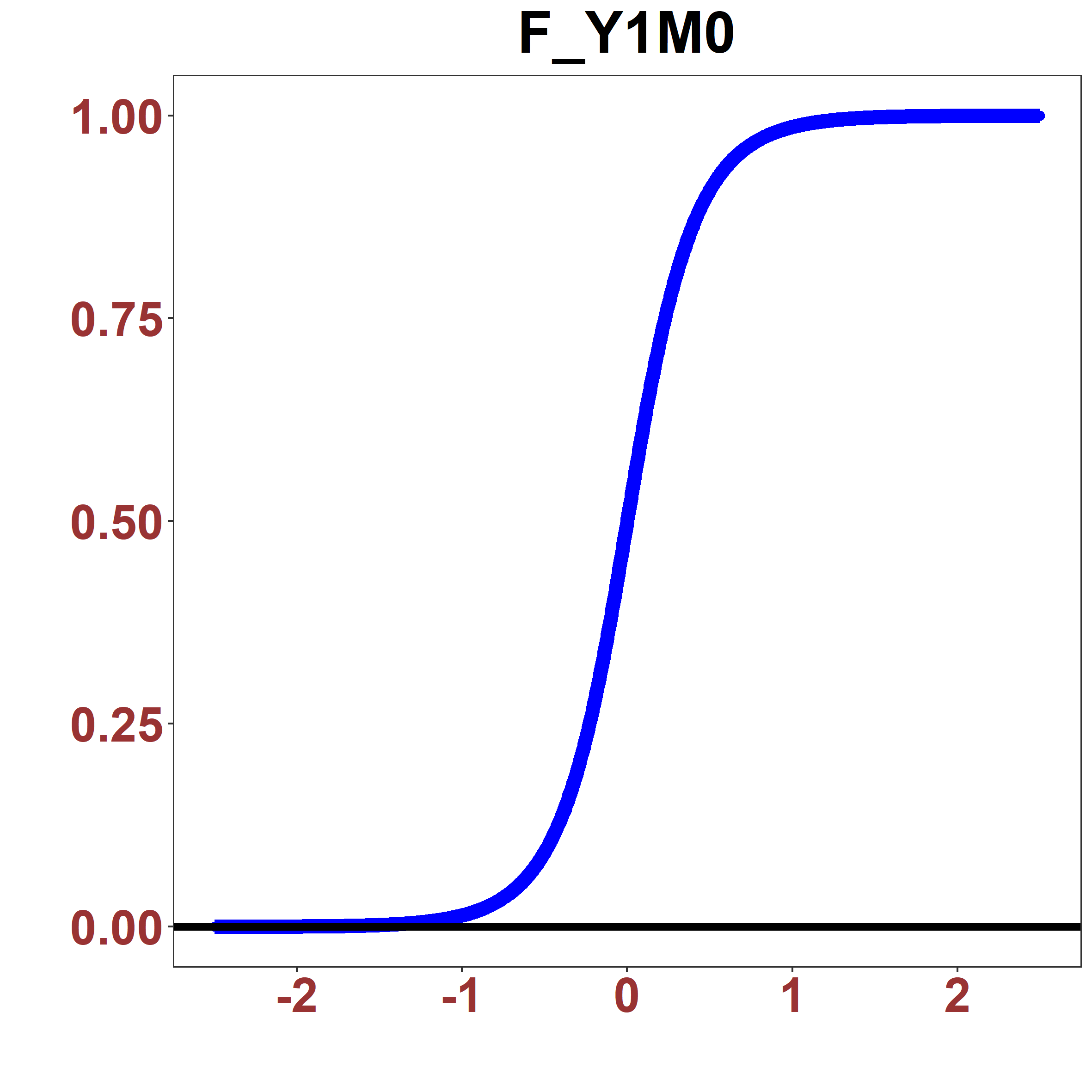}}
		\subfigure{\includegraphics[height=5cm, width = 5cm]{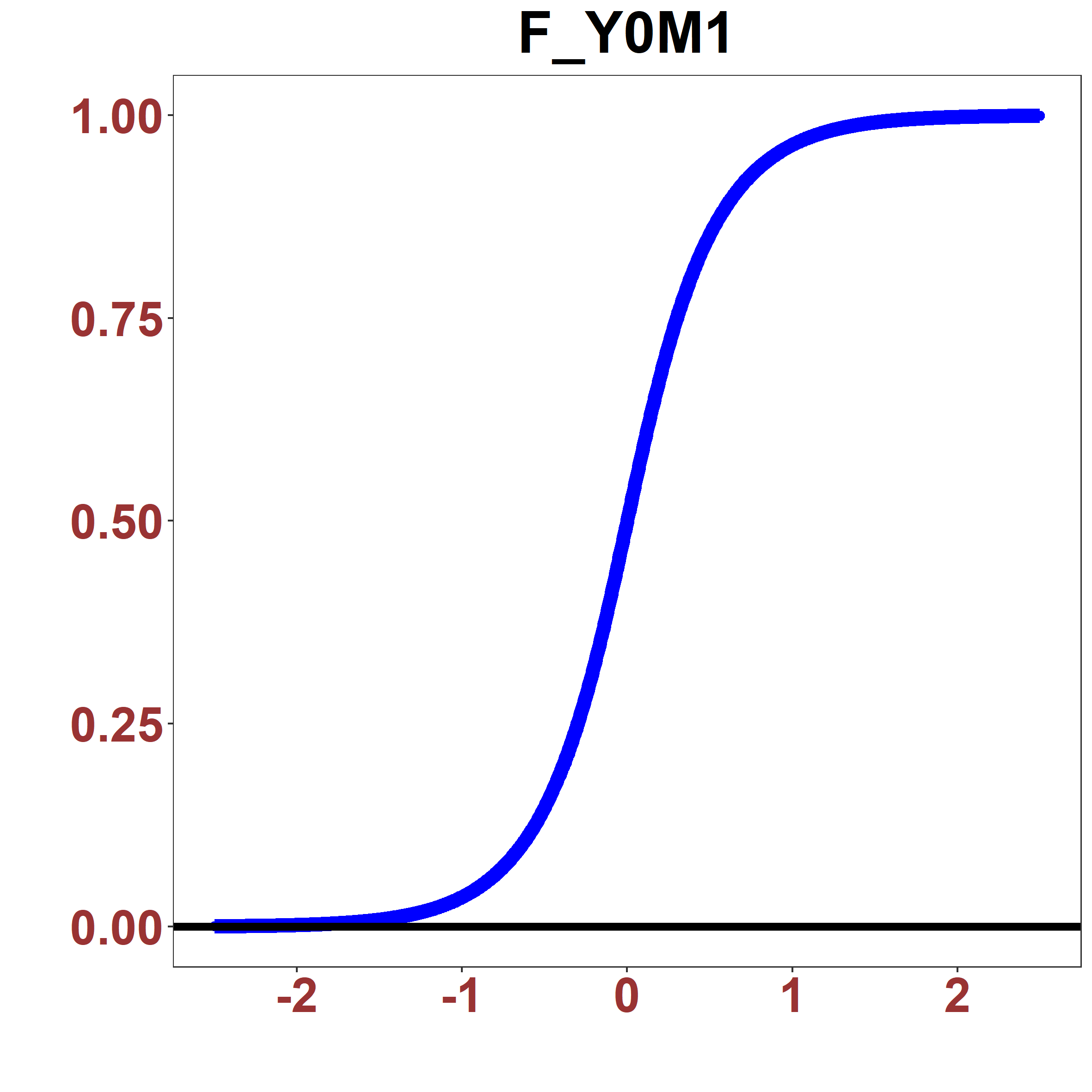}}
		\subfigure{\includegraphics[height=5cm, width = 5cm]{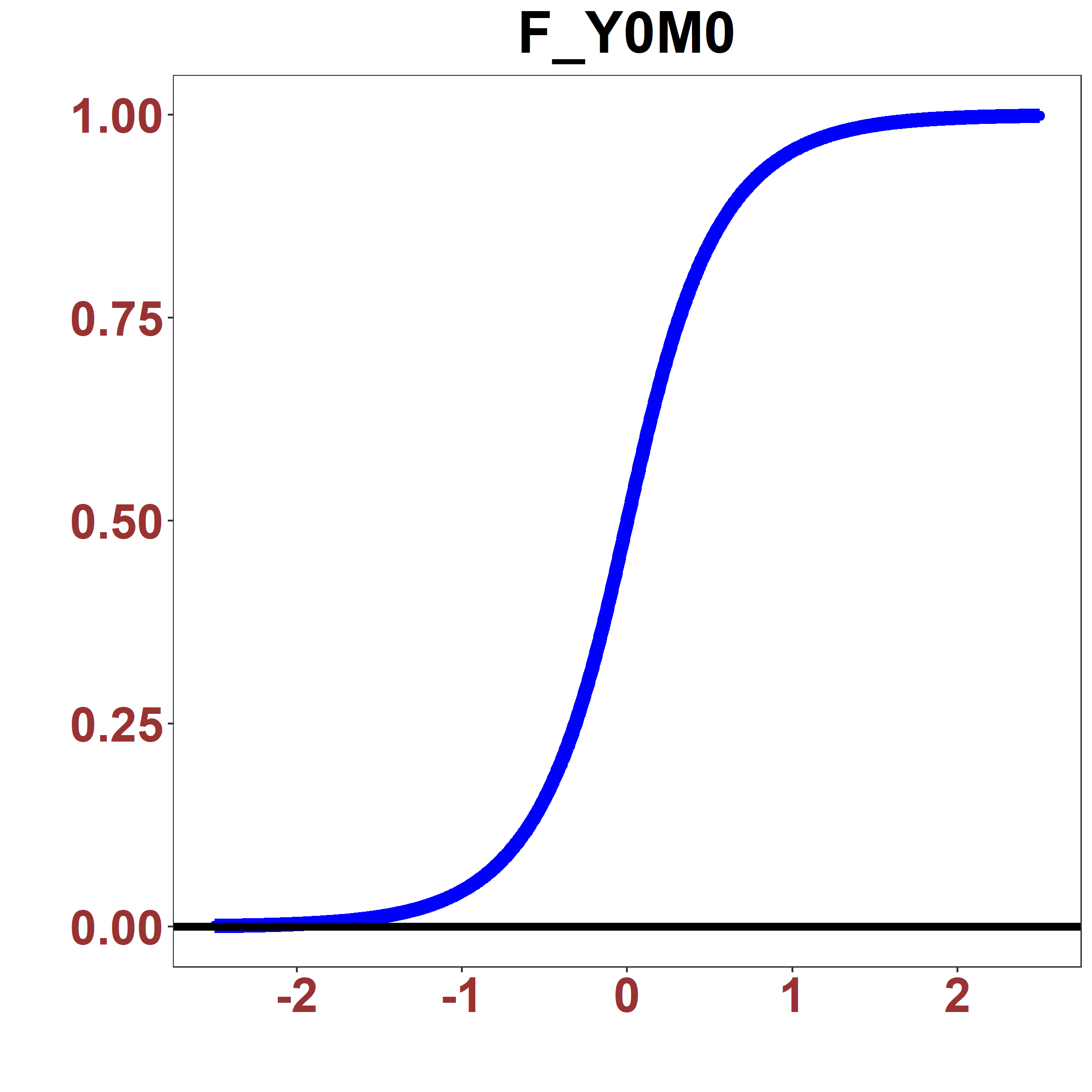}}
	}
	\mbox{\subfigure{\includegraphics[height=5cm, width = 5cm]{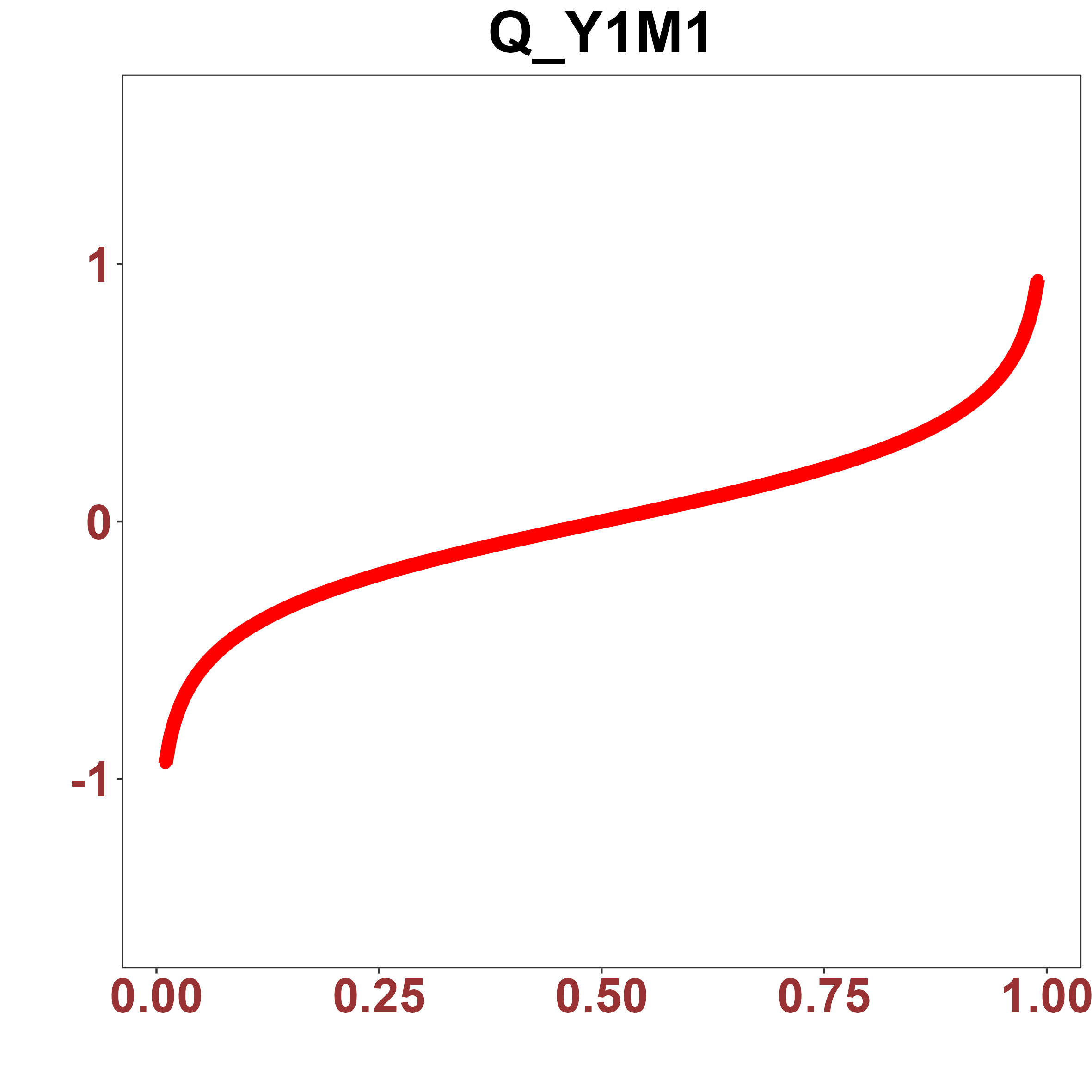}}		
		\subfigure{\includegraphics[height=5cm, width = 5cm]{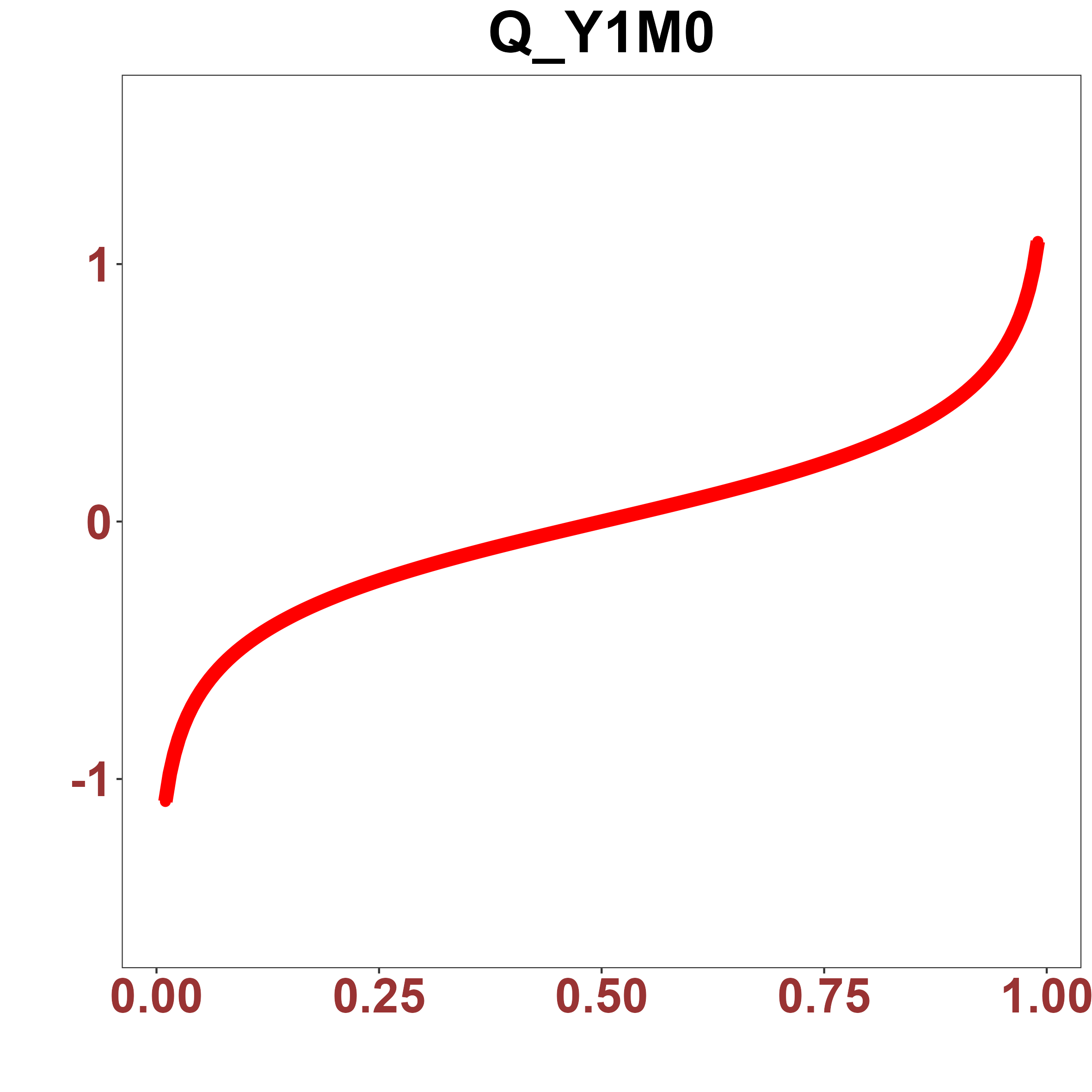}}
		\subfigure{\includegraphics[height=5cm, width = 5cm]{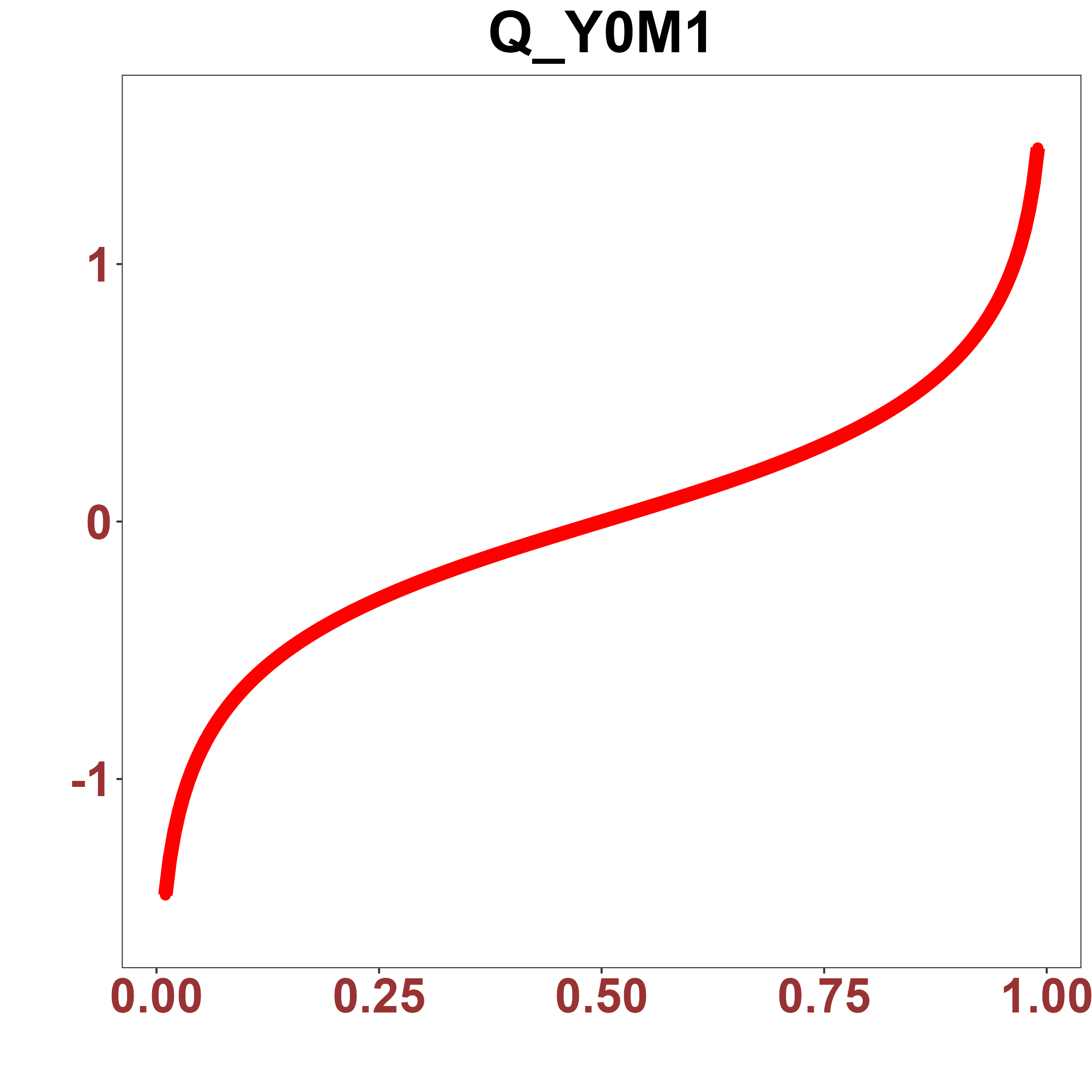}}
		\subfigure{\includegraphics[height=5cm, width = 5cm]{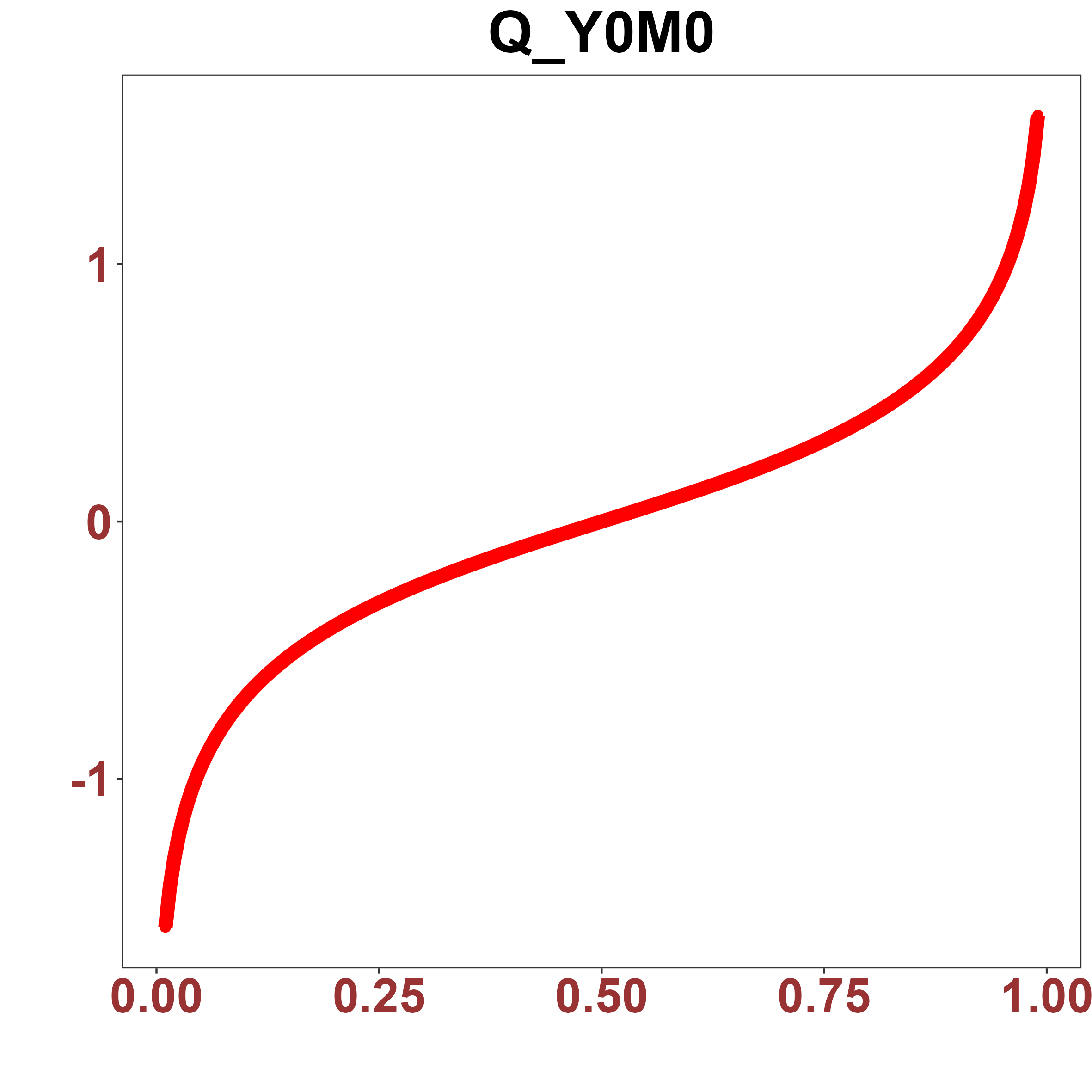}}
	}
	\caption{Approximate true c.d.f.\ profiles from 40 million Monte Carlo simulations under the data generating process described in Section 4. Top: $F_{Y\left(1,M\left(1\right)\right)}$, $F_{Y\left(1,M\left(0\right)\right)}$, $F_{Y\left(0,M\left(1\right)\right)}$ and $F_{Y\left(0,M\left(0\right)\right)}$; Bottom: $Q_{Y\left(1,M\left(1\right)\right)}$, $Q_{Y\left(1,M\left(0\right)\right)}$, $Q_{Y\left(0,M\left(1\right)\right)}$ and $Q_{Y\left(0,M\left(0\right)\right)}$. The data generating process is described in Section 4.}
	\label{figure3}
\end{sidewaysfigure}

\begin{figure}[ht]
	\centering
	
	\mbox{\subfigure{\includegraphics[height=5cm, width = 5cm]{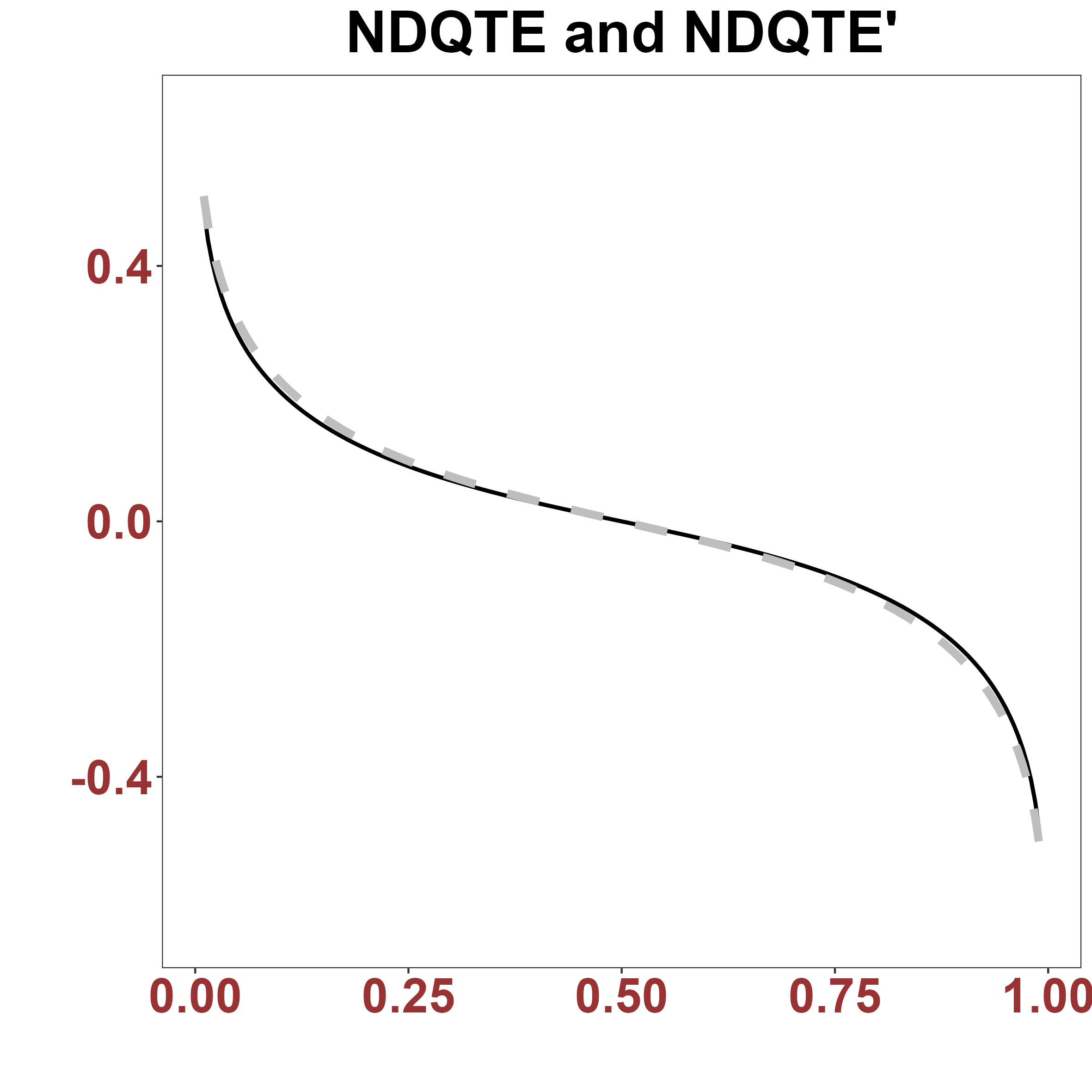}}		
		\subfigure{\includegraphics[height=5cm, width = 5cm]{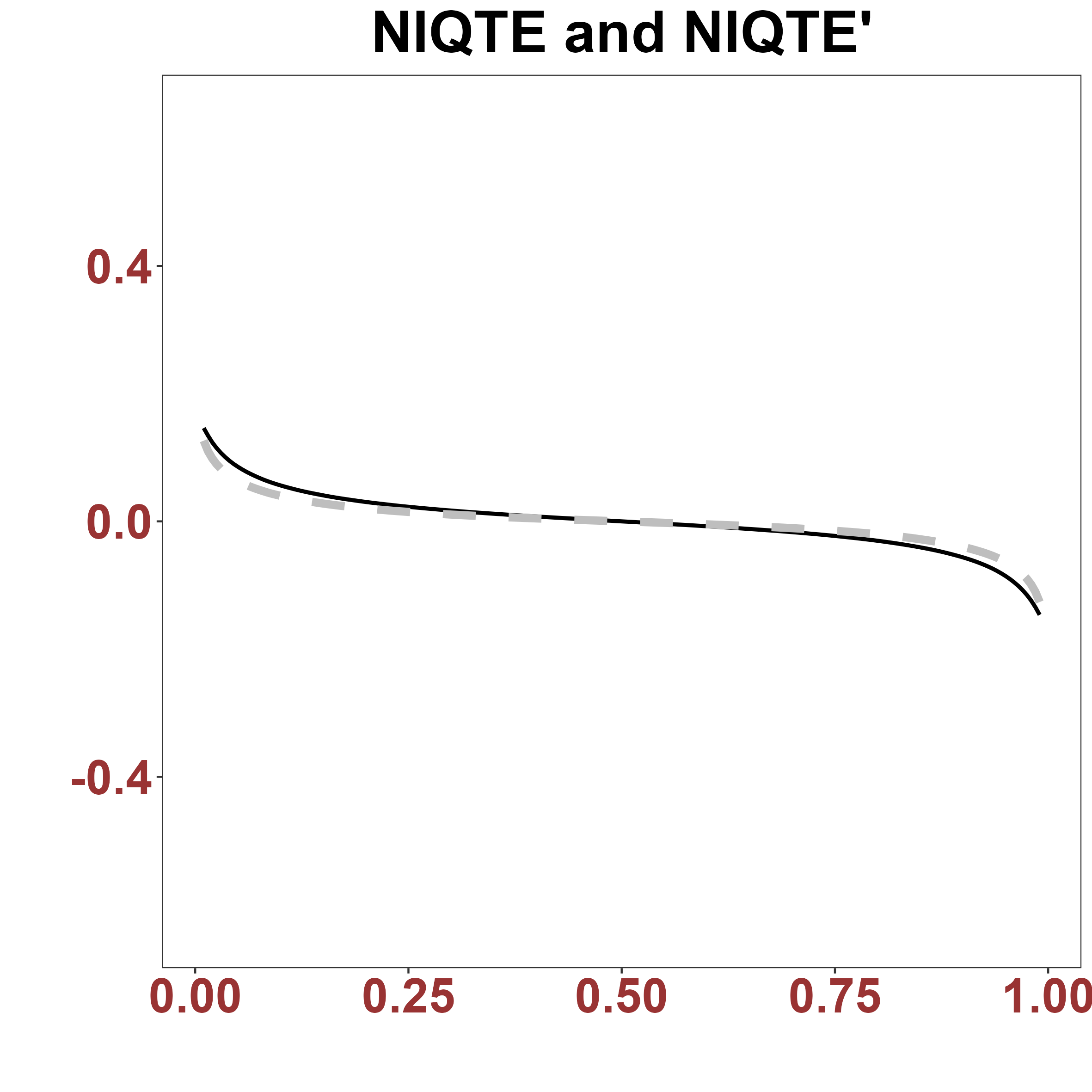}}
		\subfigure{\includegraphics[height=5cm, width = 5cm]{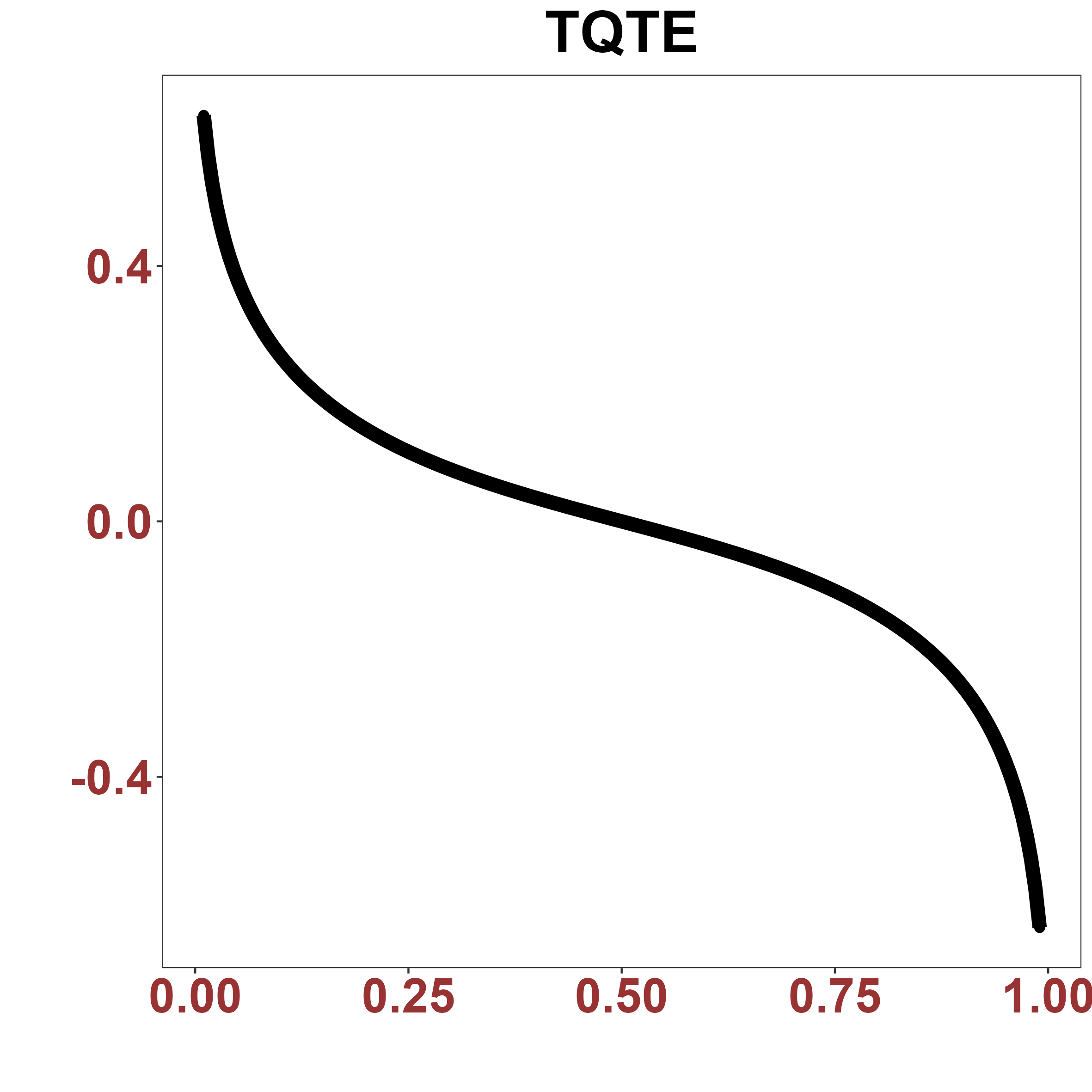}}
	}
	\caption{Approximate true effect profiles from 40 million Monte Carlo simulations under the data generating process described in Section 4: NDQTE (black solid) and NDQTE' (gray dashed), NIQTE (black solid) and NIQTE' (gray dashed) and TQTE. }
	\label{figure4}
\end{figure}

\clearpage
\bibliographystyle{ECTA}
\bibliography{ref_mediation}
\end{document}